\documentclass{jfm}
\usepackage{graphicx}
\usepackage{epstopdf, epsfig}

\newtheorem{theorem}{Theorem}

\shorttitle{Hysteresis and non-Newtonian rheology in a sheared gas-solid suspension}
\shortauthor{S. Saha and M. Alam}

\title{Revisiting ignited-quenched transition and the non-Newtonian rheology of a sheared dilute gas-solid suspension}

\author{Saikat Saha\aff{1}
 \and Meheboob Alam \aff{1}\corresp{\email{meheboob@jncasr.ac.in}}}

\affiliation{\aff{1}Engineering Mechanics Unit, Jawaharlal Nehru Centre for Advanced Scientific Research,
Jakkur P.O., Bangalore 560064, India}

\begin{document}

\maketitle

\begin{abstract}
 The hydrodynamics and rheology of  a  sheared dilute gas-solid suspension, consisting of  inelastic hard-spheres  suspended in a gas,
 are analysed using  anisotropic Maxwellian as the single particle distribution function.
The closed-form solutions for    granular temperature ($T$) and three invariants of the second-moment tensor are obtained
  as functions of the Stokes number ($St$), the mean density ($\nu$)   and the restitution coefficient ($e$).
 Multiple states of high and low temperatures are found when the Stokes number is small,
 thus recovering the  ``ignited'' and ``quenched'' states, respectively, of Tsao \& Koch (J. Fluid Mech.,1995, vol. 296, pp. 211-246).
 The phase diagram is constructed  in the three-dimensional ($\nu, St, e$)-space  
that delineates the regions of ignited and quenched states and their coexistence.
Analytical expressions for the particle-phase shear viscosity and the normal stress differences are obtained,
along with related scaling relations on the quenched and ignited states.
 At any $e$, the shear-viscosity  undergoes a discontinuous jump with increasing shear rate (i.e.~ discontinuous shear-thickening)   at the ``quenched-ignited'' transition.
  The first (${\mathcal N}_1$) and second (${\mathcal N}_2$) normal-stress differences also undergo similar first-order transitions: 
 (i) ${\mathcal N}_1$ jumps from large to small positive values  and (ii) ${\mathcal N}_2$ from positive to negative values with increasing $St$,
 with the sign-change of ${\mathcal N}_2$ identified with the system making a transition from the quenched to ignited states.
 The superior prediction of the present theory over the standard Grad's method and the Chapman-Enskog solution
 is demonstrated via comparisons of transport coefficients with simulation data  for a range of Stokes number and restitution coefficient.
 \end{abstract}

\section{Introduction}
\label{sec:introduction}

During the last few decades, a lot of research has been done to understand  the behaviours of  rapid granular flows  \citep{SJ1981,LSJC1984,JR1985,Campbell1990,SG1998,BDKS1998,Goldhirsch2003,RN2008,FP2008}, 
a collection macroscopic inelastic (the restitution coefficient $e<1$) hard-particles for which the effect of the interstitial fluid is neglected,
and the tools from dense-gas kinetic theory have been successfully employed to understand its hydrodynamics and rheology. 
The closely related research-area  of  gas-solid suspensions~\citep{DH1963,AJ1968,Buyevich1971,Gidaspow1994,Jackson2000,GM2011},
in which the viscous drag due to interstitial fluid and other related hydrodynamic effects must be incorporated,
 has also been extensively studied over the last century due to its importance in fluidized-bed and FCC reactors ~\citep{DH1963,Gidaspow1994} 
 encountered in chemical and process industries.
For continuum models of  gas-solid suspensions, the kinetic-theory-based rheological  models have been suggested 
by considering elastically colliding particles~\citep{Koch1990,TK1995}
as well as for inelastic particles~\citep{LMJ1991,Sangani1996,LS2003} interacting in a bath of a Newtonian gas.

For the present problem of a sheared gas-solid suspension of inelastic particles, 
the energy input due to shear is compensated by two mechanisms, (i) inelastic inter-particle collisions, characterized by a coefficient of normal restitution ($e$)
and (ii) the  drag force which the surrounding fluid exerts on the particles. 
The volume fraction of the suspended particles (of diameter $\sigma$ and mass $m$) is assumed to be small, i.e.~ $\nu=\pi \sigma^3 n/6\ll 1$, representing a `dilute' suspension,
 along with the conditions of (ii) small Reynolds number $\Rey=\rho_g{\dot\gamma}\sigma^2/\mu_g\ll 1$
 (where $\rho_g$ and $\mu_g$ are the  gas density and its viscosity, respectively, and ${\dot\gamma}$ is the imposed shear rate on the suspension)
  and (iii) finite Stokes number 
\begin{equation}
  St={\dot\gamma}\tau_v,
  \quad
  \mbox{with}
  \quad
   \tau_v=m/(3\pi\mu_g \sigma) 
\end{equation}
being  the viscous relaxation time which  is a measure of the time a typical particle takes to relax back to the local fluid velocity. 
The  limit of $St\to\infty$ represents the `dry' granular gas~\citep{Campbell1990,Goldhirsch2003}.
Under the above assumptions, \cite{TK1995} analysed the hydrodynamics and the non-Newtonian rheology of a dilute suspension of elastic ($e=1$) hard-particles
employing the Grad's moment-expansion method (i.e. an expansion in terms of Hermite polynomials around a Maxwellian, ~\cite{Grad1949}). 
They discovered  two qualitatively different states, dubbed (i) ``quenched''  (low temperature) and (ii) ``ignited''  (high temperature) states, 
corresponding to the time intervals (i) $\tau_c\gg \tau_v \gg {\dot\gamma}^{-1}$ and (ii) $\tau_c \ll {\dot\gamma}^{-1} \ll \tau_v$, respectively,
where  $\tau_c$ is the collision time (i.e, the average time between two successive collisions).
They analytically determined two critical Stokes numbers $St_{c_1}$ and $St_{c_2}$ (with $St_{c_2}>St_{c_1}$), below and above which
 the flow remains in the quenched and ignited states, respectively.
 They also determined the shear viscosity and the first and second normal-stress differences, and compared their theory with DSMC (direct simulation Monte Carlo) data.

\cite{Sangani1996} extended the work of \cite{TK1995} 
to  (i) a `dense' gas-solid suspension of elastic ($e=1$) particles as well as to  (ii) a `dilute' suspension of inelastic ($e< 1$) particles.
The same Grad moment-expansion was used  to derive constitutive relations from the underlying Enskog-Boltzmann equation;
 but their analysis  is deficient in the sense that they found zero value for the second normal stress difference
as they did not incorporate certain non-linear terms (see \S5 in this work).
They briefly discussed about the lower limit of Stokes number $St_{c_1}$,
but a thorough analysis of the ``ignited-quenched'' transitions, identifying the regions  for the existence of different states,
 in terms of Stokes number ($St$), particle volume fraction ($\nu$) and the coefficient of restitution ($e$) has not been worked out till date.
The latter effect of the restitution coefficient is  important for dissipative particles
which forms one motivation of the present work.

In the current decade, \cite{PS2012} analysed a sheared gas-solid suspension  by considering a distribution function that sandwiches both the ignited and quenched states --
the resulting rheological fields are reasonably well-predicted over a range of density and Stokes number,
although quantitative mis-match with simulation data exists that increase with increasing  dissipation (i.e. at smaller  $e$).
A Navier-Stokes-order continuum model  has been developed by ~\cite{GTSH2012}  for a moderately-dense gas-solid suspension following dense-gas kinetic theory.
They  solved the underlying Enskog-Boltzmann equation using a Chapman-Enskog-like expansion around
a time-dependent homogeneous cooling state for a gas-solid suspension, and the particle motion has been modelled via a Langevin-type stochastic model with Stokesian drag.
The resulting transport coefficients for the particle-phase are found to have explicit dependence on the gas-phase parameters.
However, the prediction of the latter  model for the shear viscosity of a suspension indicates large discrepancies with simulation data in the dilute limit
of low-$St$ suspension, presumably due to the presence of order-one values of normal stress differences and other non-Newtonian effects.
A related work to uncover  the non-Newtonian rheology of a `dilute' gas-solid suspension has been done recently by ~\cite{CRG2015}.
They  followed the standard Grad's method to analyse the ignited state of a gas-solid suspension,
and the related predictions on the granular temperature and the non-Newtonian stress tensor are found to be quantitatively similar to the earlier work of \cite{TK1995};
for example, the suspension viscosity is over-predicted by the Grad's moment-theory at smaller values of $e$, although the discrepancy decreases with increasing Stokes number.
Collectively, the above literature review  points toward the need to go beyond the well-studied Newtonian rheology (of Navier-Stokes-order)
for both dry granular and gas-solid suspensions.

In this paper, we  revisit and extend  the work of \cite{TK1995} by considering  a dilute system of inelastic ($e\leq 1$) particles suspended in a bath of a Newtonian gas,
and interacting via (i) a Stokeian drag force and (ii) hard-core  inelastic collisions. 
Our work  differs from all previous works on gas-solid suspensions as we adopt  the anisotropic Maxwellian distribution function~\citep{GT1978,JR1988,Richman1989} 
to analyse the underlying Boltzmann equation under homogeneous shearing conditions.
The latter assumption is motivated from our recent work \citep{SA2014,SA2016,AS2017} on `dry' ($St\to\infty$) sheared granular fluid
which established that the  transport coefficients for highly inelastic system ($e\ll 1$) of a  sheared granular fluid (both dilute and dense)
can be accurately predicted by the anisotropic Maxwellian [in comparison to (i) the  standard Grad's moment
expansion (in terms of a truncated Hermite series around a Maxwellian) as well as (ii) the Burnett-order solutions obtained from Chapman-Enskog expansion].
Here we demonstrate the superiority of the former for the case of  a sheared gas-solid suspension via a one-to-one comparison 
of two theories with simulation data.   Another  focus of the present work is to analyse and quantify  the anisotropy of 
the second-moment, ${\mathsfb M}=\langle {\boldsymbol C}{\boldsymbol C} \rangle$, of fluctuation/peculiar  velocity,
 and subsequently  tie and explain the rheological/transport coefficients of  a sheared  gas-solid suspension in terms of the anisotropies of ${\mathsfb M}$.
 The underlying analysis utilizes the  geometric structure of  the eigen-basis of both the shear tensor and the second-moment tensor;
 this  provides geometric insight into the origin of normal stress differences as found for the case of a sheared granular fluid~\citep{SA2016}.
 It must be noted that the analysis of stress anisotropy in this form was  initiated in a seminal work by \cite{GT1978} 
 and subsequently by others  \citep{AT1986,Araki1988,Shukhman1984,JR1988,Richman1989} and the present effort is a continuation of the same legacy
 to the case of a sheared gas-solid suspension.

This paper is organized as follows.  A brief account of the problem and the governing equations for the gas and particle phases 
are given in \S \ref{sec:Problem_description}. 
The anisotropic-Maxwellian distribution function is introduced in \S2.1 which is employed to analyse the ``ignited'' state of
sheared gas-solid suspension; the second moment tensor for the uniform shear flow  is constructed in \S2.1.1 in terms of its eigen-basis.  
The source term of the second moment balance equation is calculated in \S2.1.2 and \S2.2 for the ignited and quenched states, respectively.
The  second-moment balance combining both ignited ($I$)  and quenched ($Q$) states is analysed in \S2.3.
The multi-stability and hysteresis transitions in granular temperature are analysed in detail in \S3, 
along with (i) the validation and superiority of the present analysis in \S3.1,
(ii) analytical solutions for temperatures in three states in \S3.2 and (iii) the critical Stokes numbers for ``$I\leftrightarrow Q$'' transitions in \S3.3.
The non-Newtonian rheology (shear-thickening, normal stress differences) is analysed in \S4.2 and \S4.3,
 in terms of the anisotropies of the second-moment tensor (\S4.1).
 The relative merits of the present theory over the standard Grad's moment-expansion  and Chapman-Enskog expansion
 are analysed  in \S5 via comparisons with available simulation data.
 The conclusions are given in \S6. The mathematical details of various analyses are relegated to Appendices A to F.

\section{Problem description and the kinetic-theory analysis}
\label{sec:Problem_description}

We examine the uniform shear flow of a dilute gas-solid suspension in the absence of gravity,
with a collection of smooth inelastic spheres of mass $m$ and diameter $\sigma$ being suspended in a gas;  
with $x$, $y$ and $z$ pointing the velocity, gradient and vorticity directions (see figure~\ref{fig:schematic}), respectively, 
the velocity field for the suspension is given by
\begin{equation}
   {\boldsymbol u} \equiv (u_x,u_y,u_z) = (\dot{\gamma}y, 0, 0),
 \label{eqn:homogeneous_shear_eqn}
\end{equation}
where $\dot{\gamma}$ is the overall shear rate. 
 We are interested in a steady state suspension where the fluid inertia is very small  but the particle inertia remains finite. 
 Under the assumptions of the smallness of particle Reynolds number,
 the gas-phase obeys    the  Stokes equations of motion 
 \begin{equation}
  \mu_g \nabla^2v_i=\nabla_i p_g,\qquad \nabla_i v_i=0,
  \label{eqn:Stokes_equations}
 \end{equation}
 where $\mu_g$ is the shear viscosity of the gas. The velocity profile  (\ref{eqn:homogeneous_shear_eqn})  satisfies (\ref{eqn:Stokes_equations}).

For the particle-phase, we adopt the kinetic theory of granular gases~\citep{CC1970,JR1985,SG1998,BDKS1998,BP2004}.
Any physical quantity at the macroscopic level is defined as the ensemble averaged value of the same at the particle level, using the single particle distribution $f({\boldsymbol c}, {\boldsymbol x}, t)$ function
 \begin{equation}
   \langle \psi({\boldsymbol c}) \rangle\equiv\frac{1}{n}\int \psi{\boldsymbol c} f({\boldsymbol c}, {\boldsymbol x}, t) d{\boldsymbol c},
  \label{eqn:mean_valie_defn}
 \end{equation}
 with $\psi({\boldsymbol c})$ being any particle-level quantity.
 Here $n\equiv n({\boldsymbol x}, t)$ denotes the number density and $\rho({\boldsymbol x}, t)=mn\equiv \rho_p\nu$ is the mass-density of the particle-phase,
 with $\nu$ being the volume fraction of particles and $\rho_p=m/(\pi\sigma^3/6)$ is its intrinsic/material density.
The macroscopic/hydrodynamic velocity ${\boldsymbol u}=\langle{{\boldsymbol c}}\rangle$, the granular temperature $T=\langle{{\boldsymbol C}^2/3}\rangle$ and 
the particle-phase stress tensor ${\boldsymbol P}=\langle{m{\boldsymbol C \boldsymbol C}}\rangle$ are obtained by substituting $\psi={\boldsymbol c},\;\frac{1}{3}{\boldsymbol C}^2\;{\rm and}\;m{\boldsymbol CC}$, respectively, in (\ref{eqn:mean_valie_defn}), where ${\boldsymbol C}={\boldsymbol c}-{\boldsymbol u}$, is the peculiar velocity.

 For a dilute suspension ($\nu \ll 1$),
 the evolution of the single particle distribution function $(f({\boldsymbol c}, {\boldsymbol x}, t))$ follows the celebrated Boltzmann equation \citep{CC1970}
\begin{equation}
\left(\frac{\partial}{\partial t}+ {\boldsymbol c}\cdot\nabla\right)f+ \nabla_{\boldsymbol c}\cdot\left(f \frac{{\rm d}{\boldsymbol c}}{{\rm d}t}\right) = \Big(\frac{\partial f}{\partial t}\Big)_{coll},
\label{eqn:Boltzmann1}
\end{equation}
where $\nabla_{\boldsymbol c}$ is divergence operator in the velocity space;
  the acceleration of the particles is assumed to follow the Stokes's linear drag law:
\begin{equation}
   \frac{{\rm d}{\boldsymbol c}}{{\rm d}t} = - \frac{{\boldsymbol c} - {\boldsymbol v}}{\tau_v},
   \label{eqn:stokesDrag1}
\end{equation}
 with $\tau_v=m/(3\pi\mu_g\sigma)$ being the viscous relaxation time of the particles.
 Equation (\ref{eqn:stokesDrag1}) holds  if the particle Reynolds number  and the density-ratio ($\rho_f/\rho_p$) are  very small; for large Reynolds numbers, a nonlinear form of the drag-law would be necessary.
The hydrodynamic interactions have been neglected throughout the present analysis -- the particles are assumed to follow the fluid velocity, i.e., there is no 
slip ($ {\boldsymbol v}= {\boldsymbol u}$).
 These additional effects and a complete analysis of the particle-phase rheology in the dense limit (based on Enskog equation) will be considered in a future work.

For the present problem of the steady homogeneous shear flow, the mass-density $\rho$, the velocity gradient ${\boldsymbol \nabla u}\propto {\dot\gamma}$ 
and the stress tensor ${\mathsfb P}$ are constants and the heat flux vanishes. 
In this case the balance equations for mass and linear momentum are  identically satisfied
and the balance of the second moment of fluctuation velocity, ${\mathsfb M}=\langle{{\boldsymbol C \boldsymbol C}}\rangle$, reduces to
\begin{equation}
 {\mathsfb P}\cdot\nabla{\boldsymbol u}+
    \left({\mathsfb P}\cdot\nabla{\boldsymbol u}\right)^T +\frac{2{\dot\gamma}}{St} {\mathsfb P} = {\boldsymbol\aleph},
\label{eqn:balance_second_moment_usf}
\end{equation}
where $St=\dot\gamma\tau_v$ is the Stokes number and ${\boldsymbol\aleph}$ is the source (collisional production) of second moment, given by~\citep{JR1985,SA2014}
\begin{eqnarray}
 {\boldsymbol\aleph} &=& \int m{\boldsymbol C}{\boldsymbol C} \Big(\frac{\partial f}{\partial t}\Big)_{coll} d{\boldsymbol C}
                    =\frac{\sigma^2}{2}\int \Delta\Big(m{\boldsymbol C}{\boldsymbol C}\Big)f({\boldsymbol C}_1)f({\boldsymbol C}_2)  d{\boldsymbol C}_1d{\boldsymbol C}_2,
 \label{eqn:source_second_moment_usf}                    
\end{eqnarray}
with
\begin{equation}
   \Delta\Big(m{\boldsymbol C}{\boldsymbol C}\Big) =
   -\frac{m}{2}(1+e)({\boldsymbol g}\cdot{\boldsymbol  k})\left[ (1-e)({\boldsymbol g}\cdot{\boldsymbol k}){\boldsymbol k}{\boldsymbol k}
      + ({\boldsymbol j}{\boldsymbol k} + {\boldsymbol k}{\boldsymbol j}){\boldsymbol g}{\boldsymbol j} \right],
\end{equation}
where ${\boldsymbol g}={\boldsymbol c}_1-{\boldsymbol c}_2$ is the relative velocity between two colliding particles $1$ and $2$;
${\boldsymbol k}\equiv{\boldsymbol k}_{12}=({\boldsymbol x}_1-{\boldsymbol x}_2)/|{\boldsymbol x}_1-{\boldsymbol x}_2|$ is the unit contact vector 
joining the center of particle-$1$ to that particle-$2$, and ${\boldsymbol j}$ is its normal.

With an appropriate choice of the distribution function $f({\boldsymbol c}, {\boldsymbol x}, t)$, the collision integral (\ref{eqn:source_second_moment_usf}) can be evaluated,
which will be plugged into (\ref{eqn:balance_second_moment_usf}) to carry out the analysis for the particle-phase rheology and hydrodynamics
of a sheared gas-solid suspension.

\subsection{Analysis in the ignited sate}
\label{sec:ignited_state}

The ``{\it ignited}'' state~\citep{TK1995} represents the hydrodynamic state of fluidized-particles in rapid granular flow~\citep{Goldhirsch2003}, 
where the particles fly around randomly in between two collisions without getting much affected by the viscous drag  of the interstitial fluid.
A typical particle encounters successive collisions with other particles again and again before it can relax back to the local fluid velocity and hence 
the collision time is much smaller than the viscous relaxation time ($\tau_c\ll\tau_v$). 
 In this state, the particles have  strong velocity fluctuations,  resulting in  $T/\dot\gamma \sigma \gg 1$.

As in our recent work \citep{SA2014,SA2016},  the distribution function in the ignited state of a sheared suspension is assumed  to be an anisotropic Maxwellian,
\begin{equation}
 f({\boldsymbol c}, {\boldsymbol x}, t) =\frac{n}{(8\pi^3 |{\mathsfb M}|)^{1/2}}{\rm exp}\Big(-\frac{1}{2} {\boldsymbol C}\cdot{\mathsfb M}\cdot{\boldsymbol C}\Big),
 \label{eqn:anisotropic_maxwellian}
\end{equation}
 where $|{\mathsfb M}|=det({\mathsfb M})$. 
 This form of the distribution function has been used previously  in studying the  velocity dispersions in Saturn's rings \citep{GT1978,Shukhman1984,AT1986,Araki1988}
 as well as to analyse the  shear flow  of dry rapid granular flows \citep{JR1988,Richman1989,Lutsko2004}.
 
In the isotropic limit,  (\ref{eqn:anisotropic_maxwellian}) reduces to the  Maxwellian distribution function,
and an Hermite expansion of the form
\begin{eqnarray}
  f({\boldsymbol c}, {\boldsymbol x}, t) &=& \frac{n}{(2\pi T)^{3/2}}{\rm exp}\Big(-C^2/2T\Big)\sum_{i}a^{(i)}{\mathcal H}^{(i)} \nonumber \\
  &=&
  \frac{n}{(2\pi T)^{3/2}}{\rm exp}\Big(-C^2/2T\Big)\Big\{1+\frac{1}{2\rho T^2}P_{\langle \alpha\beta\rangle}C_\alpha C_\beta\Big\} + HOT,
 \label{eqn:10_moment_distribution}
\end{eqnarray}
represents the well-known Grad's moment expansion (GME)~\citep{Grad1949} --
such moment expansion has subsequently been employed to solve the Boltzmann equation for molecular gases \citep{HH1982,Kremer2010},
granular gases~\citep{JR1985,KM2011} and gas-solid suspensions \citep{TK1995,Sangani1996,CRG2015}.
Equation (\ref{eqn:10_moment_distribution}) with leading-order term 
($P_{\langle \alpha\beta\rangle}=\rho M_{\alpha\beta}-T\delta_{\alpha\beta}$ is the stress deviator)
yields the 10-moment system of \cite{Grad1949}, with density, velocity, temperature
and stress-deviator constituting the extended set of ten hydrodynamic fields~\citep{SA2016}.

\begin{figure}
\begin{center}
\includegraphics[scale=0.3]{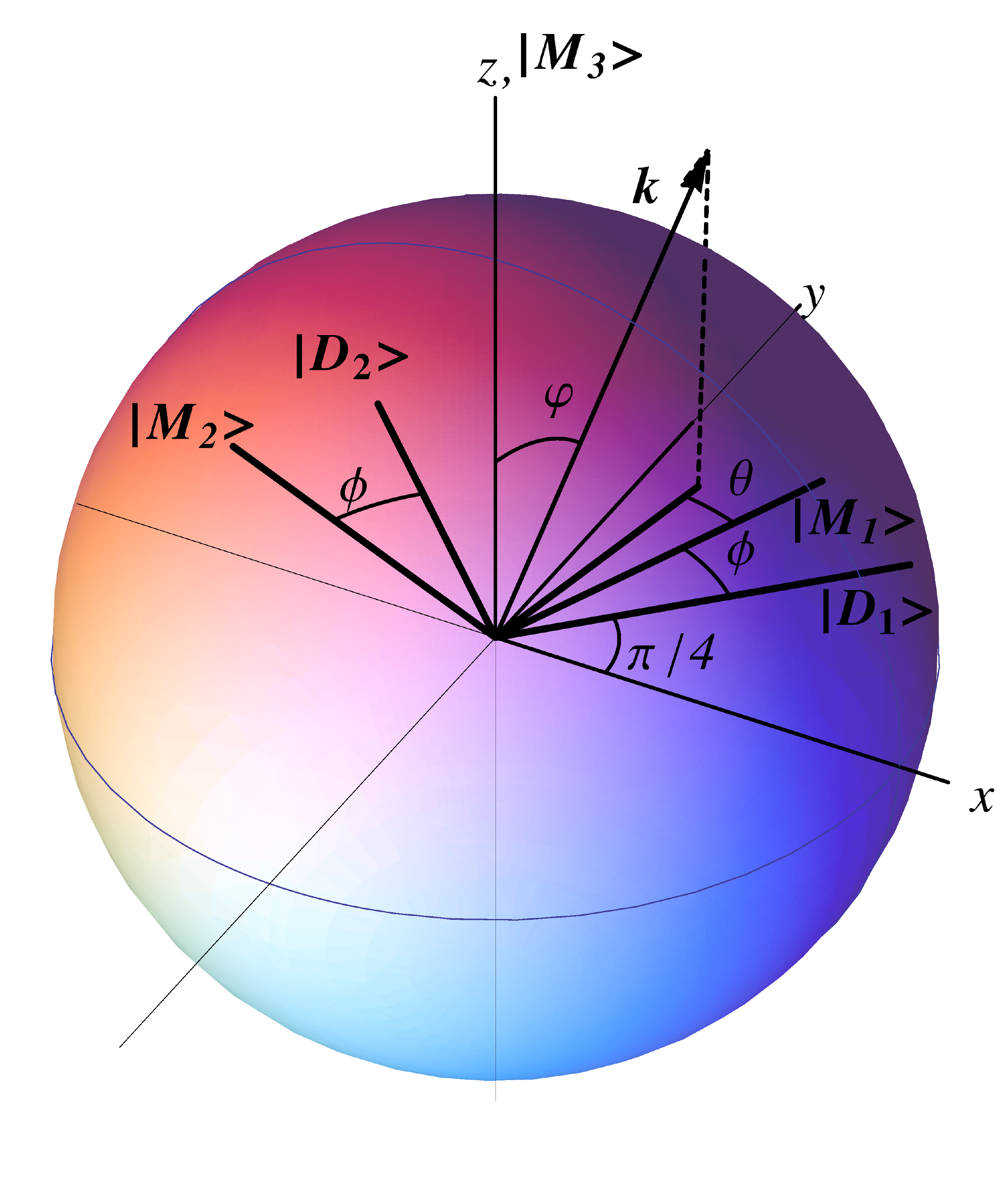}
\end{center}
 \caption{
 Schematic of the co-ordinate system and the eigen-basis for analysis;
the eigen-directions of the shear tensor ${\mathsfb D}$ and the second moment tensor ${\mathsfb M}$ are decpicted.
The uniform shear flow, ${\boldsymbol u}=({\dot\gamma}y, 0,0)$, is directed along the $x$-direction,
with the velocity gradient along the $y$-direction and the mean-vorticity along the $z$-direction.
 }
\label{fig:schematic}
\end{figure}

\subsubsection{Uniform shear flow (USF) and the second moment tensor}

The analysis in this section closely follows the theoretical  framework introduced  by \citep{GT1978,Shukhman1984,AT1986,JR1988,Richman1989}.
For the uniform shear flow, the velocity gradient tensor can be decomposed as
\begin{equation}
   \bnabla{{\boldsymbol u}} =
 {\mathsfb D} + {\mathsfb W}\equiv
 \left[
\begin{array}{ccc} 
 0 & \dot\gamma/2 & 0\\
 \dot\gamma/2 & 0 & 0 \\
 0 & 0 & 0
\end{array}
\right]
\quad
+
\quad 
\left[
\begin{array}{ccc} 
 0 & \dot\gamma/2 & 0\\
 -\dot\gamma/2 & 0 & 0\\
 0 & 0 & 0
\end{array}
\right] , 
\label{eqn:VelGrad1}
\end{equation}
where $\mathsfb D$ and $\mathsfb W$ are the shear and spin tensors, respectively.
Referring to figure~\ref{fig:schematic}, the ($x, y$)-plane is dubbed the {\it shear plane} and the $z$-direction is the {\it vorticity direction}. 
The eigenvalues of ${\mathsfb{D}}$ are $\dot{\gamma}/2$, $-\dot{\gamma}/2$ and $0$, with the corresponding orthonormal eigenvectors, respectively,
\begin{equation}
 \label{eigenvectorsofD}
 |D_1\rangle =
\left[
\begin{array}{ccc} 
 \cos\frac{\pi}{4}\\
 \sin\frac{\pi}{4}\\
 0
\end{array}
\right],
\quad
|D_2\rangle =
\left[
\begin{array}{ccc} 
 -\sin\frac{\pi}{4}\\
 \:\:\:\cos\frac{\pi}{4}\\
  0
\end{array}
\right] 
\quad
\mbox{and}
\quad
|D_3\rangle =
\left[
\begin{array}{ccc} 
 0\\
 0\\
  1
\end{array}
\right], 
\end{equation}
that are sketched in figure~\ref{fig:schematic}.
While $|D_3\rangle$ is directed along the $z$-axis, the shear-plane eigenvectors
$|D_1\rangle$ and $|D_2\rangle$ are rotated by $45\,^{\circ}$ anticlockwise from the $xy$-axes.

Since the granular temperature $T=M_{\alpha\alpha}/3$ is the isotropic measure  of the second moment tensor $\mathsfb{M}=\langle {\boldsymbol C}{\boldsymbol C}\rangle$,
we can decompose it as $ {\mathsfb{M}}/{T}= {\mathsfb I} + {\widehat{\mathsfb{M}}}/{T}$,
where $\widehat{\mathsfb{M}}/T$ is the dimensionless counterpart of its deviatoric/traceless tensor.
The eigenvalues of ${\mathsfb{M}}$ are denoted by $T(1+\xi)$, $T(1+\varsigma)$ and $T(1+\zeta)$,
with  $\xi$, $\varsigma$ and $\zeta$  being the eigenvalues of $\widehat{\mathsfb{M}}/T$ such that
\begin{equation}
   \xi+ \varsigma+\zeta=0.
\end{equation}
 The corresponding orthonormal set of eigen-directions
are assumed to be $|M_1\rangle$, $|M_2\rangle$ and $|M_3\rangle$, respectively,
as depicted in  figure~\ref{fig:schematic}.
Therefore, the second-moment tensor ${\mathsfb{M}}$ can be written in terms of its eigen-basis:
\begin{equation}
 \mathsfb{M} = 
    T(1+\xi)|M_1\rangle\langle M_1|+T(1+\varsigma)|M_2\rangle\langle M_2|+T(1+\zeta)|M_3\rangle\langle M_3| .
\label{eqn:M}
\end{equation}

Referring to figure~\ref{fig:schematic}, we assume that the shear-plane eigenvectors $|M_1\rangle$ and
$|M_2\rangle$ can be obtained by rotating the system of axes at an angle $(\pi/4+\phi)$,
with $\phi$ being unknown, 
in the anti-clockwise sense about the $z$-axis which coincides with $|M_3\rangle$:
\begin{equation}
 \label{eigenvectorsofM}
 |M_1\rangle =
\left[
\begin{array}{ccc} 
 \cos\left(\phi + \frac{\pi}{4}\right)\\
 \sin\left(\phi + \frac{\pi}{4}\right)\\
 0
\end{array}
\right],
\quad
|M_2\rangle =
\left[
\begin{array}{ccc} 
 - \sin\left(\phi + \frac{\pi}{4}\right)\\
 \:\:\:\cos\left(\phi + \frac{\pi}{4}\right)\\
  0
\end{array}
\right] 
\quad
\mbox{and}
\quad
|M_3\rangle =
\left[
\begin{array}{ccc} 
 0\\
 0\\
  1
\end{array}
\right]. 
\end{equation}
We further assume that the contact vector ${{\boldsymbol k}}$ makes an angle $\varphi$ with $|M_3\rangle$, 
and $\theta$ is the angle between $|M_1\rangle$ and 
${\boldsymbol{k}}-({\boldsymbol{k}}\bcdot{\boldsymbol{z}}){\boldsymbol{z}}$,
the  projection of ${\boldsymbol{k}}$ on the shear plane, as shown in figure~\ref{fig:schematic}.
Inserting (\ref{eigenvectorsofM}) into (\ref{eqn:M}), we obtain the
following expression for the second moment tensor
 \begin{equation}
 \mathsfb{M} = 
 T [\delta_{\alpha\beta}] + \widehat{\mathsfb M},
\label{eqn:tensorM}
\end{equation}
with its deviatoric part being given by  
\begin{equation}
 \widehat{\mathsfb M} = T
\left[
\begin{array}{ccc}
  \lambda^2 + \eta\sin2\phi & - \eta\cos2\phi & 0\\
 -\eta\cos2\phi & \lambda^2 - \eta\sin2\phi & 0\\
  0 & 0 & - 2\lambda^2
\end{array}
\right].
\label{eqn:hatM}
\end{equation}
Here we have introduced the following notations
\begin{equation}
  \eta \equiv \frac{1}{2}(\varsigma-\xi) \geq  0
  \quad
  \mbox{and}
  \quad
  \lambda^2 \equiv \frac{1}{2}(\varsigma+\xi) = -\frac{\zeta}{2} \geq  0,
\label{eqn:eta0}
\end{equation}
such
that the eigenvalues in the shear-plane can be expressed 
in terms of $\eta$ and $\lambda$ via
\begin{equation}
\xi =  \lambda^2 -\eta
\quad
\mbox{and}
\quad
\varsigma = \lambda^2 + \eta > \xi,
\label{eqn:M1M2}
\end{equation}
with the eigenvalue, $\zeta$, along the vorticity direction ($z$), being given by (\ref{eqn:eta0}).

Since $\phi=0$ implies that the shear tensor (${\mathsfb D}$) and the second-moment tensor (${\mathsfb M}$) have same principal directions,
a non-zero value of $\phi$ is a measure of  the {\it non-coaxiality} angle between the principal directions of ${\mathsfb D}$ and ${\mathsfb M}$.
It is straightforward to show that $\eta\sim (T_x - T_y)$ is proportional to the difference between two temperatures $T_x$ and $T_y$ on the shear-plane ($x,y$),
and hence $\eta\neq 0$ is indicative of  the degree  of {\it temperature-anisotropy} on the shear plane.
On the other hand, a non-zero value of $\lambda^2$ is a measure of the {\it excess temperature}~\citep{SA2016},
\begin{equation}
  T^{ex}_z = (T-T_z) = 2\lambda^2 T
  \quad
  \Rightarrow\;\;
  \lambda^2 = \frac{T^{ex}_z}{2T},
  \label{eqn:Texcess1}
\end{equation}
along the mean vorticity direction. In summary, the anisotropy of ${\mathsfb M}$ is quantified in terms of three
dimensionless quantities: (i) $\eta\propto(T_x-T_y) \neq 0$, or, $\phi\neq 0$
and (ii) $\lambda^2\propto T_z^{ex} \neq 0$.

The second-moment tensor (\ref{eqn:tensorM}-\ref{eqn:hatM}) in the USF of suspension, constructed from its eigen-basis,
is therefore completely determined when $T$, $\eta$, $\phi$ and $\lambda^2$  are specified;
the dependence on the Stokes number $St$ and the particle volume fraction ($\nu$) is implicit as will be made clear below.

 \subsubsection{Source term in the  ignited state}

Employing (\ref{eqn:anisotropic_maxwellian}), the collisional production term (\ref{eqn:source_second_moment_usf}) for the ignited state has been evaluated
\begin{eqnarray}
& &  \aleph_{\alpha\beta}^{is}=-\frac{6(1+e)\rho_p \nu^2}{\upi^\frac{3}{2}\sigma}\Big\{(1-e)\int k_\alpha k_\beta 
 (\boldsymbol{k}\cdot {\mathsfb M} \cdot \boldsymbol{k})^\frac{3}{2} d\boldsymbol{k} \nonumber\\
   & & \hspace*{4.0cm} + 2\int (k_\alpha j_\beta+j_\alpha k_\beta) 
           (\boldsymbol{k}\cdot {\mathsfb M} \cdot \boldsymbol{k})^\frac{1}{2} (\boldsymbol{k}\cdot {\mathsfb M} \cdot \boldsymbol{j}) d\boldsymbol{k}\Big\}.
    \label{eqn:source_ignited}        \\
 && \qquad =  -\frac{4(1+e)\rho_p\nu^2 T^{3/2}}{35\sigma\sqrt{\pi}} \Big\{(1-e)\times\nonumber\\
 &&
 \left[
 \begin{array}{ccc}
  70+9\eta^2+42\lambda^2+42\eta\sin2\phi & -42\eta\cos2\phi & 0\\
  -42\eta\cos2\phi & 70+9\eta^2+42\lambda^2-42\eta\sin2\phi & 0\\
   0 & 0 & 70+3\eta^2-84\lambda^2
 \end{array}
\right]\nonumber\\
&& \quad + \; 4\left[
 \begin{array}{ccc}
  \eta^2+21\lambda^2+21\eta\sin2\phi & -21\eta\cos2\phi & 0\\
  -21\eta\cos2\phi & \eta^2+21\lambda^2-21\eta\sin2\phi & 0\\
   0 & 0 & -2(\eta^2+21\lambda^2)
 \end{array}
\right] \Big\} ,
\label{eqn:source_ignited_explicit}
\end{eqnarray}
which is a function of $\nu$, $e$, $T$, $\eta$, $\phi$ and $\lambda^2$.
In the final expression (\ref{eqn:source_ignited_explicit}), we have retained  terms that are up-to second-order in $\eta$, $\sin\phi$ and $\lambda$
-- we shall show in the end that this is sufficient to yield accurate predictions of transport coefficients 
of  a sheared  dilute suspension.for a wide range of (i) restitution coefficient $e$ and (ii) Stokes number $St$.

\subsection{Analysis in the quenched sate}
\label{sec:quenched_state}

\cite{TK1995} envisaged a scenario of a dilute gas-solid suspension in which  the particle inertia is  very low  such that the particles tend to align with fluid streamlines after a collision. 
Most of the particles will be having their individual velocity equal to the fluid velocity $({\boldsymbol c}\approx{\boldsymbol u})$ which implies 
that the peculiar velocity ${\boldsymbol C}\approx 0$ and therefore  the particle agitation is very small  ($T/ \dot\gamma \sigma \ll 1$) -- this is dubbed the {\it quenched} state.
The collisions in this state are mainly shear-induced with some occasional variance-driven collisions and the particles relax back to the local fluid velocity after such a collision before they encounter a second collision and therefore the viscous relaxation time is much smaller than the collision time $\tau_v \ll \tau_c$. 
The velocity distribution function of the quenched state is taken to be a delta function
\begin{equation}
  f=n\delta({\boldsymbol C}) ,
 \label{eqn:distribution_quenched}
\end{equation}
which is a solution of the Boltzmann equation.
Using (\ref{eqn:distribution_quenched}), the collisional production term at second-order can be evaluated as
\begin{eqnarray}
 \aleph_{\alpha\beta}^{qs} &=& -\rho{\dot\gamma}^3 \sigma^2\frac{3(1+e)^2\nu}{2\upi}\int_{k_x, k_y> 0}(k_xk_y)^3k_\alpha k_\beta d\boldsymbol{k}, \nonumber \\
 &=& \rho_p{\dot\gamma}^3 \sigma^2\frac{(1+e)^2\nu^2}{16}
\left[
 \begin{array}{ccc}
  \frac{512}{315\upi} & -\frac{16}{35} & 0\\
  -\frac{16}{35} & \frac{512}{315\upi} & 0\\
   0 & 0 & \frac{128}{315\upi}
 \end{array}
\right].
\label{eqn:source_quenched}
\end{eqnarray}
Note that this expression differs from that of \cite{TK1995} by a numerical-factor $2$ which was also noted previously~\citep{PS2012}.

\subsection{Second moment balance combining quenched and ignited states}
\label{sec:quenched_to_ignited}

Combining the ignited and quenched states, the second-order moment balance equation (\ref{eqn:balance_second_moment_usf})
for a `dilute' gas-solid suspension undergoing uniform shear flow is
\begin{equation}
    P_{\delta\beta}u_{\alpha,\delta}+P_{\delta\alpha}u_{\beta,\delta} +\frac{2\dot\gamma}{St} P_{\alpha\beta} =\aleph_{\alpha\beta}
       \equiv  \aleph_{\alpha\beta}^{qs}+ \aleph_{\alpha\beta}^{is},
 \label{eqn:balance_second_moment}
\end{equation}
where the superscripts $qs$ and $is$ stand for the source of second moment in quenched and ignited states, respectively. 
Following (\ref{eqn:tensorM}-\ref{eqn:hatM}), the expression for the stress tensor can be written as
\begin{equation}
 {\mathsfb P} = \rho{\mathsfb M} =\rho_p\nu T
\left[
\begin{array}{ccc}
 1+\lambda^2+\eta\sin2\phi & -\eta\cos2\phi & 0\\
 -\eta\cos2\phi & 1+\lambda^2-\eta\sin2\phi & 0\\
  0 & 0 & 1-2\lambda^2
\end{array}
\right] .
\label{eqn:tensorP}
\end{equation}

Substituting (\ref{eqn:source_ignited}), (\ref{eqn:source_quenched}) and (\ref{eqn:tensorP})
into  (\ref{eqn:balance_second_moment}), we obtain the following four independent equations:
\begin{equation}
\left.
\begin{array}{rcl}
 -2T\eta\cos2\phi +\frac{2}{St}T (1+\lambda^2+\eta\sin2\phi) &=& \Big[-\frac{2(1-e^2)\nu T^\frac{3}{2}}{35\sqrt{\upi}}(70+9\eta^2+42\lambda^2+42\eta\sin2\phi)\\
 &&-\frac{8(1+e)\nu T^\frac{3}{2}}{35\sqrt{\upi}}(\eta^2+21\lambda^2+21\eta\sin2\phi)\Big]\\
 &&+\Big[\frac{128(1+e)^2\nu}{315\pi}\Big],\\
 \frac{2}{St}T (1+\lambda^2-\eta\sin2\phi) &=& \Big[-\frac{2(1-e^2)\nu T^\frac{3}{2}}{35\sqrt{\upi}}(70+9\eta^2+42\lambda^2-42\eta\sin2\phi)\\
 &&-\frac{8(1+e)\nu T^\frac{3}{2}}{35\sqrt{\upi}}(\eta^2+21\lambda^2-21\eta\sin2\phi)\Big]\\
 &&+\Big[\frac{128(1+e)^2\nu}{315\pi}\Big],\\
 \frac{2}{St}T (1-2\lambda^2) &=& \Big[-\frac{2(1-e^2)\nu T^\frac{3}{2}}{35\sqrt{\upi}}(70+3\eta^2-84\lambda^2)\\
 &&+\frac{16(1+e)\nu T^\frac{3}{2}}{35\sqrt{\upi}}(\eta^2+21\lambda^2)\Big]+\Big[\frac{32(1+e)^2\nu}{315\pi}\Big],\\
 T(1+\lambda^2-\eta\sin2\phi)-\frac{2}{St}T\eta\cos2\phi &=& \Big[\frac{12(1-e)(3-e)\nu T^\frac{3}{2}}{5\sqrt{\upi}}\eta\cos2\phi -\frac{4(1+e)^2\nu}{35\pi}\Big].
 \end{array}
\right\},
\label{eqn:components_balance_second_moment}
\end{equation}
Note that  the terms involving the Stokes number ($St$) on the left-hand sides of  (\ref{eqn:components_balance_second_moment})
vanish in the limit of $St\to\infty$, thereby recovering the second-moment balance for the shear flow of a `dry' granular gas~\citep{SA2016}.

In (\ref{eqn:components_balance_second_moment}), we have made temperature dimensionless via $T=T/(\dot\gamma\sigma/2)^2$.
The coupled system of equations (\ref{eqn:components_balance_second_moment}) must be solved to determine $\eta$, $\lambda$, $\phi$ and $T$ for specified values of 
(i) particle volume fraction ($\nu$), (ii) Stokes number ($St$) and (iii) restitution coefficient ($e$). 
Analytical progress can be made to solve (\ref{eqn:components_balance_second_moment}) as discussed in \S3 and \S4.

Before proceeding further, it may be noted that the analysis of the second moment balance (\ref{eqn:balance_second_moment})
or (\ref{eqn:components_balance_second_moment})
 in the ignited state (i.e.~with $\aleph_{\alpha\beta}^{qs}=0$)
is considerably simplified for elastically-colliding ($e=1$) particles, see Appendix A.  The related  analytical results on the temperature field 
provide a lower-bound on the Stokes number for the existence of the ignited state (and consequently on the multiple states and hysteresis, \S3.2)
in a dilute gas-solid suspension.

\begin{figure}
 \begin{center}
 (a)
  \includegraphics[scale=0.38]{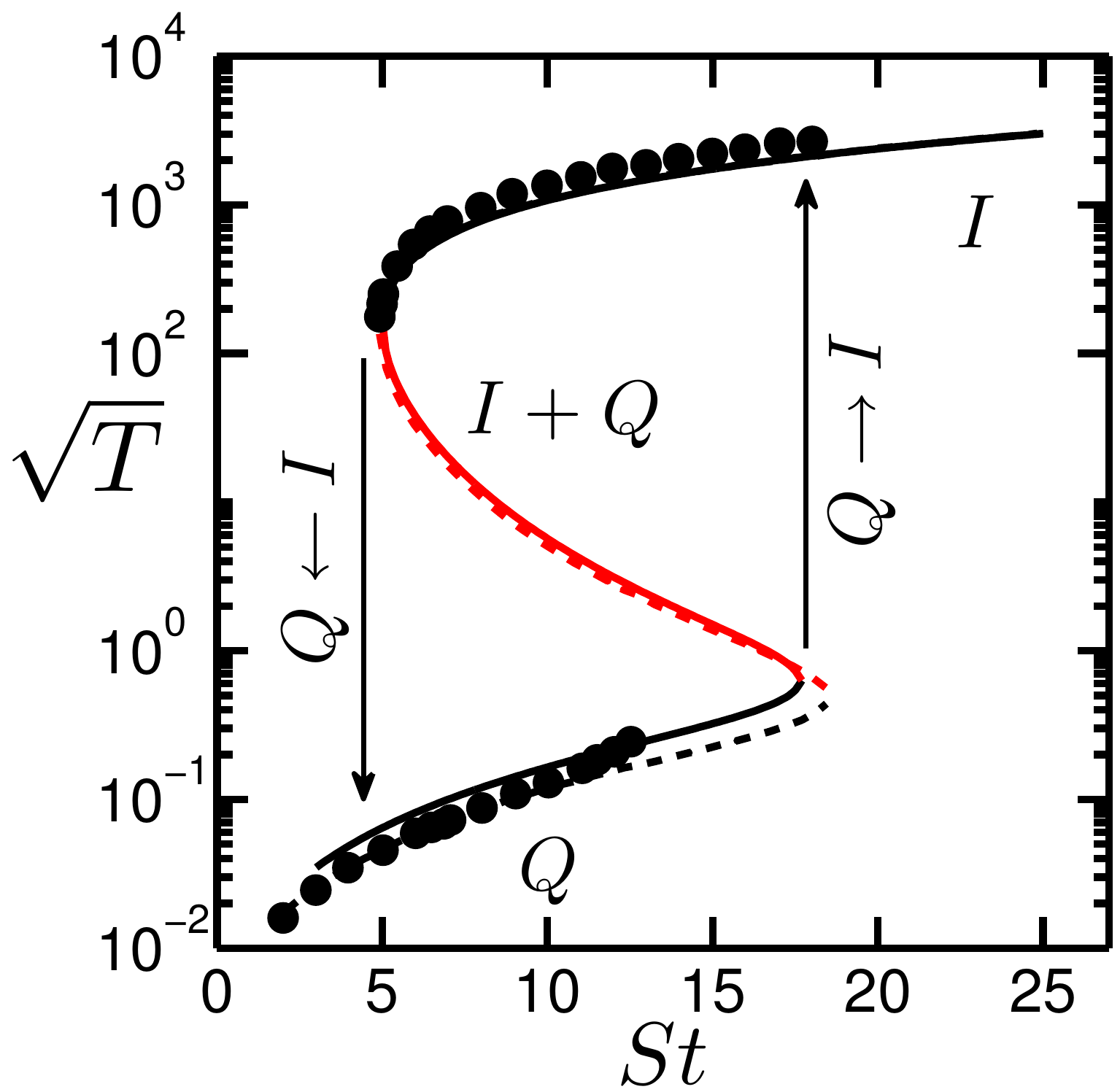}
  (b)
  \includegraphics[scale=0.38]{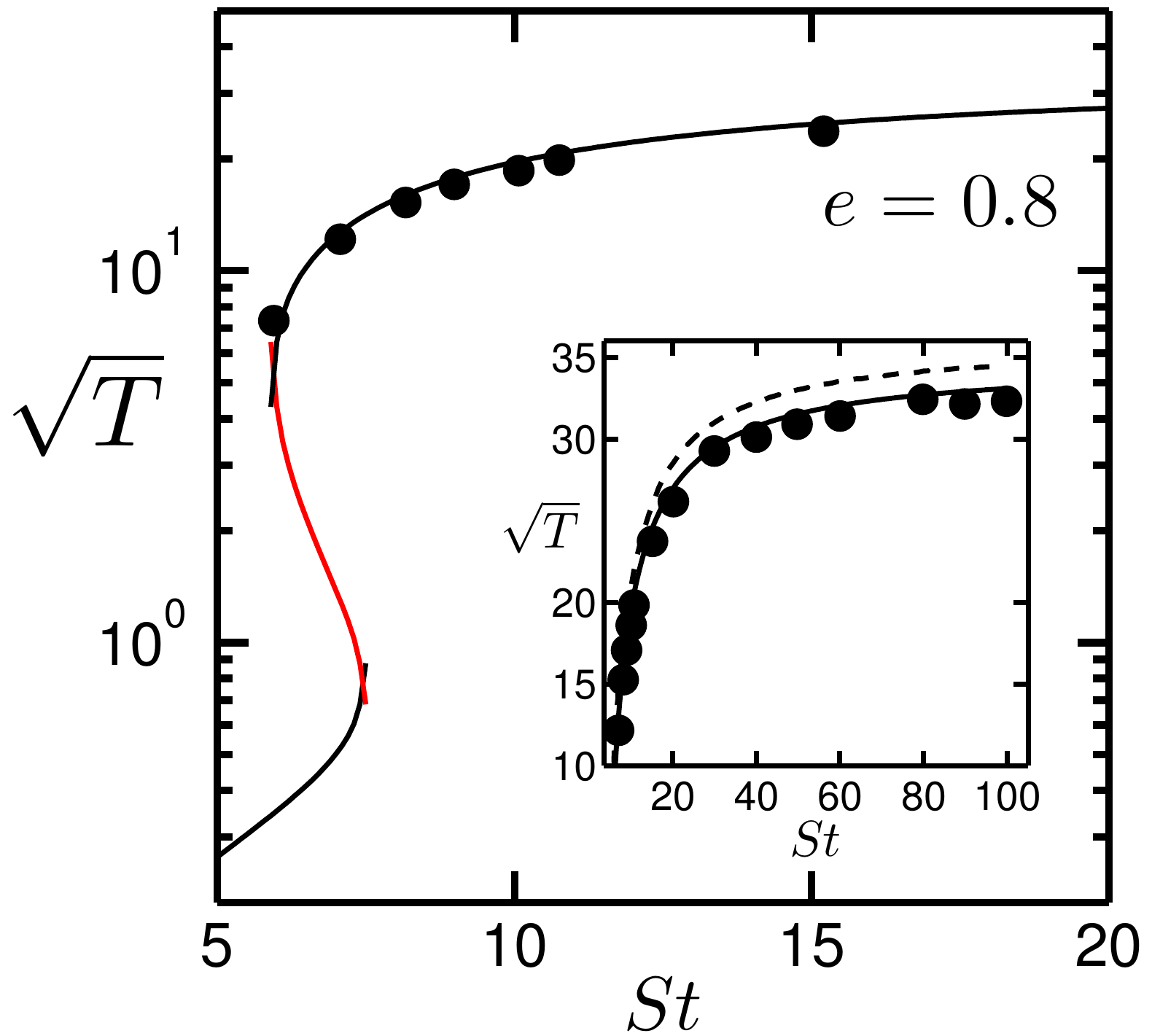}
  (c)
  \includegraphics[scale=0.38]{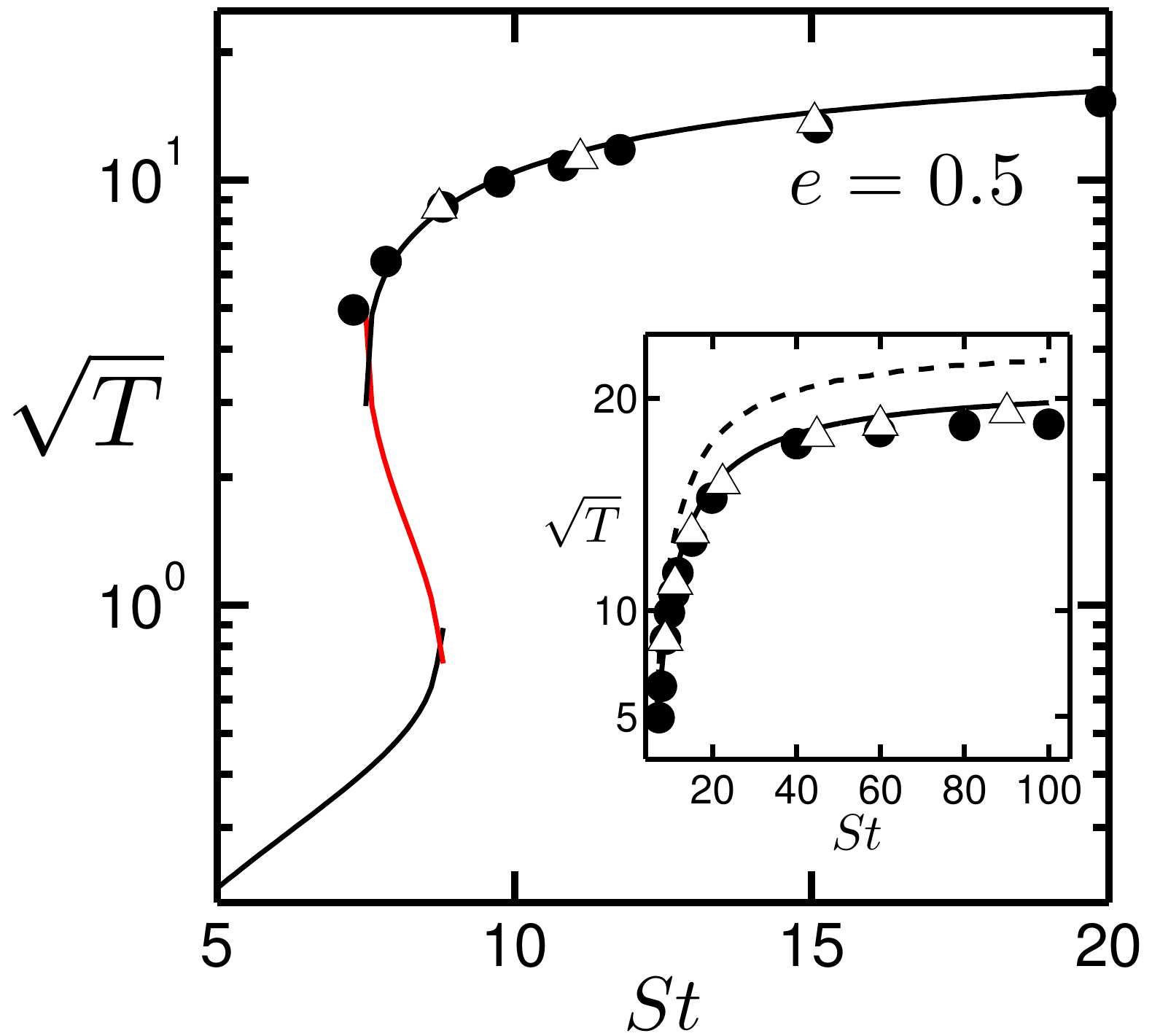}
  (d)
  \includegraphics[scale=0.38]{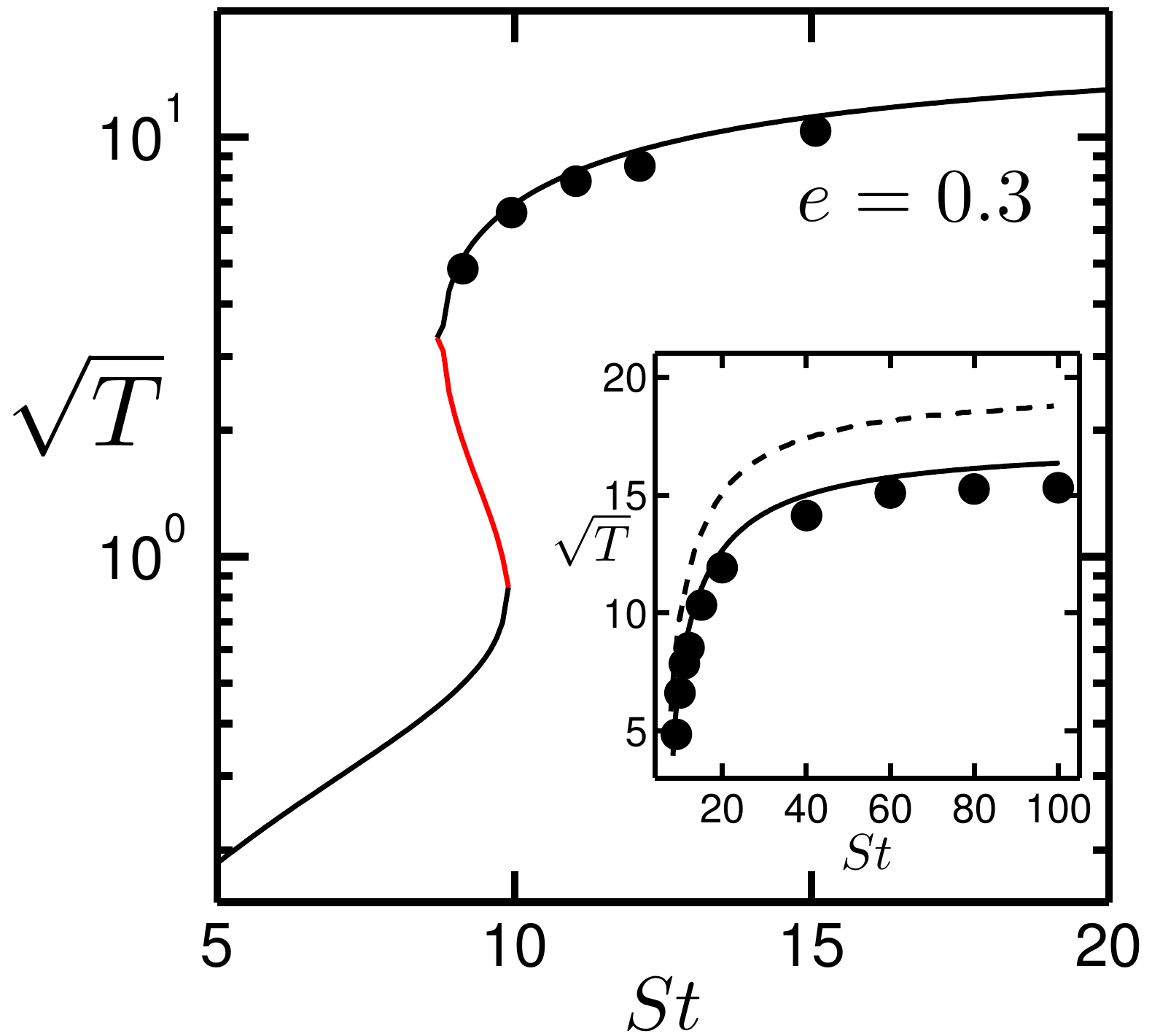}
\end{center}
\caption{
Hysteretic/first-order  transitions of granular temperature for
(a) $e=1$ and $\nu=5\times 10^{-4}$; (b) $e=0.8$, (c) $e=0.5$ and (d) $e=0.3$ with $\nu=0.01$.
The solid and dashed (inset) lines denote the present anisotropic-Maxwellian theory and the Maxwellian theory \citep{TK1995,Sangani1996}, respectively.
The filled-circles  represent the DSMC data of  \cite{Sangani1996}; the open-triangles in panel $c$ denote the DSMC data of \cite{CRG2015}.
In each panel, the black and red lines represent stable and unstable solutions, respectively, of Eq. (3.1).
}
\label{fig:fig2}
\end{figure}

\section{Granular temperature: Multi-stability and ignited-to-quenched state transitions}

After some tedious algebra, we found that   (\ref{eqn:components_balance_second_moment}) can be decoupled to yield
 a $10$-th degree polynomial for granular temperature $\xi=\sqrt{T}$:
\begin{equation}
 \mathcal{G}(\xi) \equiv a_{10}\xi^{10}+a_{9}\xi^{9}+a_{8}\xi^{8}+a_{7}\xi^{7}+a_{6}\xi^{6}+a_{5}\xi^{5}+a_{4}\xi^{4}+a_{3}\xi^{3}+a_{2}\xi^{2}+a_{1}\xi+a_{0}= 0,
 \label{eqn:energy_balance}
\end{equation}
the explicit expressions of the coefficients $a_i$  are given in Appendix B.
It is straightforward to verify  that for the case of elastically colliding particles ($e=1$), $a_{10}=0=a_9=a_8$ and hence (\ref{eqn:energy_balance})
reduces to a polynomial of 7th-degree; in fact these three roots vapourize to $-\infty$ at $e=1$ and remain negative for $e<1$ and hence unphysical.
It has been verified numerically (as well as via an ordering analysis, see Appendix B) that at most three roots of (\ref{eqn:energy_balance}) are real positive, 
depending on the values of $\nu$, $St$ and $e$, and the remaining roots are negative and/or complex.

\subsection{Validation of present anisotropic-Maxwellian theory}

First, we  solve the temperature equation (\ref{eqn:energy_balance}) numerically
and compare it with simulation data in order to validate the present theory.

Figure \ref{fig:fig2}(a,b,c,d) shows the variations of the granular temperature with Stokes number $(St)$ at particle volume fractions 
of $(a)$ $\nu=5\times10^{-4}$ and $(b,c,d)$ $\nu=0.01$, with  different values of the restitution coefficient (a) $e=1$, (b) $e=0.8$, (c) $e=0.5$ and $e=0.3$. 
In each panel and inset, the symbols represent the DSMC (direct simulation Monte Carlo) data of \cite{Sangani1996} 
which are compared with the (i) present anisotrpic-Maxwellian theory (solid line)
and (ii) the standard moment expansion (dashed line) of Tsao \& Koch (1995, for $e=1$) and Sangani et al. (1996, for $e\neq 1$),
Figure~\ref{fig:fig2}($a$) indicates  that for the case of elastically colliding particles, the present theory is on par with Tsao-Koch theory.
On the other hand, for inelastic particles ($e<1$), the insets of figure~\ref{fig:fig2}(b,c,d) confirm that the present theory is able to better predict the temperature-variation with $St$; 
however, the agreement with Tsao-Koch theory worsens with increasing dissipation.
In panel $c$, the recent DSMC data (open triangles) of \cite{CRG2015} for $e=0.5$ also agree quantitatively  with the present theory.

Overall, the moment theory with anisotropic-Maxwellian as the leading term 
seems   better suited for a dilute gas-solid suspension of inelastic particles undergoing shear flow for a large range of $e<1$ at small and moderate values of Stokes number.
It may be noted that a similar analysis~\citep{SA2014,SA2016} for a sheared granular gas ($St=\infty$) 
provides excellent predictions for temperature and rheological quantities for highly dissipative particles.
The same conclusions seem to  carry over to the limit of small Stokes numbers of a sheared gas-solid suspension too 
-- this issue is further discussed in \S5 (with respect to predictions for viscosity and normal stress differences).

\subsection{Analytical solution for three temperatures: hysteresis and multi-stability}

Returning to figure~\ref{fig:fig2}, we note that the temperature is a multi-valued function of Stokes number
for a range of $St$ over which  there are three possible solutions;
there are  hysteretic/discontinuous jumps  in temperature  from the low/high temperature branches with increasing/decreasing $St$.
For a better understanding of this hysteresis phenomenon, 
equation (\ref{eqn:energy_balance}) has been solved in the asymptotic limit $\nu\ll1$, $St\gg1$, and $St^3\nu\ll1$ via an ordering analysis,
the details of which are given in Appendix C.
Three real solutions have been found,
\begin{eqnarray}
 \sqrt{T_{is}} &=& \frac{5(1+e)^{-1}(1691+539e-1223e^2+337e^3)\sqrt{\upi}}{48(3-e)(12607-19952e+10099e^2-1746e^3)}\left(\frac{St}{\nu}\right)
 \;\stackrel{e=1}{\equiv} \; \frac{5\sqrt{\pi}}{144}\frac{St}{\nu} ,
 \label{eqn:temperature_ignited}
 \\
 \sqrt{T_{qs}} &=& \sqrt{\frac{32(1+e)^2}{945\upi}} St^{3/2}\nu^{1/2} 
  \;\stackrel{e=1}{\equiv} \; \frac{8\sqrt{2}}{3\sqrt{105\pi}} St^{3/2}\nu^{1/2} ,
 \label{eqn:temperature_quenched}
\\
 \sqrt{T_{us}} &=& \frac{840\sqrt{\upi}}{(1+e)(107+193e)}\left(\frac{1}{St^3\nu}\right)
  \;\stackrel{e=1}{\equiv} \; \frac{7\sqrt{\pi}}{5}\left(\frac{1}{St^3\nu}\right) ,
 \label{eqn:temperature_unstable}
\end{eqnarray}
which correspond to the temperatures in the ignited ($T_{is}$), quenched ($T_{qs}$) and unstable ($T_{us}$) states, respectively.
These three solutions (\ref{eqn:temperature_ignited}-\ref{eqn:temperature_unstable}) 
 can be identified in  figure~\ref{fig:fig2} as the high-, low-, and intermediate-temperature branches, respectively;
the red-colored solution branch in each panel of figure~\ref{fig:fig2}  represent $T_{us}$ which is of course {\it unstable} from stability viewpoint (see \S4.2 for related discussions).

It is clear from (\ref{eqn:temperature_ignited}) that $T_{is}$ increases with increasing Stokes number $St$, but decreases with increasing  particle volume fraction $\nu$.
On the other hand,  the quenched-state temperature (\ref{eqn:temperature_quenched}) increases with increasing  $St$ and $\nu$,
whereas  the unstable temperature (\ref{eqn:temperature_unstable}) decreases with  increasing $St$ and $\nu$.
These overall predictions are verified in  figure~\ref{fig:fig3} which display
the variations of  granular temperature as functions of ($\nu, e$) for two values of  Stokes number (a) $St=10$ 
and (b) $St=20$. In each panel, the upper-most branch corresponds to the ignited-state of high temperature $T_{is}$;
 the middle and the lower-most planes  represent the unstable and quenched states, respectively.
 The latter two states are connected via a  line of turning-points,
 resulting in   saddle-node bifurcations (jump-transitions) from ``$Q\to I$'' with increasing $\nu$, above which  the ignited state is the only solution.
 The critical density $\nu=\nu_c(St, e)$ for this transition increases with increasing inelasticity
but decreases with increasing $St$ (see panel $b$).
The corresponding Stokes number for ``$Q\to I$''-transition is denoted by $St_{c_2}(\nu,e)$
which can also be identified with the right limit-point in figure~\ref{fig:fig2}.

\begin{figure}
(a)
\includegraphics[scale=0.32]{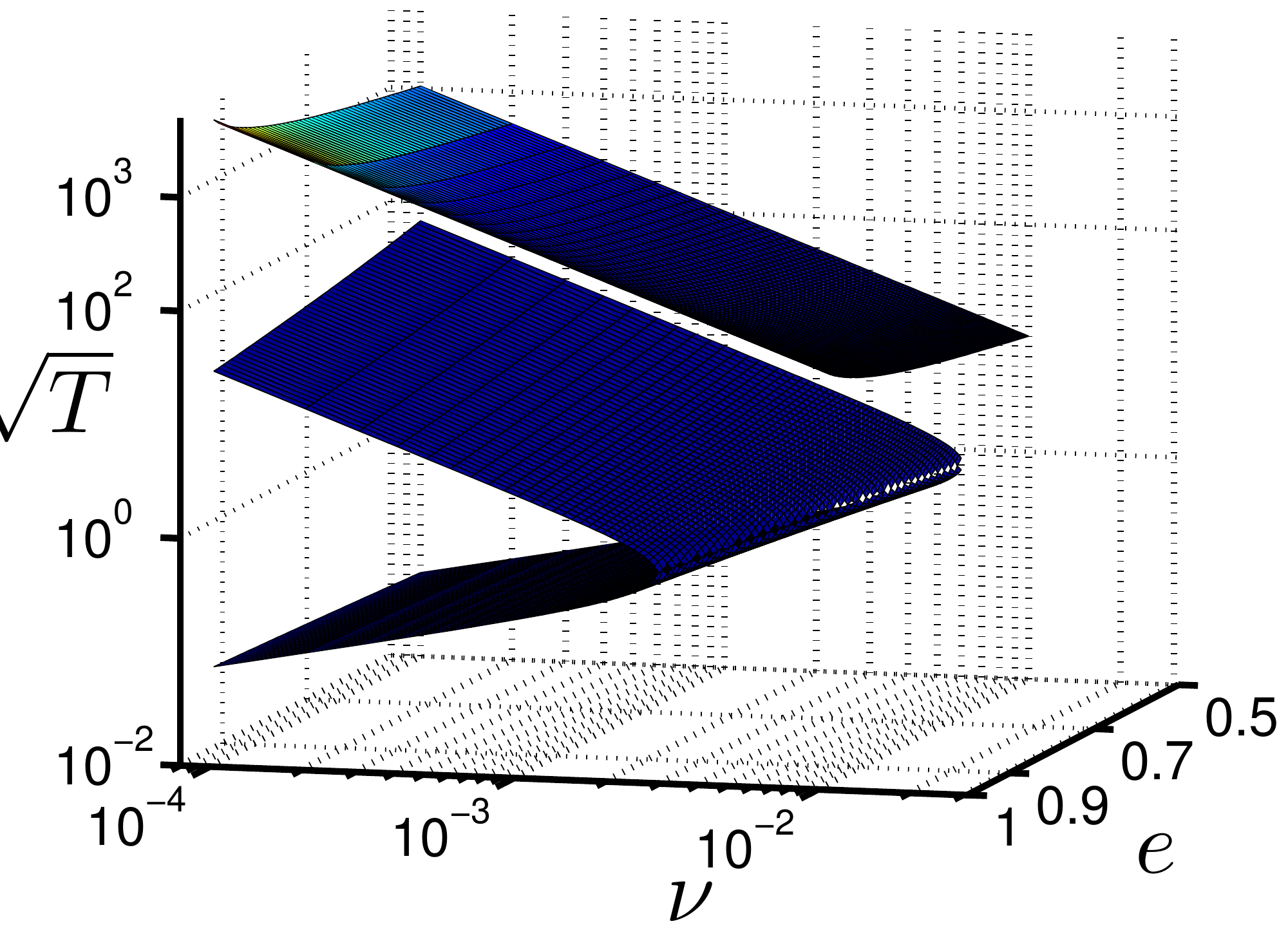}
  (b)
  \includegraphics[scale=0.32]{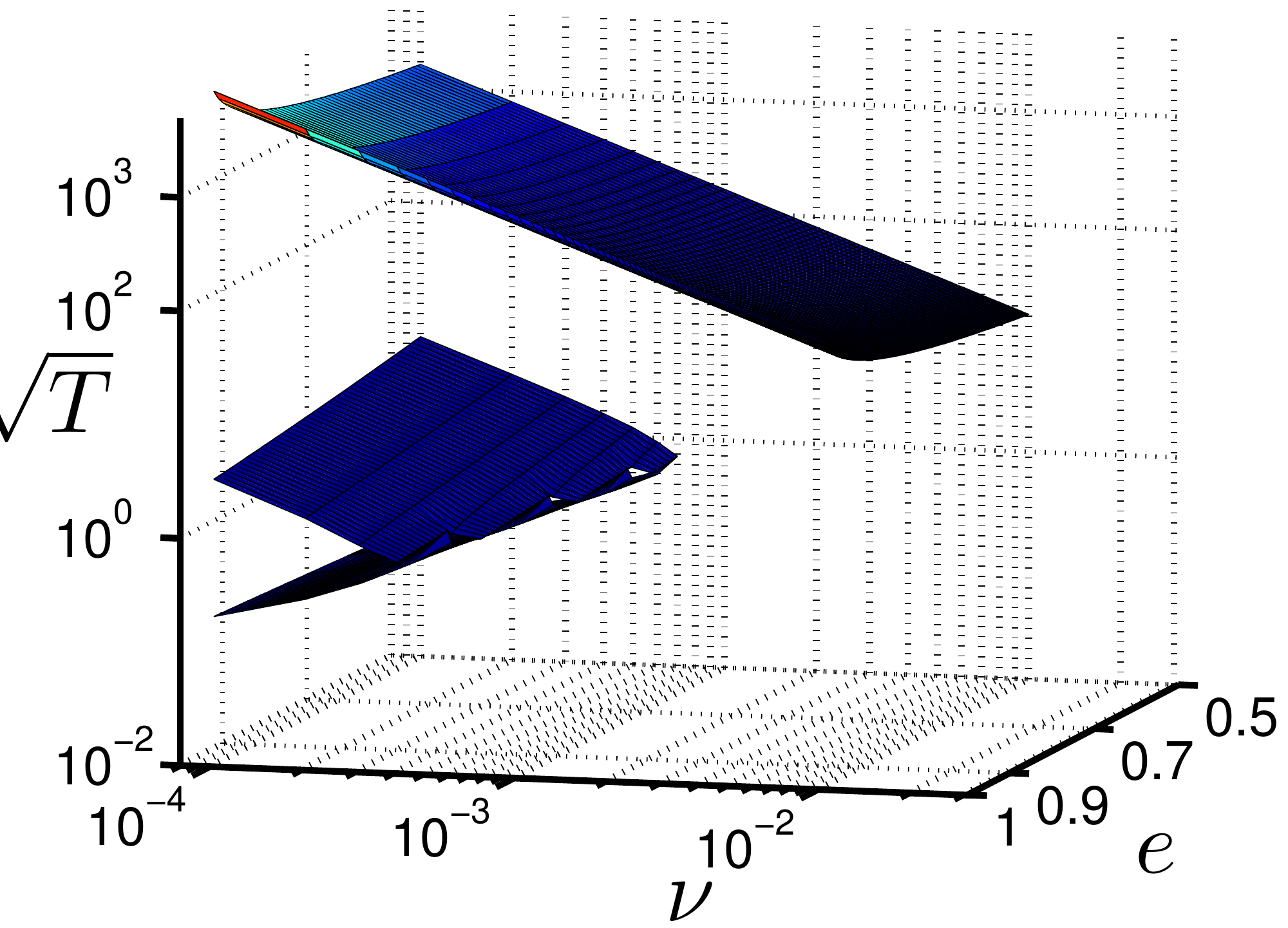}
  \caption{
   Multiple states of granular temperature as functions of the mean volume fraction $\nu$ and restitution coefficient $e$ for (a) $St=10$ and (b) $St=20$.
    }
\label{fig:fig3}
\end{figure}

A noteworthy feature of figure~\ref{fig:fig3} is that the ignited branch [$T\propto \nu^{-2}$, see (\ref{eqn:temperature_ignited})]
is disconnected from the quenched and unstable branches, and 
therefore there is no jump-transitions (on decreasing $\nu$) from $I\to Q$ at $St=10$ (panel $a$) and $20$ (panel $b$).
However, on further decreasing the Stokes number (below $St=10$), the ignited state solution  disappears below a minimum $St$ --
how this process occurs is explained in figures~\ref{fig:fig4}(a,b,c) for $e=1$, $0.8$ and $0.5$, respectively.
In particular, at any $e$, the unstable branch (red line) and the ignited-branch come closer  with decreasing $St$
and merge with each other at some minimum $St$ below which only the quenched-state solution [$T\propto \nu$, see (\ref{eqn:temperature_quenched})] survives.
Similarly, by fixing the Stokes number at $St=6$ but increasing the inelasticity (decreasing $e$) also
results in the disappearance of the ignited state solution, see figure~\ref{fig:fig4}(d).
Therefore, the quenched state is the only possible solution below a minimum Stokes number $St=St_{c_1}(e,\nu)$ --
this can  be identified with the left limit-point in figure~\ref{fig:fig2} for ``$I\to Q$'' transition.

\begin{figure}
(a)
\includegraphics[scale=0.38]{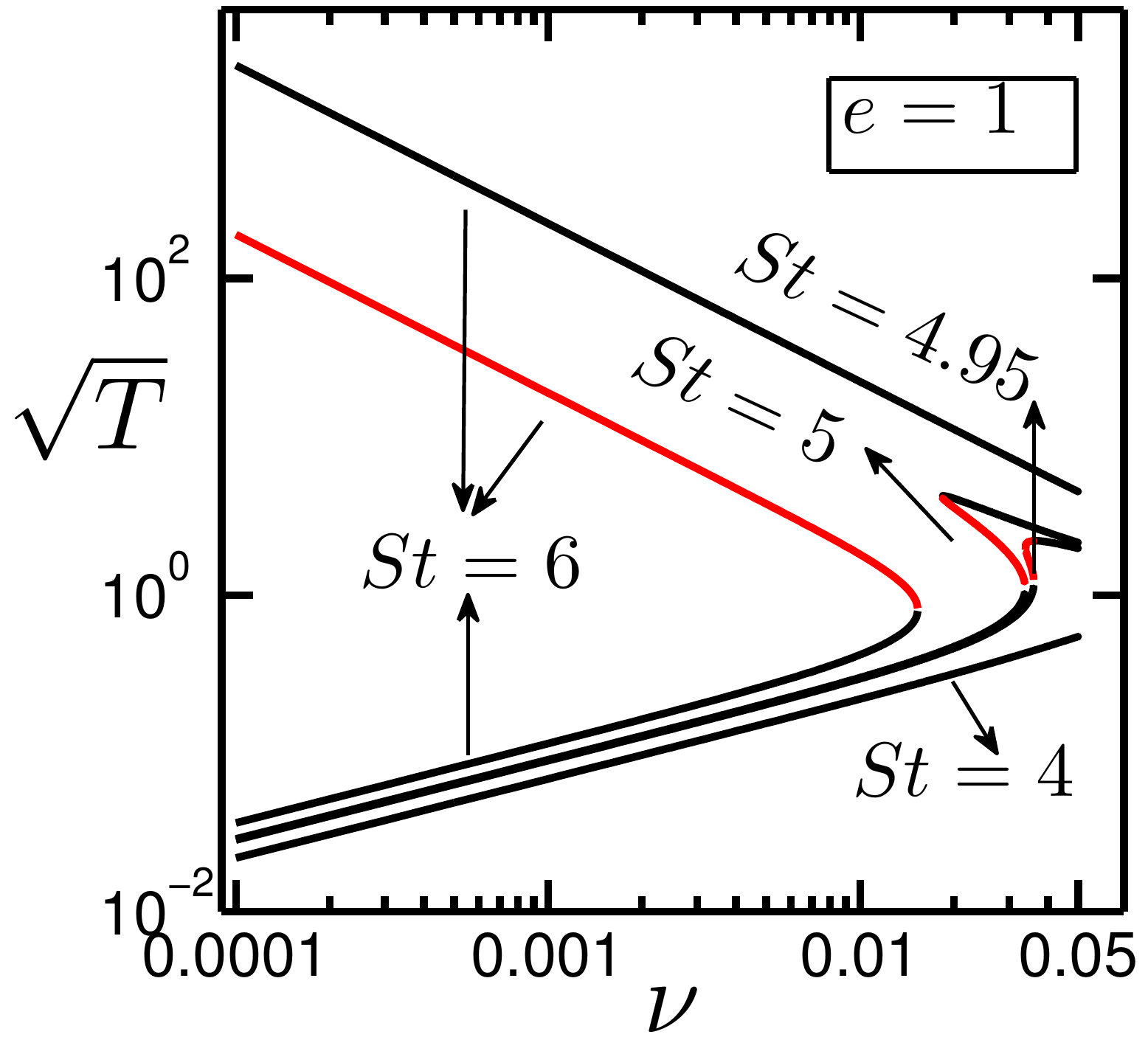}
    (b)
   \includegraphics[scale=0.38]{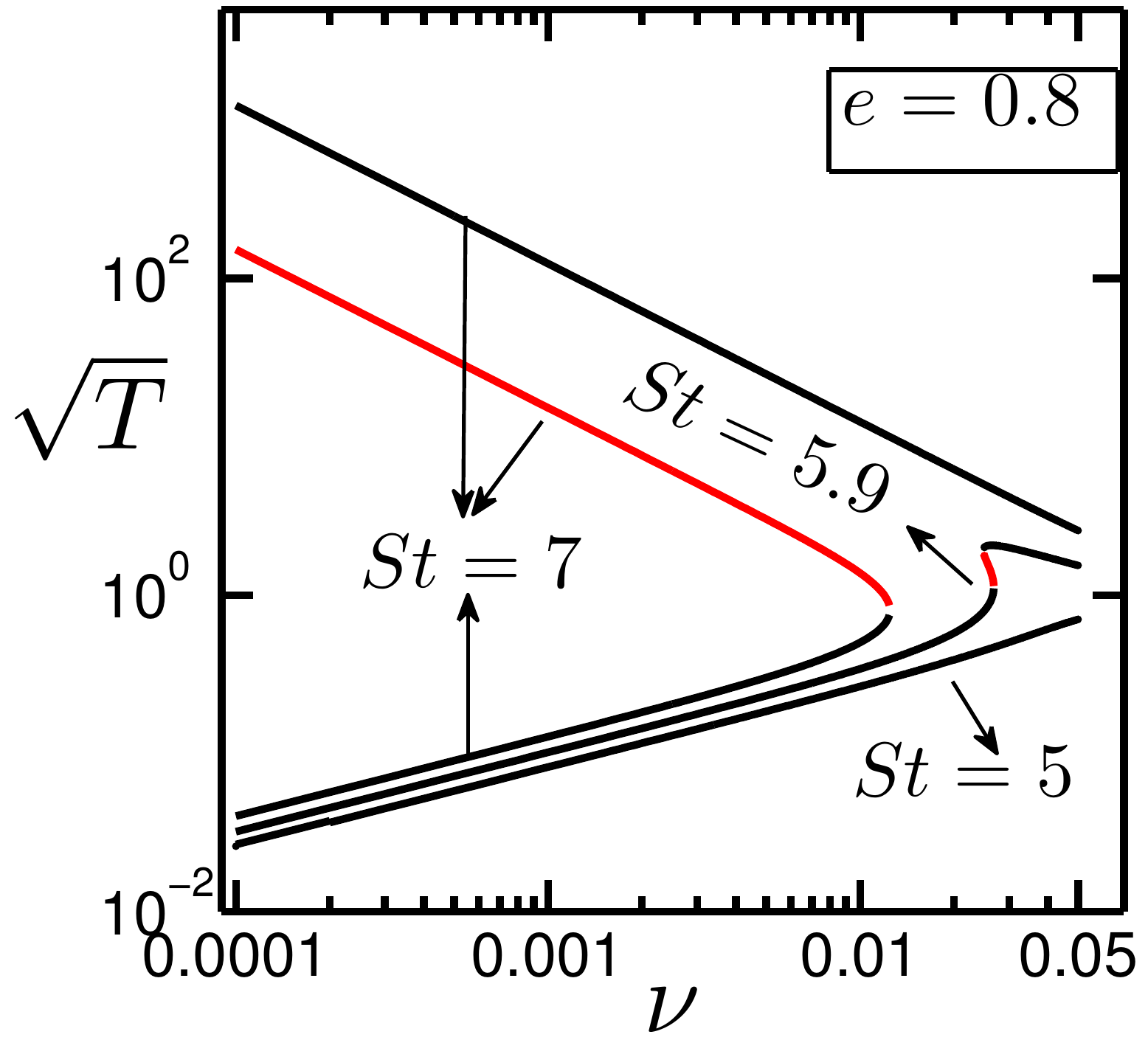}\\
     (c)
    \includegraphics[scale=0.38]{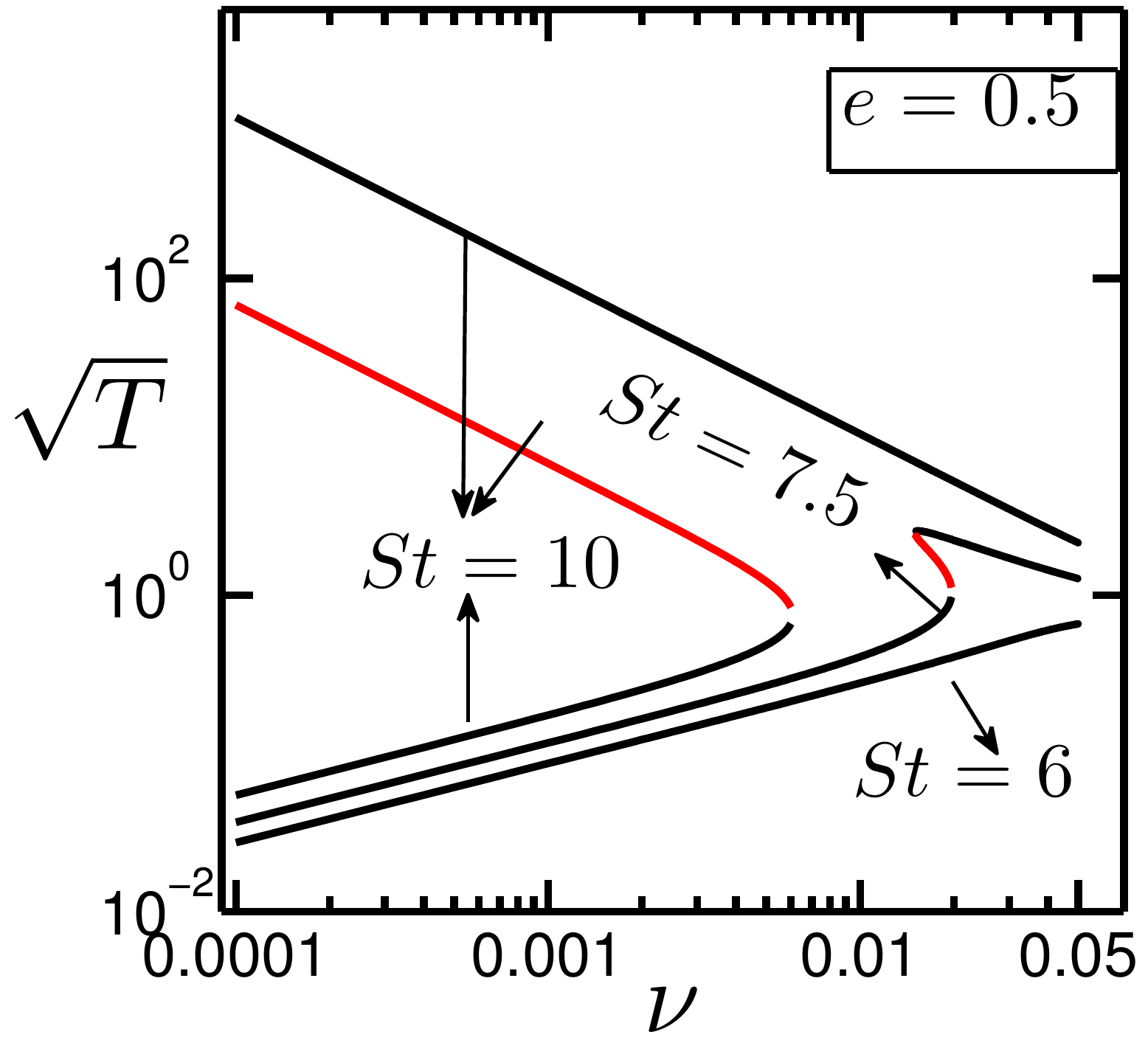}
    (d)
    \includegraphics[scale=0.38]{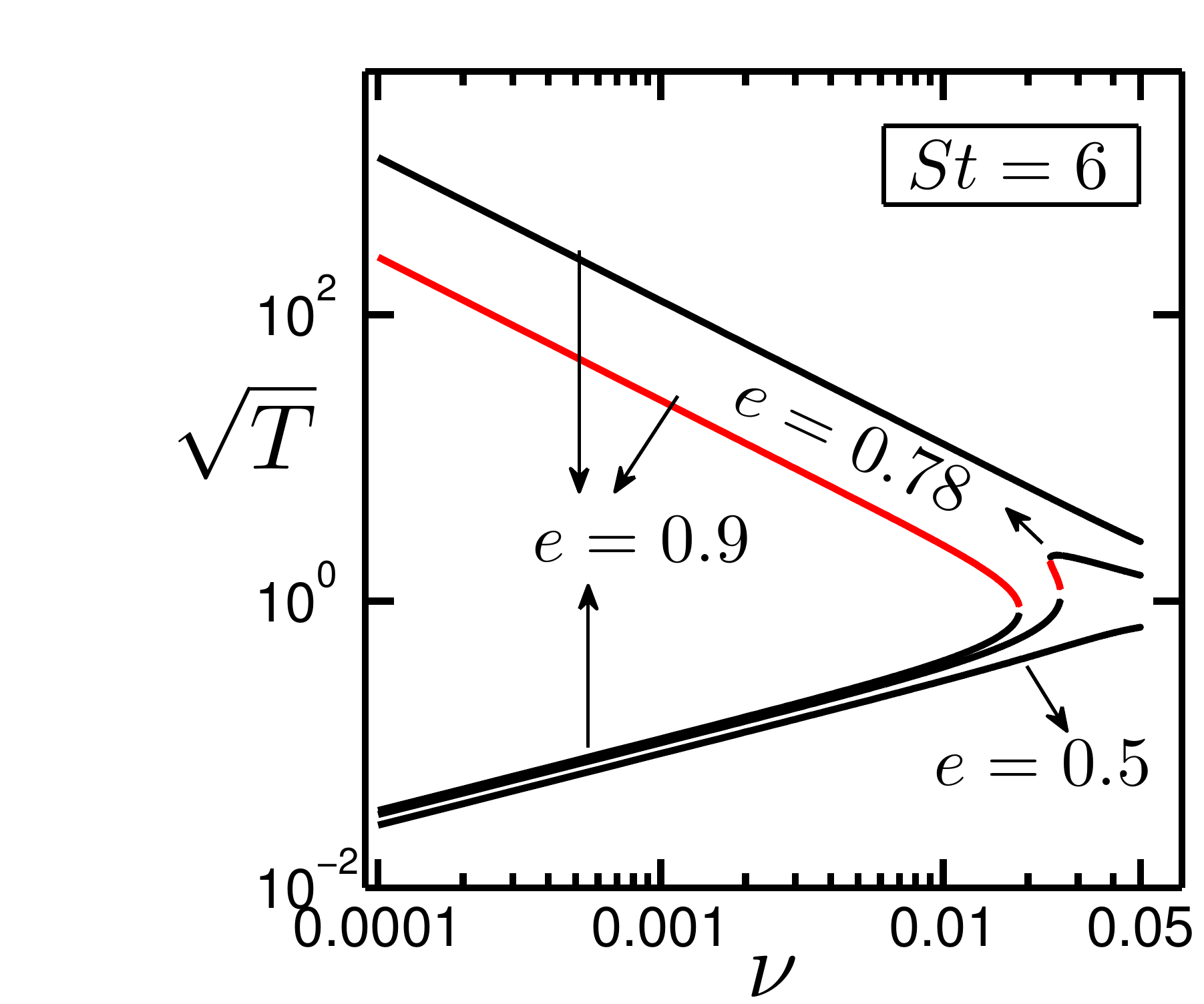}
  \caption{
  Disappearance of the ignited-state branch with (a,b,c) decreasing Stokes number at (a) $e=1$, (b) $e=0.8$ and (c) $e=0.5$,
  and (d) the same with decreasing restitution coefficient  at  $St=6$.
    }
\label{fig:fig4}
\end{figure}

\subsection{Critical Stokes numbers ($St_{c_1}, St_{c_2}$) and the master phase-diagram}
\label{sec:critical_points}

Referring to figure~\ref{fig:fig2},   two critical/limit points (at $St=St_{c_1}$ and $St_{c_2}$, with $St_{c_2}>St_{c_1}$) 
correspond to the double roots of (\ref{eqn:energy_balance}) at which  the following conditions must be satisfied:
\begin{equation}
 \mathcal{G}(\xi_c)=0
 \quad
 \mbox{and}
 \quad
 \mathcal{G}'(\xi_c)=0.
 \label{eqn:limitpoint1}
\end{equation}
This implies that  two solution branches, corresponding to two different states [(i) ignited $(T_{is})$, (ii) quenched $(T_{qs})$ and (iii) unstable $(T_{us})$]
meet at $\xi=\xi_c$, leading to saddle-node bifurcations  from one stable state to another stable state.

The  discontinuous ``$Q\to I$'' transition corresponds to a limit point ($St=St_{c_2}$, viz.~figure~\ref{fig:fig2}$a$)
at which the quenched and unstable solution branches meet.
Carrying out the asymptotic analysis of (\ref{eqn:energy_balance}) with $T_{qs}=T_{us}$ and satisfying (\ref{eqn:limitpoint1})  (see Appendix D for details),
we obtain the following  relation
\begin{equation}
 St_{c_2}^3\nu_c=\Bigg(\frac{3087000 \pi^2}{(1+e)^4 (107+193e)^2}\Bigg)^\frac{1}{3},
 \label{eqn:final_St_c}
\end{equation}
that represents a {\it critical-surface} in the ($\nu, St, e$)-plane, above which only the ignited state exists.
Equation (\ref{eqn:final_St_c}) is depicted in figure~\ref{fig:fig_PD1} as a blue-surface.
In the elastic limit of $e=1$, (\ref{eqn:final_St_c}) reduces to $St_{c_2}^3\nu_c=2.7685$
 which differs from  the prediction $(\approx 3.23)$ of \cite{TK1995}.

The critical Stokes number, $St_{c_1}$, for the ``$I\to Q$'' transition  (on decreasing $St$)
corresponds to the limit point at which $T_{is}=T_{us}$.
The asymptotic analysis of (\ref{eqn:energy_balance}) yields the following expression for  $St_{c_1}$ (see Appendix D for details):
\begin{equation}
 St_{c_1} \approx  9.9 - 4.91e,
 \label{eqn:critical_stokes_1}
\end{equation}
which  is marked as a brown-shaded plane in figure~\ref{fig:fig_PD1}, to the left of which only the quenched state exists.
For elastically colliding particles ($e=1$), we have  $St_{c_1}\approx 4.99$ which is close to our numerical solution of $4.94...$;
both are close to the result of $\sqrt{169.5/7}\approx 4.92$ obtained by \cite{TK1995}.
Note  that (\ref{eqn:critical_stokes_1})  depends  only  on the restitution coefficient, and therefore
the minimum value of Stokes number $(St_{c_1})$, below which only the quenched-state exists, is independent of the volume fraction 
 for a dilute gas-solid suspension.

\begin{figure}
 \begin{center}
  \includegraphics[scale=1.0]{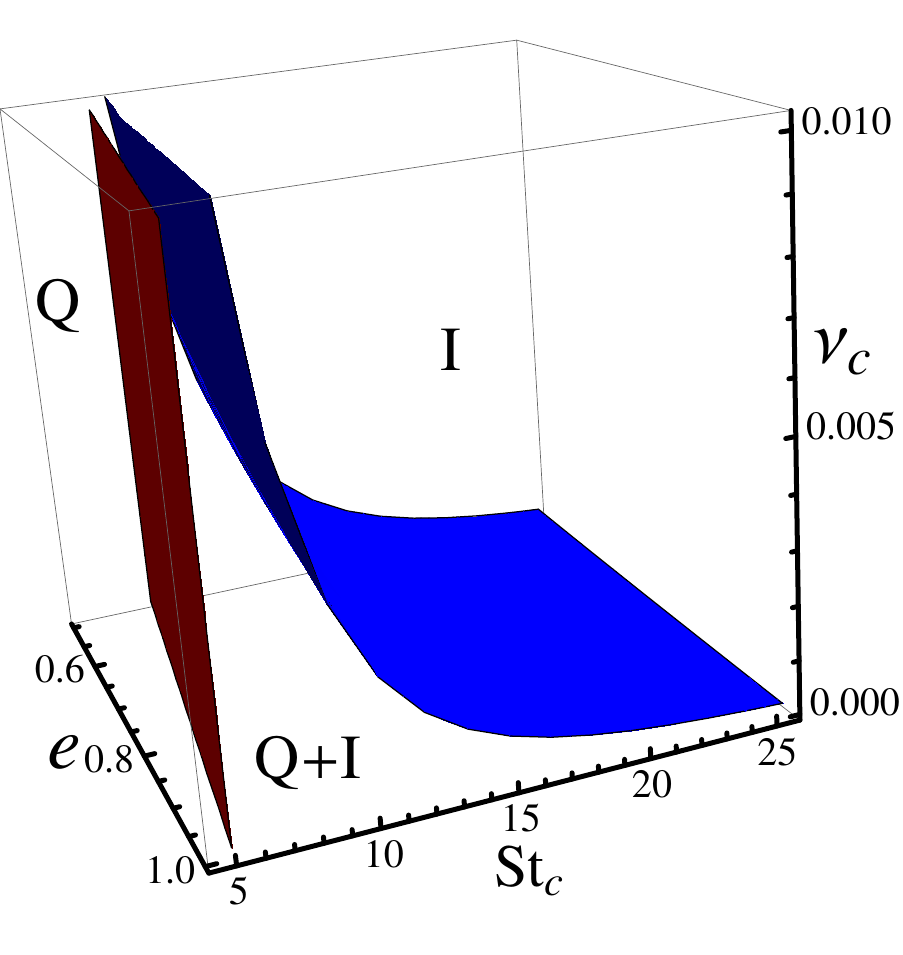}\\
 \end{center}
\caption{
Complete phase diagram of different states [``ignited'' (I), ``quenched'' (Q) and their coexistence (Q+I)] in the ($\nu, St, e$)-plane.
 The blue- and brown-colored  planes, above and below which, respectively, the ignited and quenched states  exist,
have been determined analytically from an ordering analysis of (3.1) in the dilute limit; for details, see the text  in \S3.3 and Appendix D.
}
\label{fig:fig_PD1}
\end{figure}

The master phase-diagram in figure~\ref{fig:fig_PD1} summarizes  all possible states  in the ($\nu, St, e$)-plane:
(i) the ignited state ($I$) exists above the blue-surface, (ii) the quenched state ($Q$) is the only solution to the left of the brown surface
and (iii) the coexistence of ignited and quenched ($I+Q$) states occurs for parameter values lying between the blue and brown surfaces.
Two critical surfaces in figure~\ref{fig:fig_PD1} would meet along a curve, 
thus acting as an upper bound for the  existence of the unstable state ($T_{us}$) solution (and hence the existence of the mixed state $I+Q$).
By equating $St_{c_1}=St_{c_2}$, the equation of this  curve is obtained as
\begin{equation}
     \nu_{us}^l(e)=\Bigg(\frac{3087000 \pi^2}{(1+e)^4 (107+193e)^2}\Bigg)^\frac{1}{3}/{(9.9-4.91e)^3},
 \label{eqn:nuc_Tus}
\end{equation}
which  is  a decreasing function of the restitution coefficient.
Note that (\ref{eqn:nuc_Tus}) is not a critical point, rather it represents an upper-bound on density below which the phase-coexistence [$I+Q$] occurs
in the  small-$St$ regime of a sheared gas-solid suspension.

It is clear from from (\ref{eqn:critical_stokes_1}) and (\ref{eqn:final_St_c}) that
the critical Stokes numbers $St_{c_1}$ and $St_{c_2}$ increase with decreasing $e$ (i.e.~ increasing inelasticity) at a fixed volume fraction $\nu<\nu_{us}^l$, 
When dissipative particles ($e\ll 1$)  collide with each other they loose more energy and hence loose more of their inertia; in that case the recovery time $(\tau_v)$ reduces and the adjustment with the local fluid velocity becomes faster, leading to the quenched state. On the other hand, for  nearly elastic ($e\sim 1$) collisions,
the particles lose very little  kinetic energy during collisions and take much more time to come back to the bulk flow and hence the recovery process becomes slow. 
Therefore, at higher values of $e$, both  ignited and quenched states exist but only the quenched state is possible if we increase inelasticity of the system, 
leading to the behaviour of $St_{c_1}$ as in (\ref{eqn:critical_stokes_1}).  Similar argument holds for the variation of $St_{c_2}$ with inelasticity as well.

\section{Non-Newtonian rheology:  second-moment anisotropy, discontinuous  shear-thickening and normal stress differences}

Once the temperature field is solved from  (\ref{eqn:energy_balance}) for specified values of $\nu$, $St$ and $e$,
the non-coaxiality  angle $\phi$, the temperature-anisotropy $\eta$ and
the  excess temperature $\lambda^2$ can be calculated from the remaining equations of (\ref{eqn:components_balance_second_moment}) --
these are amenable to analytical solutions as described in \S4.1. The behaviour of shear viscosity and  normal stress differences are analysed in \S4.2 and \S4.3, respectively.

\subsection{Anisotropies of second-moment tensor: analytical solution for  $\phi$, $\eta$ and $\lambda^2$}

After some algebra and rearrangement of terms in (\ref{eqn:components_balance_second_moment}),
the closed-form solutions for $\phi$, $\eta^2$ and $\lambda^2$ have been found:
\begin{eqnarray}
   \phi &=& \frac{1}{2}\tan^{-1}\left( \frac{2}{St} + \frac{12(1+e)(3-e)\nu \sqrt{T}}{5\sqrt{\pi}}\right)^{-1}, 
   \label{eqn:phi1}
   \\
   \eta^2 & =& -\frac{\mathfrak b}{2\mathfrak a} - \frac{1}{2\mathfrak a}\sqrt{{\mathfrak b}^2 - 4{\mathfrak a}{\mathfrak c}},
\label{eqn:eta1}
\\
\lambda^2  
&=& \frac{  \frac{5\sqrt{\pi}}{2St}T + (1+e)\nu T^{3/2}[5(1-e) - (5+3e)\frac{\eta^2}{14}] - \frac{8(1+e)^2\nu}{63\sqrt{\pi}} }
         {  \left( \frac{5\sqrt{\pi}}{St}T +{6(1+e)(3-e)\nu T^{3/2}} \right)},
\label{eqn:lambda1}
\end{eqnarray}
with $T$ being calculated from (\ref{eqn:energy_balance}) for specified values of $St$, $\nu$ and $e$.
The solution for the temperature-anisotropy $\eta$ follows from
the quadratic equation ${\mathfrak a}\eta^4 + {\mathfrak b}\eta^2 + {\mathfrak c} = 0$,
where
\begin{equation}
\left.
\begin{array}{lcl}
{\mathfrak a} &=& \frac{9(1-e^2)^2\nu^2 T^{3}}{25{\pi}}  >0
\\
{\mathfrak b} &=& \frac{6(1-e^2)\nu T^{3/2}}{5\sqrt{\pi}}\left(\frac{3}{St}T - \frac{16\nu(1+e)^2}{35\pi} +\frac{6(1-e^2)\nu T^{3/2}}{\sqrt{\pi}}\right) -  T^2 \cos^2{2\phi}
\\
{\mathfrak c} &=&\left(\frac{3}{St}T - \frac{16\nu(1+e)^2}{35\pi} +\frac{6(1-e^2)\nu T^{3/2}}{\sqrt{\pi}}\right)^2 >0
\end{array}
\right\} .
\end{equation}
For a suspension of elastically colliding particles ($e=1$, with finite $St$),  
we have ${\mathfrak a}=0$ and ${\mathfrak b}= - T^2 \cos^2{2\phi}$, and hence the above solutions  (\ref{eqn:phi1}-\ref{eqn:lambda1}) simplify to
\begin{equation}
\left.
\begin{array}{rcl}
   \phi(e=1) &=& \frac{1}{2}\tan^{-1}\left( \frac{2}{St} + \frac{48\nu \sqrt{T}}{5\sqrt{\pi}}\right)^{-1} > 0,
\label{eqn:phi_e1}
   \\
   \eta^2(e=1) &=& -\frac{{\mathfrak c}}{\mathfrak b} \equiv \left(\frac{3}{St}T - \frac{64\nu}{35\pi} \right)^2 T^{-2} \sec^2{2\phi} >0 ,
\label{eqn_eta_e1}
   \\
   \lambda^2(e=1)  &=&   \frac{  \frac{5\sqrt{\pi}}{2St}T  - \frac{8}{7}\nu T^{3/2} \eta^2 - \frac{32\nu}{63\sqrt{\pi}} }
                  {  \left( \frac{5\sqrt{\pi}}{St}T +{24\nu T^{3/2}} \right) } >0.
\label{eqn:lambda_e1}
\end{array}
\right\}
\end{equation}
Recall from (2.17)  that the non-zero values of  ($\phi, \eta, \lambda^2$) quantify the degree of anisotropy of the second-moment tensor $\mathsfb{M}$ (and hence
is a measure of the anisotropy of the kinetic stress tensor, $\mathsfb{P}=\langle \rho{\boldsymbol C}{\boldsymbol C}\rangle = \rho \mathsfb{M}$, too). 

\begin{figure}
 \begin{center}
 (a)
 \includegraphics[scale=0.26]{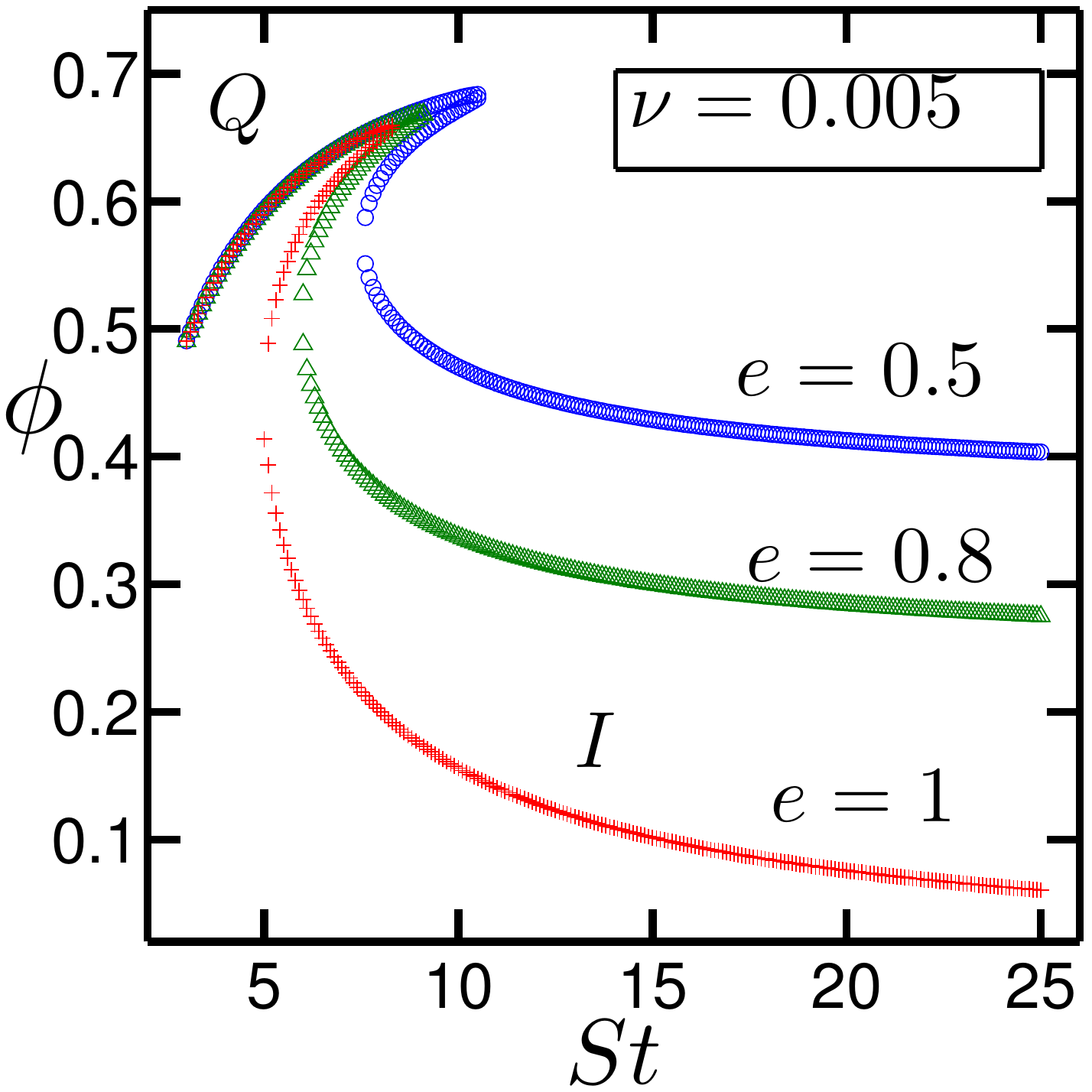}
  (b)
  \includegraphics[scale=0.26]{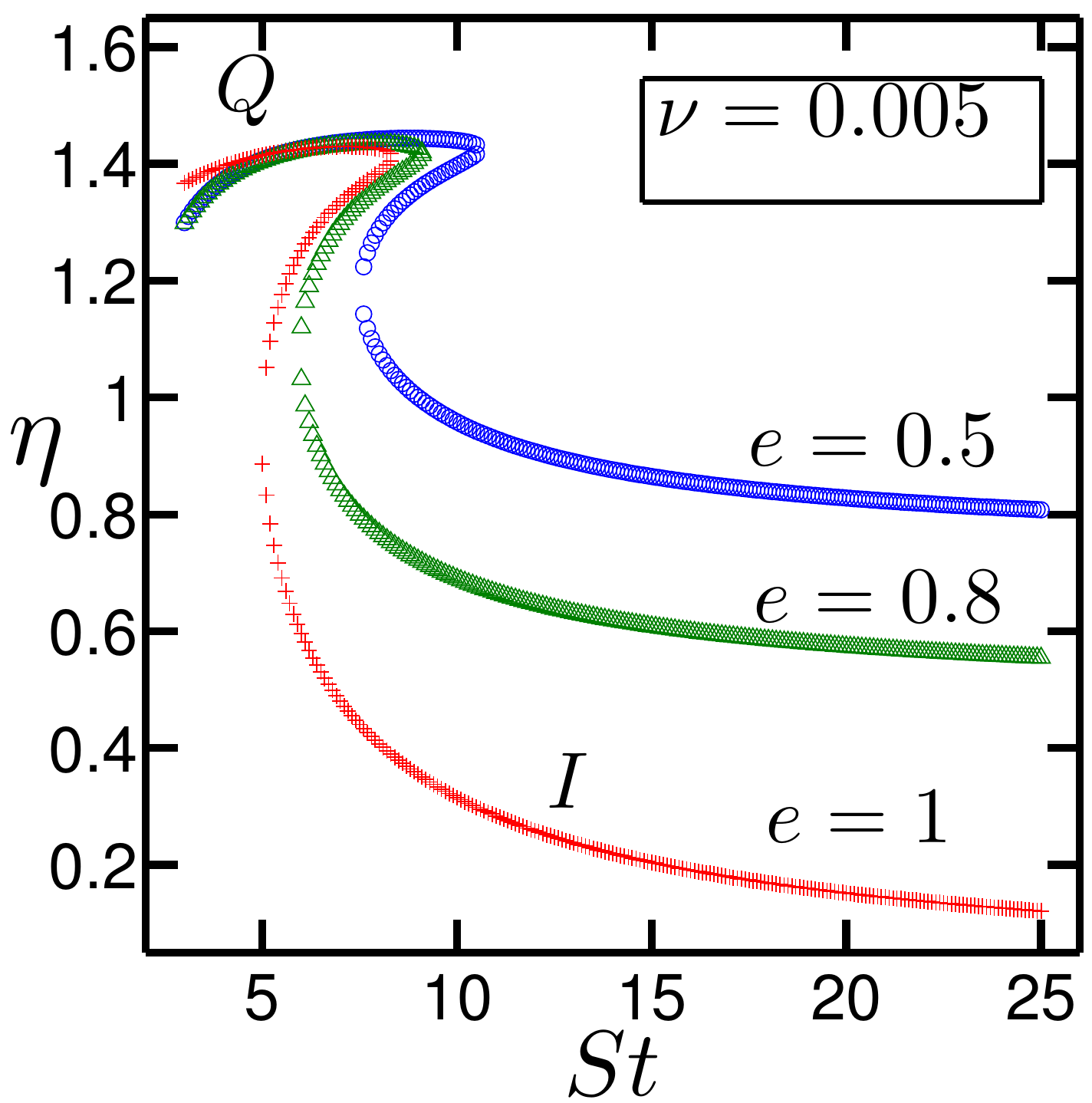}
  (c)
  \includegraphics[scale=0.26]{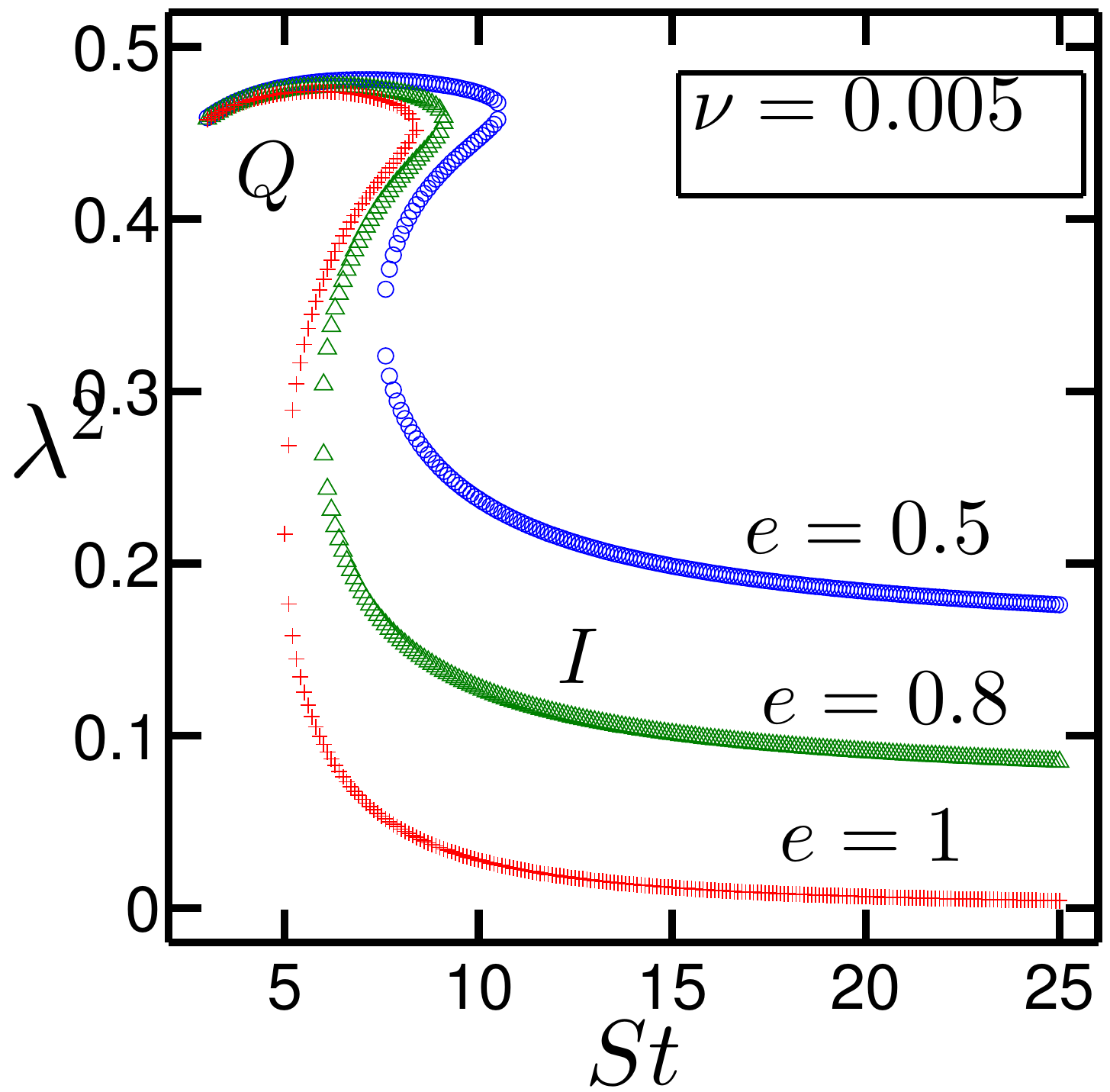}
\end{center}
\caption{
Variations of (a) the non-coaxiality angle $\phi$, (b) the shear-plane anisotropy $\eta$
and (c) the excess temperature $\lambda^2$ with Stokes number for different $e$.
 The mean volume fraction is set  to $\nu=0.005$.
  }
\label{fig:fig_phieta1}
\end{figure}

The positivity of (\ref{eqn:phi1}-\ref{eqn:lambda1})  is verified in  figures~\ref{fig:fig_phieta1}(a), \ref{fig:fig_phieta1}(b) and \ref{fig:fig_phieta1}(c), respectively,
which display the variations of $\phi$, $\eta$ and $\lambda^2$
with Stokes number for different values of  the  restitution  coefficient $e\leq 1$, at a mean volume fraction of $\nu=0.005$ -- the
results look qualitatively similar at other values of $\nu <\nu_{us}^{l}$ (\ref{eqn:nuc_Tus}).
It is seen from figure~\ref{fig:fig_phieta1} that  the increasing inelasticity markedly increases the values of ($\phi, \eta, \lambda^2$) on the ignited state,
thereby enhancing the anisotropy of the second-moment tensor.
In contrast, the inelasticity does not noticeably affect ($\phi, \eta, \lambda^2$)  on the quenched state
in which the particle collisions are rare and the dynamics is primarily dictated by fluid inertia.
Interestingly, increasing shear  makes the second-moment tensor more anisotropic on the quenched branch --
this can be understood by considering   the scaling relations of  ($\phi, \eta, \lambda^2$) at $St\sim 0$ as follows.
Using the closed-form solutions for three temperatures (3.2-3.4), the non-coaxiality angle for $e=1$ can be rewritten as
\begin{equation}
  \tan{2\phi_{qs}} = \frac{St}{2 +\frac{128\sqrt{2}}{5\sqrt{105}\pi}\nu^{3/2} St^{5/2}} \sim St/2 \quad \mbox{at} \quad St\sim 0.
\label{eqn:phiscaling-e1}
\end{equation}
Therefore, in the limit of small $St$, the inertia enhances the non-coaxiality angle in the quenched  state.
On the other hand, increasing $St$ decreases $\phi$ in the ignited state, reaching some asymptotic value  (depending on $e$)
at large enough $St$ as seen in figure~\ref{fig:fig_phieta1}(a).
This can be explained from an  analysis of the ignited branch solution, leading to:
\begin{equation}
\tan{2\phi_{is}} = \frac{3St}{6 + St^2} \sim \frac{3}{St},  \quad \mbox{for} \quad St\gg 1 .
\label{eqn:phiscaling-e2}
\end{equation}
Similar scalings (\ref{eqn:phiscaling-e1}-\ref{eqn:phiscaling-e2}) hold for the temperature anisotropy $\eta$ and the excess temperature $\lambda^2$ too,
that explain the observed behaviour in figures~\ref{fig:fig_phieta1}(b) and \ref{fig:fig_phieta1}(c), respectively.
In summary, the degree of anisotropy of  the second-moment tensor in the quenched and ignited states is
 primarily dictated by the background shear and inelasticity, respectively.
The latter effect of inelasticty can be understood from following scaling arguments.

It may be  noted that the scaling relation (\ref{eqn:phiscaling-e2}) is  not strictly valid at $St\to\infty$
since the double-limit of $e\to1$ and $St\to\infty$ leads to a singular behaviour of temperature $T\to\infty$
(and hence a thermostat is necessary to achieve a steady shearing state of elastically colliding particles in the absence of fluid drag).
The case of a sheared granular gas ($St=\infty$ at $e\neq 1$) has been analysed previously~\citep{JR1988,Richman1989,SA2014,SA2016}; 
it can be verified that the above solutions (\ref{eqn:phi1}-\ref{eqn:lambda1}) for the ignited-branch reduce to the low-density solution of  \cite{SA2016}:
\begin{equation}
\left.
\begin{array}{rcl}
   \lambda^2 &\approx&  \frac{1}{48e}\left(168 + 53(1-e)\right) \left[  \sqrt{1 + 5760e(1-e)(168+53(1-e))^{-2}} -1  \right] \\
    &\approx& \frac{5}{14}(1-e)\left(1+\frac{53}{168}(1-e)\right)\left(1- \frac{53}{84}(1-e)\right) \\
    \eta^2 &=& \frac{3\lambda^2(7+6\lambda^2)}{6+\lambda^2} \approx \frac{7}{2}\lambda^2 = \frac{5}{4}(1-e)\left(1+\frac{53}{168}(1-e)\right)\left(1- \frac{53}{84}(1-e)\right)  \\
     \sin{2\phi} &=& \frac{\eta}{1+\lambda^2} \approx \eta \sim \sqrt{1-e}\\
  \sqrt{T} &=& \frac{5\sqrt{\pi}\eta \cos{\phi}}{3(1-e^2)\nu (10+\eta^2) } 
     \approx \frac{\sqrt{\pi}}{6(1-e^2)\nu  } \eta (1-\eta^2/10)(1-\eta^2/2)\\
     &\approx& \frac{\sqrt{\pi}}{6(1-e^2)\nu  } \eta (1-\frac{3}{5}\eta^2) \sim (1-e)^{-1/2}
  \end{array}
  \right\}
  \label{eqn:scaling1}
\end{equation}
Therefore, in the limit  ($St\to\infty$) of a granular gas, $\eta\sim\lambda\sim \sin{2\phi}\sim \sqrt(1-e)$,
with the granular temperature diverging  like $T\sim(1-e)^{-1}$ -- the latter finding rules out the possibility of the quenched-state solution in a sheared granular gas.
The scaling relations (\ref{eqn:scaling1}) hold  at leading-order in $\sqrt{1-e}$ for $St\gg 1$,
and therefore we conclude that the inelasticty enhances the degree of anisotropy of ${\mathsfb M}$ on the ignited branch, see figure~\ref{fig:fig_phieta1}.

\begin{figure}
 \begin{center}
(a)
\includegraphics[scale=0.6]{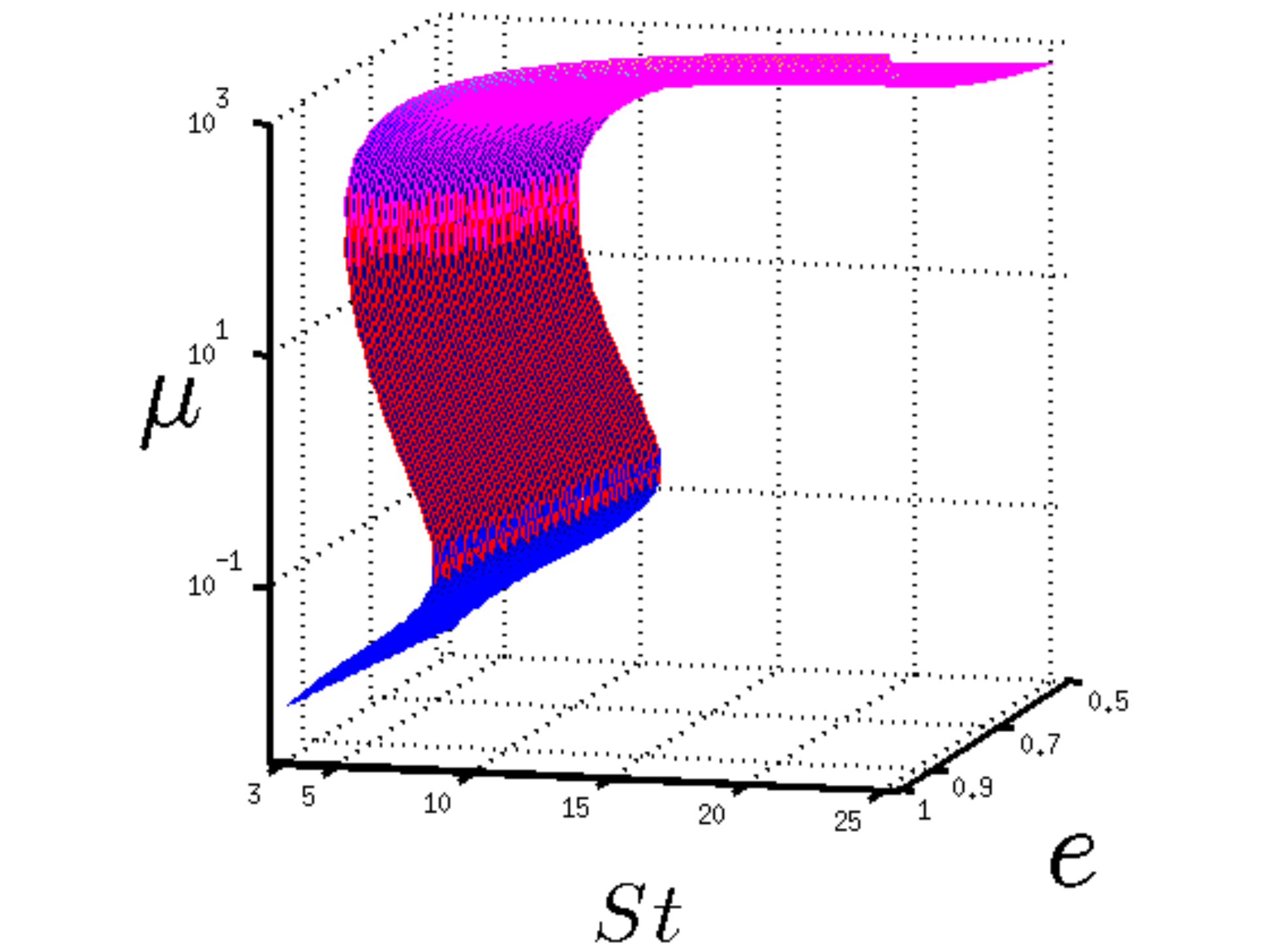}\\
 (b)
 \includegraphics[scale=0.38]{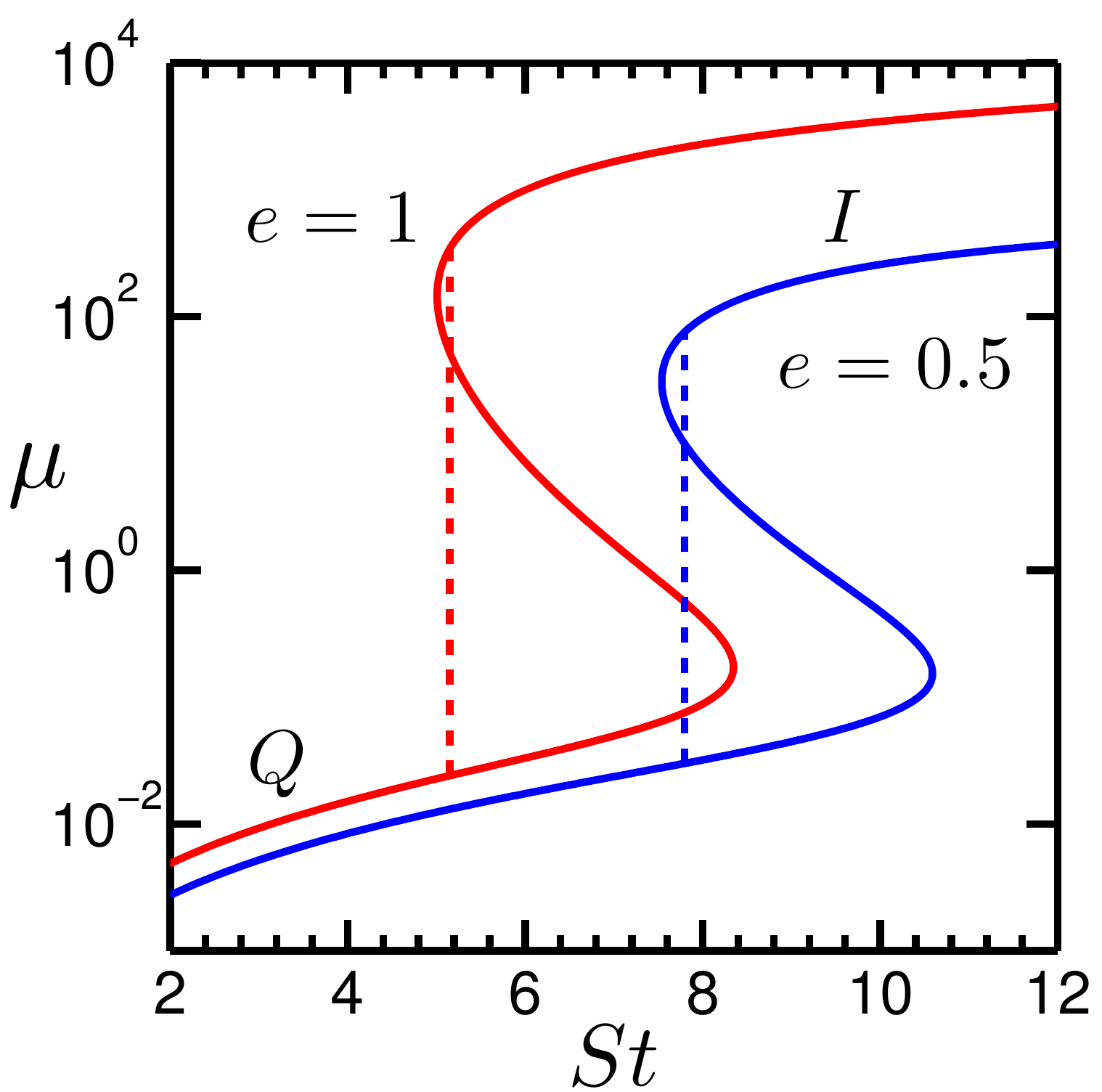}
(c)
\includegraphics[scale=0.38]{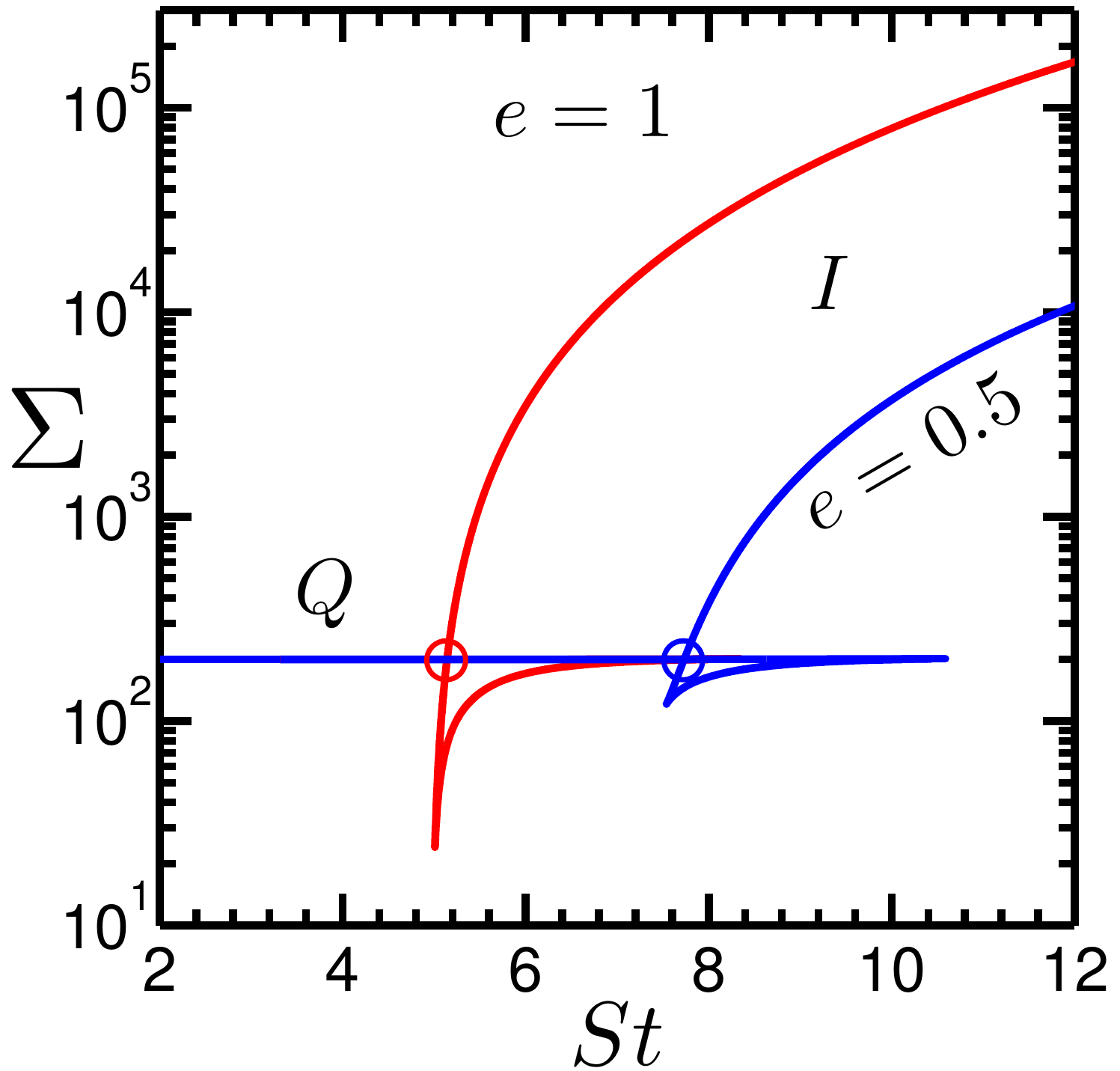}
\end{center}
\caption{
 (a)  Hysteretic behaviour of particle-phase  viscosity ($\mu$) as functions of ($St, e$)  for a volume fraction of  $\nu=0.005$;
  this represents DST (discontinuous shear-thickening) behaviour for any $e$ at  $\nu<\nu_{us}^l$ [(3.8)].
 (b) Viscosity versus $St$ for $e=1$ (red line) and $e=0.5$ (blue line); the vertical dotted lines represent the coexistence-point,
  marked by circles in panel $c$,  at which two states $I$ and $Q$ coexists with each other.
 (c) An effective Massieu function (\ref{eqn:Sigma1}), with parameter values as in panel $b$; see the text for details.
  }
\label{fig:fig_mu}
\end{figure}

\subsection{Shear viscosity: continuous and discontinuous shear-thickening (DST)}

The dimensionless shear viscosity for the particle phase is given by
\begin{eqnarray}
 \mu & =& -\frac{P_{xy}}{\rho_p\nu(\dot\gamma \sigma/2)^2}=\eta\cos(2\phi)T \nonumber\\
   &\equiv& \frac{3}{St}T - \frac{16\nu(1+e)^2}{35\pi} + \frac{3(1-e^2)\nu T^{3/2}}{5\sqrt{\pi}} (10 + \eta^2),
     \label{eqn:viscosity1}
   \\
   &\stackrel{St\to\infty}{\equiv} &- \frac{16\nu(1+e)^2}{35\pi} + \frac{3(1-e^2)\nu T^{3/2}}{5\sqrt{\pi}} (10 + \eta^2) > 0,
   \qquad \forall \quad e<1 .
   \label{eqn:viscosity2}
\end{eqnarray}
For  the ignited-state solution only (i.e.~$\aleph\equiv \aleph^{is}$), it can be verified that the shear viscosity for elastically colliding particles ($e=1$)
is  $ \mu = {3T}/{St}$ which represents the first term in (\ref{eqn:viscosity1}).

The variation of (\ref{eqn:viscosity1}) as functions of ($St, e$) is depicted in figure~\ref{fig:fig_mu}(a) for  particle volume fraction of $\nu=0.005$.
Similar to  granular temperature, the shear viscosity undergoes hysteretic jumps at $St=St_{c_2}$ (``$Q\to I$'') and $St_{c_1}$ (``$I\to Q$'')
on increasing and decreasing $St$, respectively.
The effect of dissipation ($e<1$) is to reduce the viscosity of the particle-phase in each state, see figure~\ref{fig:fig_mu}(b).
On the other hand, the effect of Stokes number can be understood by considering the  viscosity of elastically colliding ($e=1$) particles 
as given by
\begin{eqnarray}
\mu_{is} &\approx&   \frac{75\pi}{20736}\frac{St}{\nu^2} ,
\quad
\mu_{qs} \approx  \frac{384}{945\pi}\nu St^2 ,
\quad
\mbox{and} 
\quad
\mu_{us} \approx \frac{147\pi}{25}\nu^{-2} St^{-7},
\end{eqnarray}
in the ignited, quenched and unstable states, respectively.
Clearly, two shear-thickening branches ($Q$ and $I$) are connected via a shear-thinning  branch.

The `discontinuous shear thickening' (DST) behaviour,  such as in figure~\ref{fig:fig_mu}(a,b),
occurs only in the small Stokes-number limit of a dilute gas-solid suspension at $\nu<\nu_{us}^l$, (\ref{eqn:nuc_Tus}), for  any restitution coefficient. 
The middle-branch in figure~\ref{fig:fig_mu}(a,b), over which $\mu$ decreases with increasing $St$ (i.e.~the shear-thinning branch), is unstable.
This is a thermodynamic/constitutive instability which can be understood from a phenomenological viewpoint.
Let us calculate the following quantity,
\begin{equation}
  \Sigma(\dot\gamma) = \int_{\dot\gamma_R}^{\dot\gamma} \mu(\dot\gamma) \dot\gamma {\rm d}\dot\gamma + \Sigma_R,
  \label{eqn:Sigma1}
\end{equation}
which is a measure of the stress work and the reference value $\Sigma_R$ is added to make $\Sigma(\dot\gamma)$ positive definite.
The variation of (\ref{eqn:Sigma1}) is plotted against $St$ in figure~\ref{fig:fig_mu}(c) for $e=1$ (red line) and $0.5$ (blue line).
For each case, the upper-most envelope in figure~\ref{fig:fig_mu}(c) represents the stable solution, and the intersection between the ignited and quenched
branches represent the coexistence point at which both states coexist with each other.
The latter point is marked by vertical dashed lines in figure~\ref{fig:fig_mu}(b) -- this also follows from the well-known Maxwell's equal-area rule.
For the present problem, the effective shear work (\ref{eqn:Sigma1}) behaves like a Massieu  function~\citep{Callen1985} for the selection of  the 
`coexisting'  solution branch, although it must be noted that the choice of (\ref{eqn:Sigma1}) is not unique.
For example, if we choose to probe the jump in dynamic friction, $\mu/p$,
the location of the coexisting branch gets slightly shifted (not shown).
A proper  identification of a Massieu/entropy  function, or, a thermodynamic potential for the present
sheared suspension may require a stability analysis of the underlying moment equations subject to uniform shear flow, which is left to a future work.

\begin{figure}
  \begin{center}
  (a)
  \includegraphics[scale=0.35]{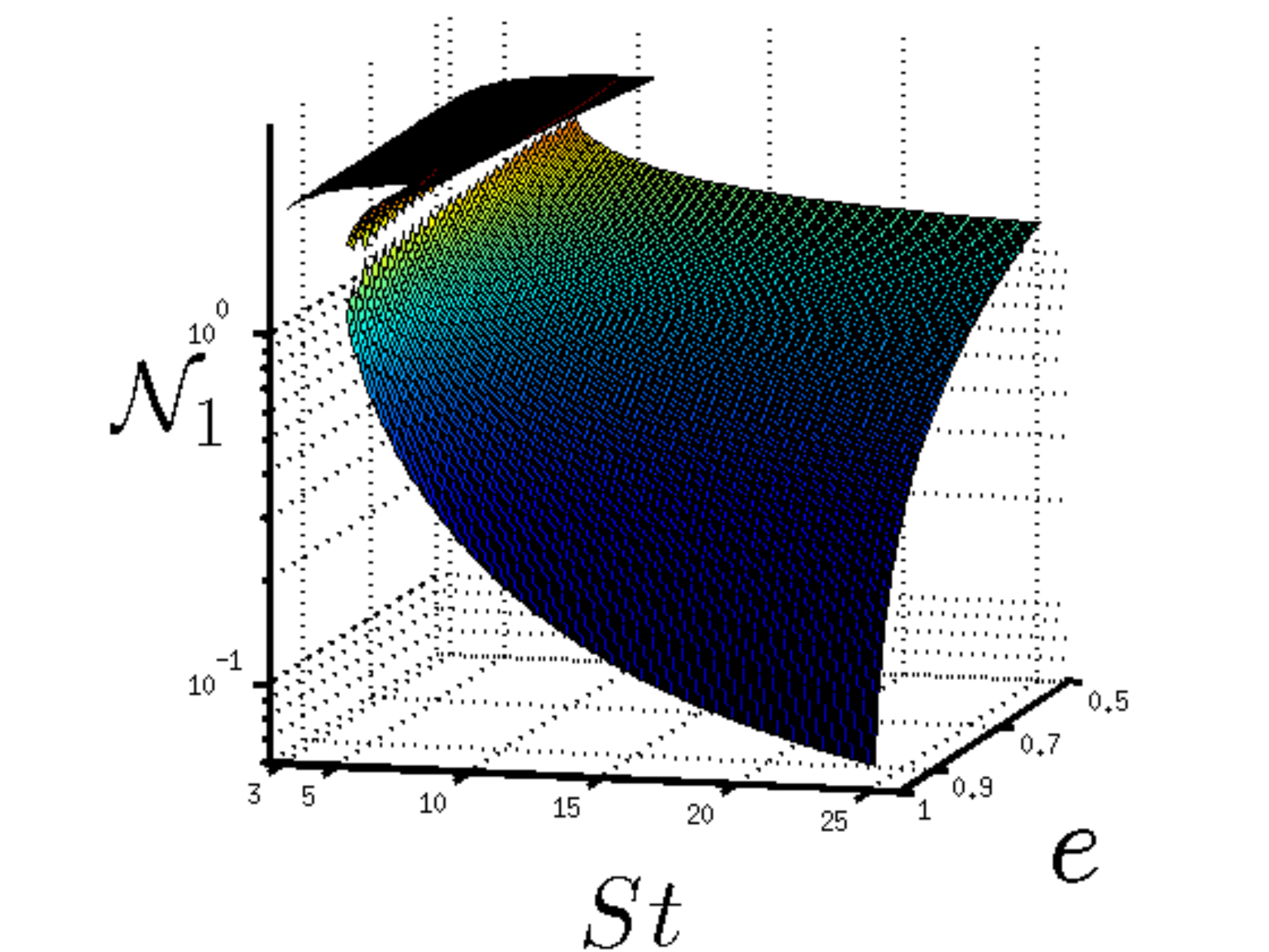}
  (b)
  \includegraphics[scale=0.36]{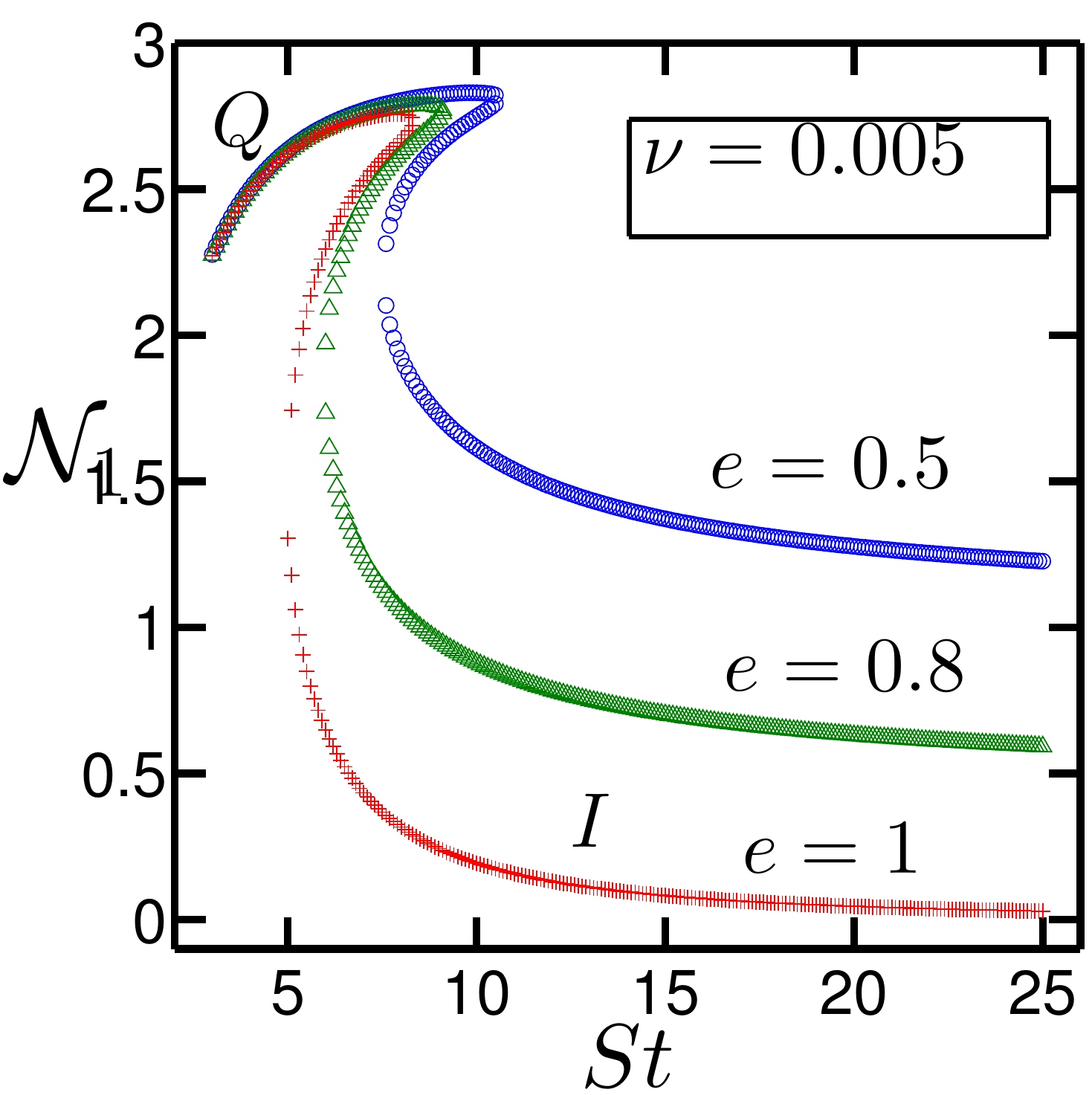}
\end{center}
  \caption{
  Variations of the first $({\mathcal N}_1)$ normal stress differences against Stokes number ($St$) and restitution coefficient ($e$) at  $\nu=0.005$.
   In panel $b$  the projection of panel $a$ is displayed for  different $e$.
   }
\label{fig:fig_N1}
\end{figure}

In the area of liquid-solid suspensions, the shear-thickening and its discontinuous analog 
are well-known since the original work of \cite{Hoffman1972}. There have been a renewed research activity to understand  the origin of
DST in the ``dense'' regime of colloidal and non-colloidal suspensions as well as in dense granular media~\citep{BJ2014,DM2014}.
Extending the present theoretical formalism to the dense regime of suspensions, by incorporating frictional interactions and
 related physics~\citep{Seto2013,Fernandez2013,WC2014,CBMF2017}, would be an interesting future work.
We became aware of a recent work that uses  gas kinetic theory~\citep{HT2016} in the context of a dilute ``thermerlized'' granular gas,
and  their finding on DST as a ``saddle-node'' bifurcation is similar to the present findings~\citep{SA2016a} -- however, they did not refer
to the work of \cite{TK1995} from which the present work follows. How a thermalized granular gas is related to present system of a gas-solid suspension needs 
to be investigated.

\subsection{First and second normal stress differences}

The expression for the first normal stress difference is
\begin{eqnarray}
{\mathcal N}_1 &=& \frac{P_{xx}-P_{yy}}{p} = 2\eta\sin{2\phi} ,
 \label{eqn:N1}
 \end{eqnarray}
 which has been `scaled' by the mean pressure $p=(P_{xx} + P_{yy} + P_{zz})/3$;
 in (\ref{eqn:N1}), $\phi$ and $\eta$ are calculated   from (\ref{eqn:phi1}) and (\ref{eqn:eta1}), respectively.
 The variation of (\ref{eqn:N1}) as functions of ($St, e$) is displayed in figure~\ref{fig:fig_N1}(a,b).
 The quenched-branch ${\mathcal N}_1$ remains unaffected by inelasticity (see panel b), however, on the ignited branch, increasing inelasticity increases ${\mathcal N}_1$;
 the effect of the gas-phase (i.e. decreasing $St$) also increases the ignited branch ${\mathcal N}_1$.
 On the whole, the dependence of ${\mathcal N}_1$ on both $St$ and $e$ mirrors that of the non-coaxiality angle ($\phi$) and the
shear-plane  temperature-anisotropy ($\eta$), compare  figure~\ref{fig:fig_N1}(b) with figure~\ref{fig:fig_phieta1}(a,b).
It is clear from (\ref{eqn:N1}) that the origin of the first normal stress difference is tied to 
the shear-plane anisotropies ($\eta$ and $\phi$) of the second-moment tensor as in the case of a sheared granular gas~\citep{JR1988,SA2016} --
the dependence of $St$ on its origin remains implicit via two anisotropy parameters ($\phi, \eta$).
 
\begin{figure}
  \begin{center}
  (a)
  \includegraphics[scale=0.4]{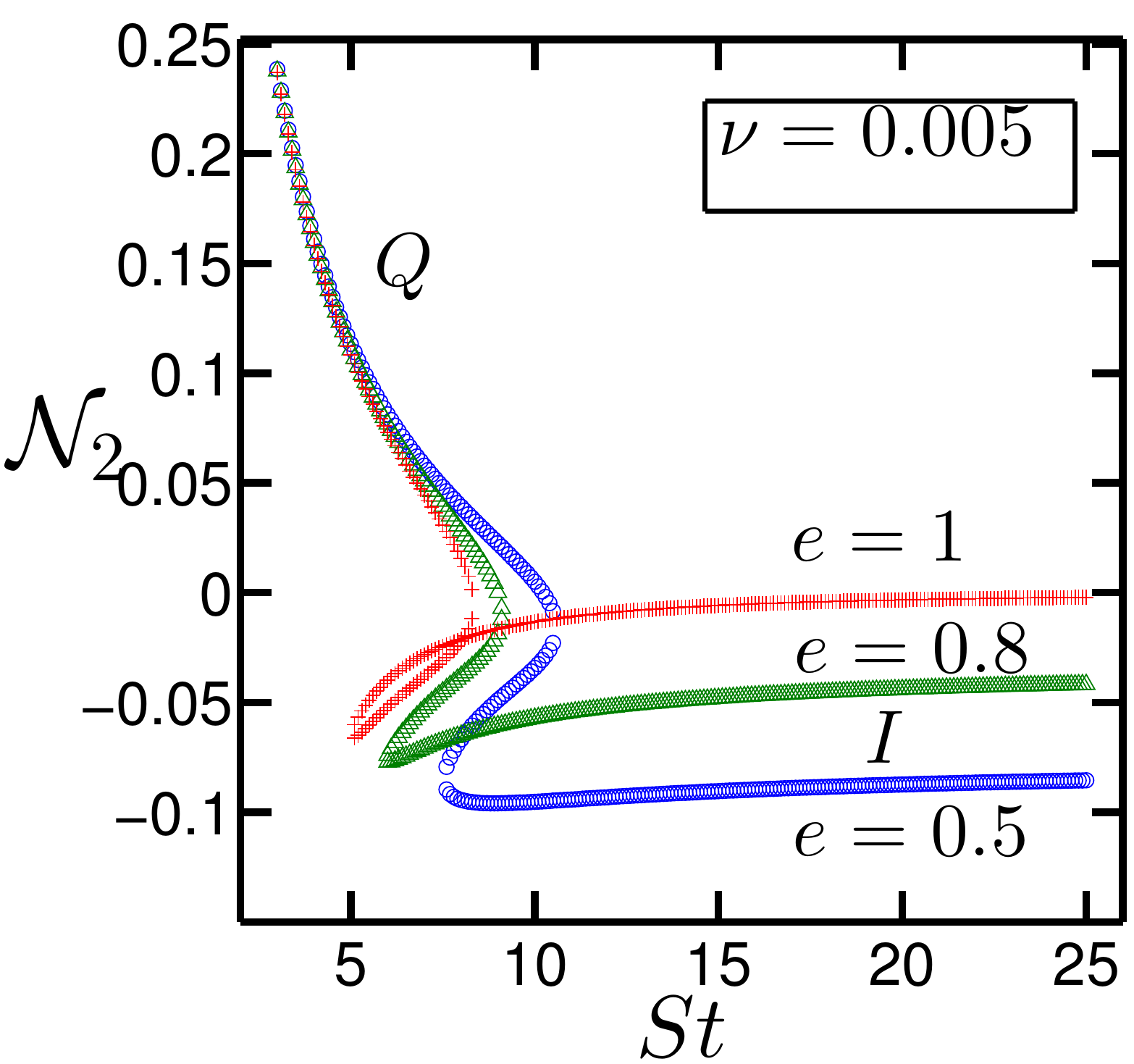}
   (b)
  \includegraphics[scale=0.4]{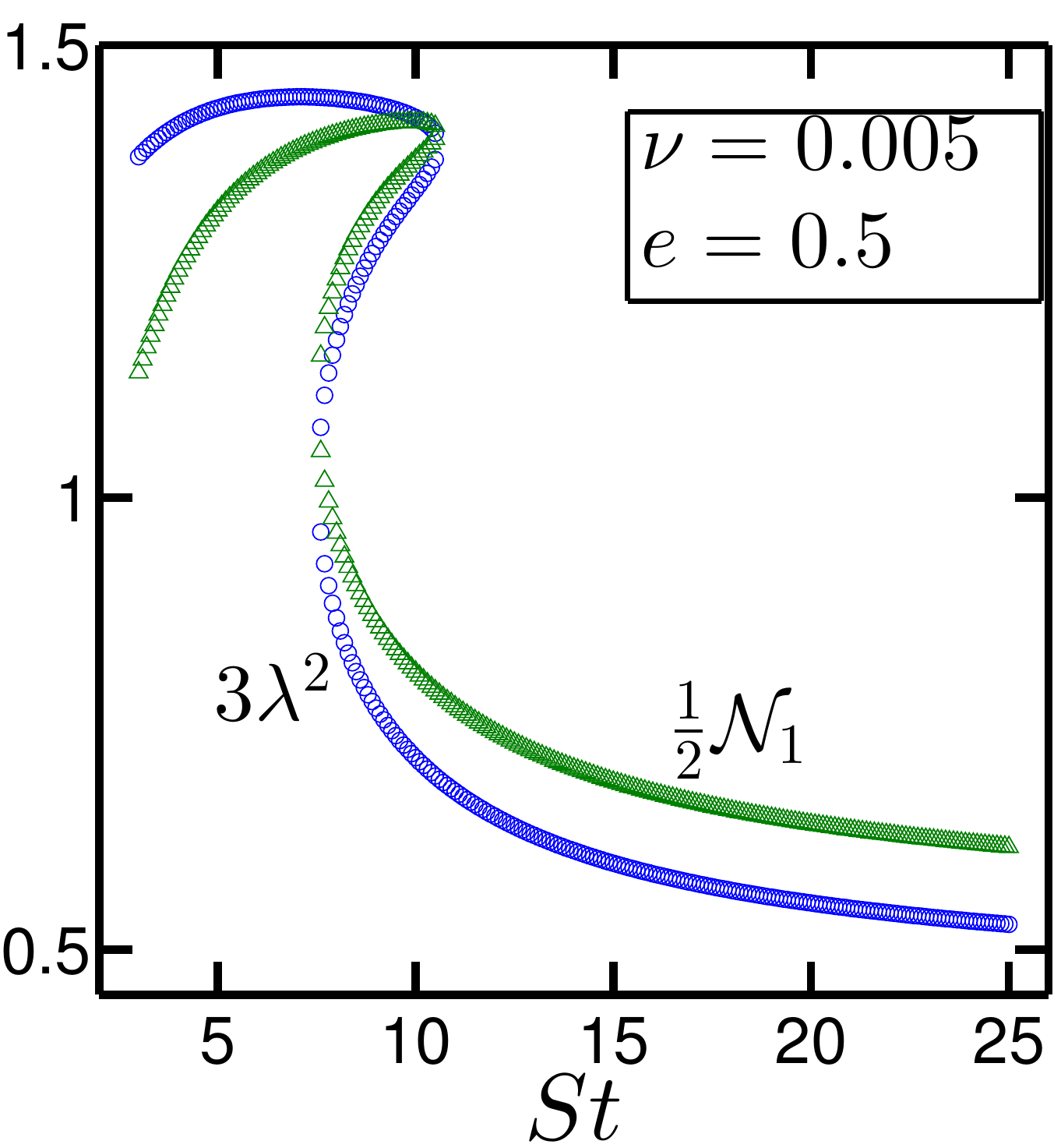}\\
   (c)
   \includegraphics[scale=0.4]{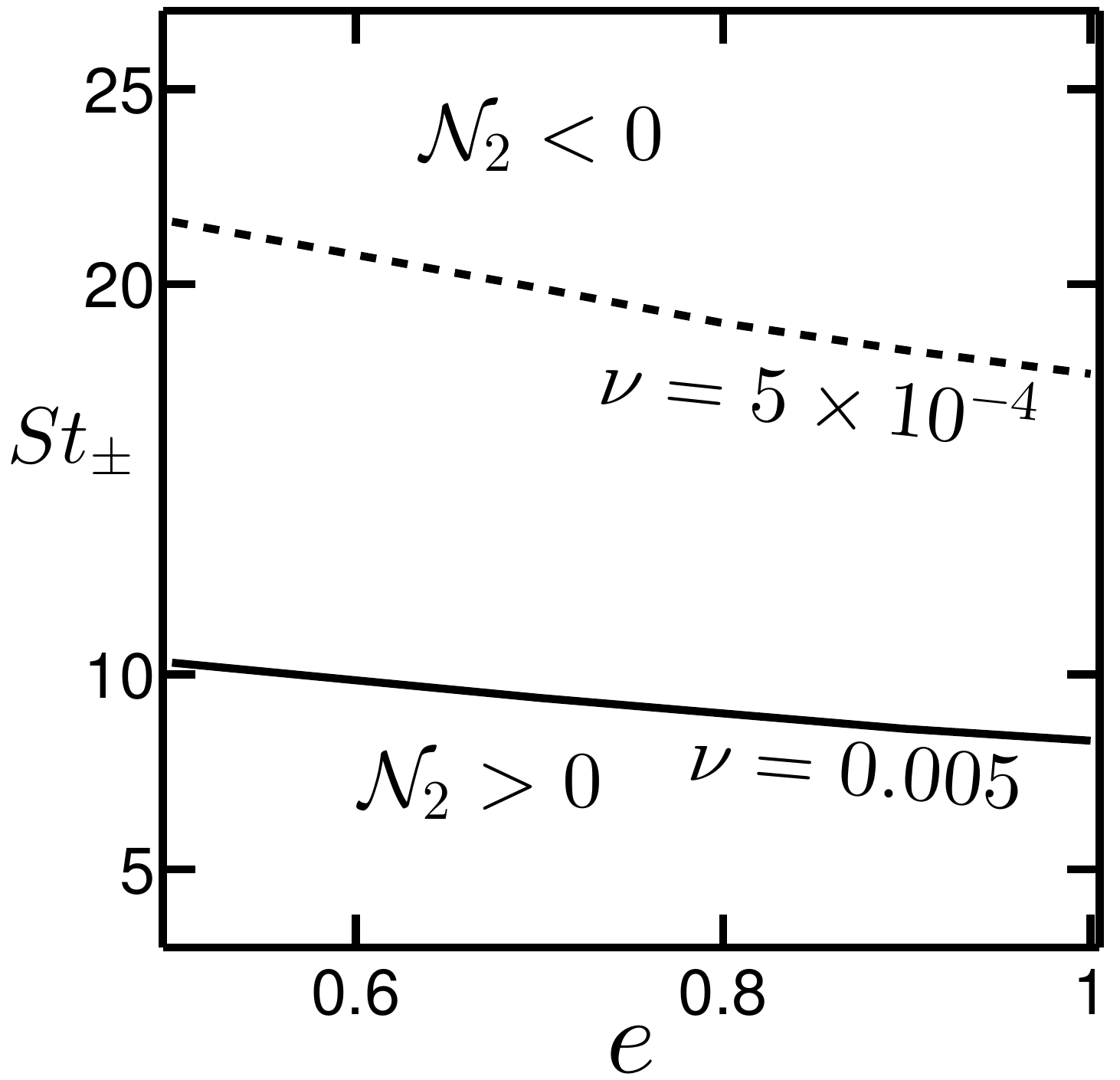}
  \end{center}
  \caption{
 (a)  Variation of the second $({\mathcal N}_2)$ normal stress difference against Stokes number $(St)$ for different values of the restitution coefficient;
 the particle volume fraction is  $\nu=0.005$.
 (b) Variations of $3\lambda^2$ (blue circles) and ${\mathcal N}_1/2$ (green triangles) with $St$ for $e=0.5$, with other parameters as in panel $a$.
 (c) Variations of the critical Stokes number $St_{\pm}$ (at which  ${\mathcal N}_2=0$) with $e$ for $\nu=0.005$ (solid line) and $\nu=0.0005$ (dashed line).
  }
\label{fig:fig_N2}
\end{figure}
 
The  scaled second normal stress difference is given by
\begin{eqnarray}
   {\mathcal N}_2 &=& \frac{P_{yy}-P_{zz}}{p} 
   = 3\lambda^2  -  \eta\sin{2\phi} = 3\lambda^2  - \frac{1}{2}{\mathcal N}_1.
   \label{eqn:N2}
  \end{eqnarray}
The variation of (\ref{eqn:N2}) with $St$ is shown in figure~\ref{fig:fig_N2}(a) for different values of the restitution coefficient $e$.
Similar to ${\mathcal N}_1$, the effect of inelasticity is to increase the magnitude of the second normal stress difference on the ignited branch,
but the quenched-branch  ${\mathcal N}_2$ remains  unaffected (expectedly) by changing $e$.
It is noteworthy in  figure~\ref{fig:fig_N2}(a) that ${\mathcal N}_2$ is positive and negative in the quenched and ignited states, respectively.
This sign-change can be understood from figure~\ref{fig:fig_N2}(b) which display the variations of two terms in (\ref{eqn:N2}) with $St$.
In the quenched state the excess temperature ($3\lambda^2\propto T_z^{ex}$) dominates over the shear-plane anisotropies ($\eta\sin{2\phi}\equiv {\mathcal N}_1/2$),
whereas the latter dominates over the former in the ignited state,
resulting in the sign-change of ${\mathcal N}_2$ at some finite value of $St$.

The parameter combinations ($St, e, \nu$) at which ${\mathcal N}_2$ undergoes   sign-reversal can be  calculated by solving the following equation
\begin{equation}
   {\mathcal N}_1 - 6\lambda^2 =0,
\end{equation}
along with (\ref{eqn:N1}) and (\ref{eqn:lambda1}).
Figure~\ref{fig:fig_N2}(c) shows the variation of $St_{\pm}$ with restitution coefficient: ${\mathcal N}_2$ is positive and negative, respectively,
below and above  each line for a specified density $\nu$.
It is seen that the effect of inelastic dissipation is to increase the critical value of  $St_{\pm}$ at which ${\mathcal N}_2$  changes its sign;
reducing the mean-density increases   $St_{\pm}$ at any $e$.

It may be noted that for a `dense' sheared granular gas ($St\to\infty$), 
the second normal-stress difference undergoes sign-change \citep{AL2005,SA2016} at some critical density ($\nu_{\pm}\sim 0.2$),
with ${\mathcal N}_2$ being negative and positive in the dilute and dense limit, respectively;
the competition between (i) the collisional anisotropies in a dense system (that makes  the particle-motion increasingly streamlined~\citep{AL2005}
with increasing density) and (ii) the second-moment anisotropies ($\phi, \eta, \lambda^2$)  is known to be responsible for this  sign-change~\citep{SA2016}.
For the present case of a `dilute' suspension, 
 the behaviour of ${\mathcal N}_2$ in the quenched state resembles that in a sheared `dense' granular fluid;
 this could possibly be due to the `streamlined'  particle motion in both systems, characterizing the underlying anisotropy.

\section{Discussion: Comparison with  Grad's moment-expansion (GME)}

Recall that in figure~\ref{fig:fig2}, we have made a detailed comparison between the predictions of two moment theories:
(i) the standard Grad's moment-expansion (GME) around a Maxwellian~\citep{Grad1949,TK1995,Sangani1996,CRG2015}
using Hermite polynomials and (i) the  present anisotropic-Maxwellian moment-expansion (AME).
 Overall,  the AME predictions for granular temperature 
are found to be  more accurate  (see insets in figure~\ref{fig:fig2}) than that of GME, especially at lower values of restitution coefficient,
via a  comparison with available simulation data.
This conclusion holds for shear viscosity too (not shown) since $\mu\propto \sqrt{T}$ -- in the following we focus on the
predictive abilities of the present theory (AME) with reference to two normal-stress differences.
(The reader is referred to  \cite{SA2014} for details on AME that has been used to derive a generalized Fourier law for heat-flux vector,
along with conductivity tensors; the heat-flux, however, vanishes in uniform shear flow as in the present case.)

\subsection{Suspension of elastic and inelastic hard spheres: ${\mathcal N}_1$ and ${\mathcal N}_2$}

From the present AME theory, the normal stress differences for elastic ($e=1$) hard-sphere suspensions  in the ``ignited'' state are given by (Appendix A)
\begin{equation}
 \mathcal{N}_1 = \frac{18}{6 + \Omega St^2}
 \quad
  \mbox{and}
  \quad
-  \mathcal{N}_2 = \frac{9(9 +\Omega St^2)\Omega St^2}{(6+\Omega St^2)\left[252 + {87}\Omega St^2 + {7}\Omega^2 St^4\right]}>0.
  \label{eqn:N12-AME}
\end{equation}
with
\begin{equation}
 \label{dilute_omega1}
   \Omega = \frac{1}{2 St^2}\left[\left(St^2 - \frac{171}{7}\right) 
       +  \left( \left(St^2 - \frac{3}{7}\right)^2 - (12\sqrt{2})^2 \right)^{1/2}\right].
\end{equation}
The last quantity $\Omega$  is positive for $St> St_{c_1}=\sqrt{171/7}$ (the critical Stokes number for ``ignited-to-unstable'' transition, viz.~eqn.~(3.7)),
and asymptotically approaches unity, $\Omega(St\to\infty) =1$, and hence $\Omega\in (0,1)$ at any $St>St_{c_1}$.

 The AME-predictions (\ref{eqn:N12-AME})  can be compared with the  corresponding GME predictions
 for ${\mathcal N}_1$ and ${\mathcal N}_2$:
\begin{equation}
  {\mathcal N}_1 = \frac{18}{6 + \Theta St^2}
  \quad
  \mbox{and}
  \quad
 -  {\mathcal N}_2 =  \frac{\frac{9}{14}\Theta }{6 + \Theta St^2} >0,
 \label{eqn:N12-SME}
  \end{equation}
  where
  \begin{equation}
  \Theta = \frac{1}{2St^{2}}\left[ \left(St^2 - \frac{171}{14}\right) +  \left( \left(St^2 - \frac{171}{14}\right)^2 - 12^2 \right)^{1/2} \right].
  \label{eqn:Theta}
  \end{equation}
  In (\ref{eqn:N12-SME}) that there is a minor correction in the expression for ${\mathcal N}_2$: the numerical factor $9/14$ in the numerator 
   was taken as $9/7$ in \cite{TK1995}.
  The positivity of (\ref{eqn:Theta}) follows from the positivity of its discriminant, resulting in $St> St_{c_1}=\sqrt{169.5/7}$,
 which is very close to  $\sqrt{171/7}$ for the positivity of (\ref{dilute_omega1}).
It is worth pointing out that the functional dependence of both (\ref{dilute_omega1}) and (\ref{eqn:Theta}) yields
almost identical values for $\Omega$ and $\Theta$ at any $St>St_{c_1}$.

\begin{figure}
\begin{center}
\includegraphics[scale=0.5]{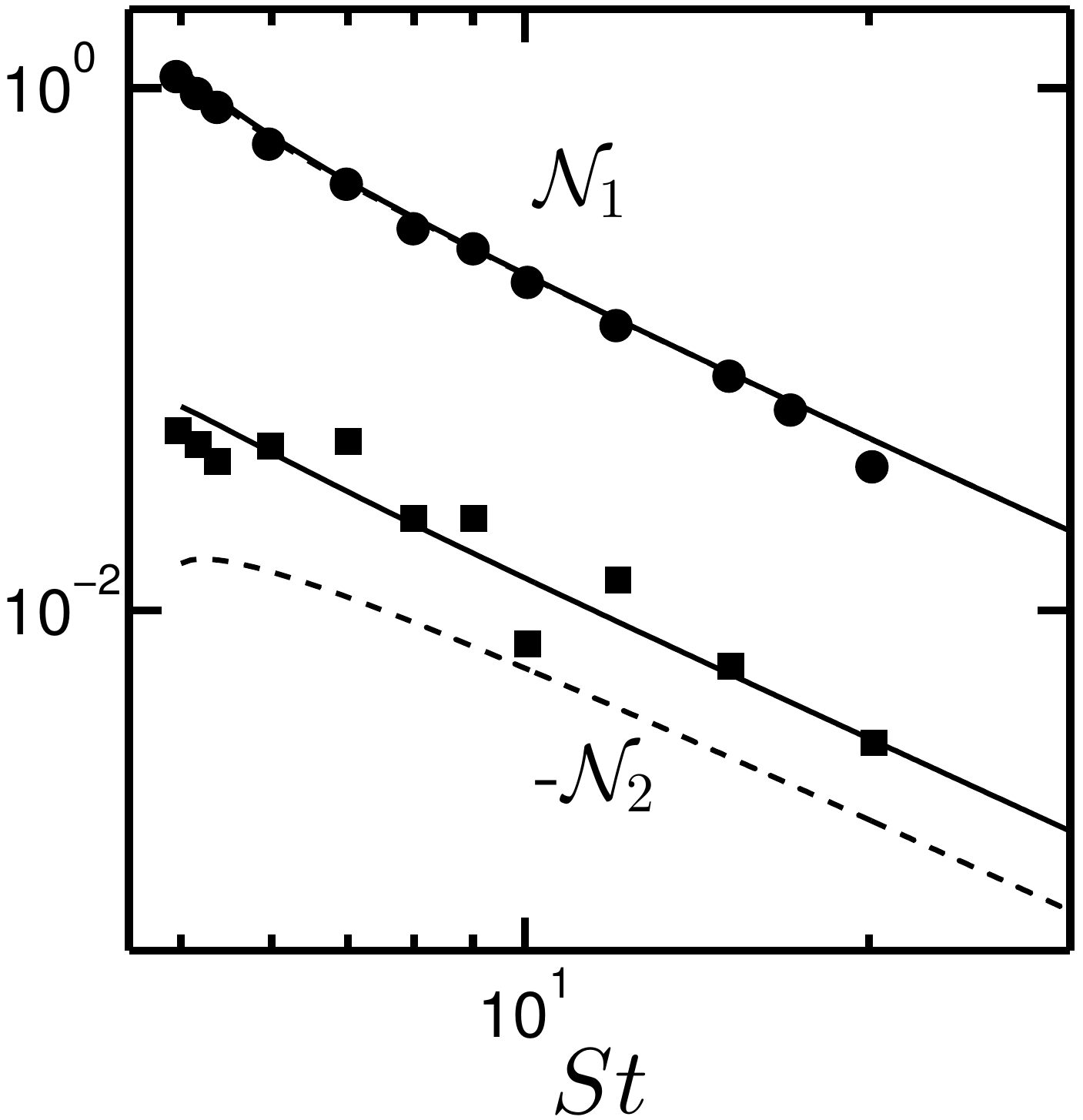}
\end{center}
 \caption{
 Variations of  the first (circles) and second (squares) normal-stress differences with Stokes number  for a suspension of elastic ($e=1$) hard-spheres --
the particle volume fraction is $\nu=0.01$, representing a `dilute' suspension.
The solid lines represent the present theory (\ref{eqn:N12-AME})   and the dashed lines represent the standard Grad's  moment theory (\ref{eqn:N12-SME});
the  DSMC simulation data  \citep{TK1995} are denoted by symbols.
 }
\label{fig:fig-N12-com}
\end{figure}

Figure~\ref{fig:fig-N12-com} shows a comparison of (\ref{eqn:N12-AME}) (denoted by solid lines) for ${\mathcal N}_1$ and ${\mathcal N}_2$
with (i) the DSMC simulation data (symbols) of \cite{TK1995} and (ii) the GME theory (\ref{eqn:N12-SME}) (dashed lines) --
the particle volume fraction is set  to $\nu=0.01$, representing a `dilute' gas-solid suspension.
 It is seen that both (\ref{eqn:N12-AME}) and (\ref{eqn:N12-SME}) predict the correct behaviour of  ${\mathcal N}_1$ -- two theories are almost indistinguishable from each other,
 with excellent quantitative agreement with simulation.
 However, there is a significant disagreement  (by a factor of about $2$) between  (\ref{eqn:N12-SME}) and  the DSMC data for the second normal-stress difference ${\mathcal N}_2$;
in contrast, the predictions of AME (\ref{eqn:N12-AME}) are uniformly good for both ${\mathcal N}_1$ and ${\mathcal N}_2$ over a range of Stokes number.

\begin{figure}
\begin{center}
(a)
\includegraphics[scale=0.38]{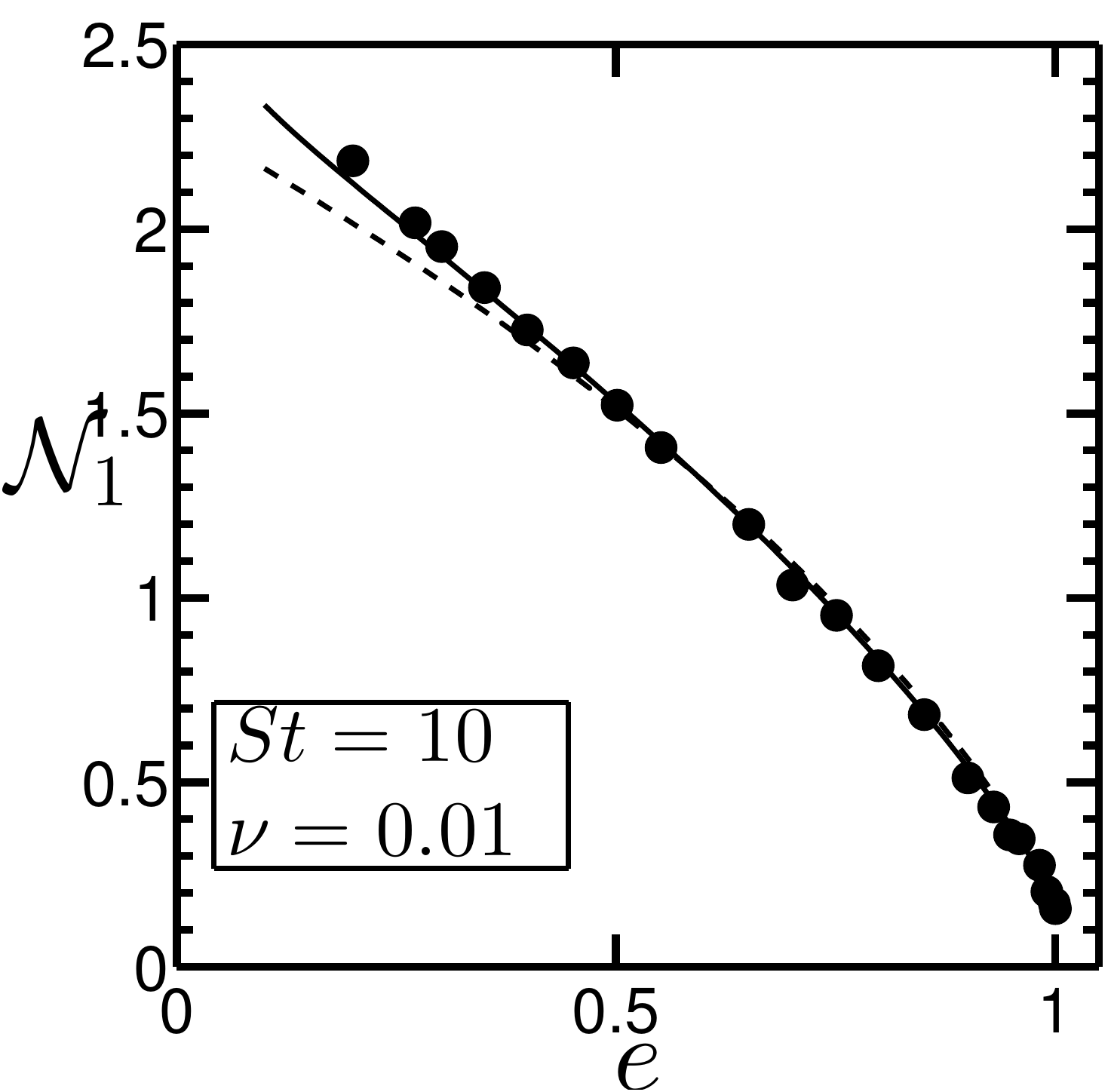}
 (b)
 \includegraphics[scale=0.38]{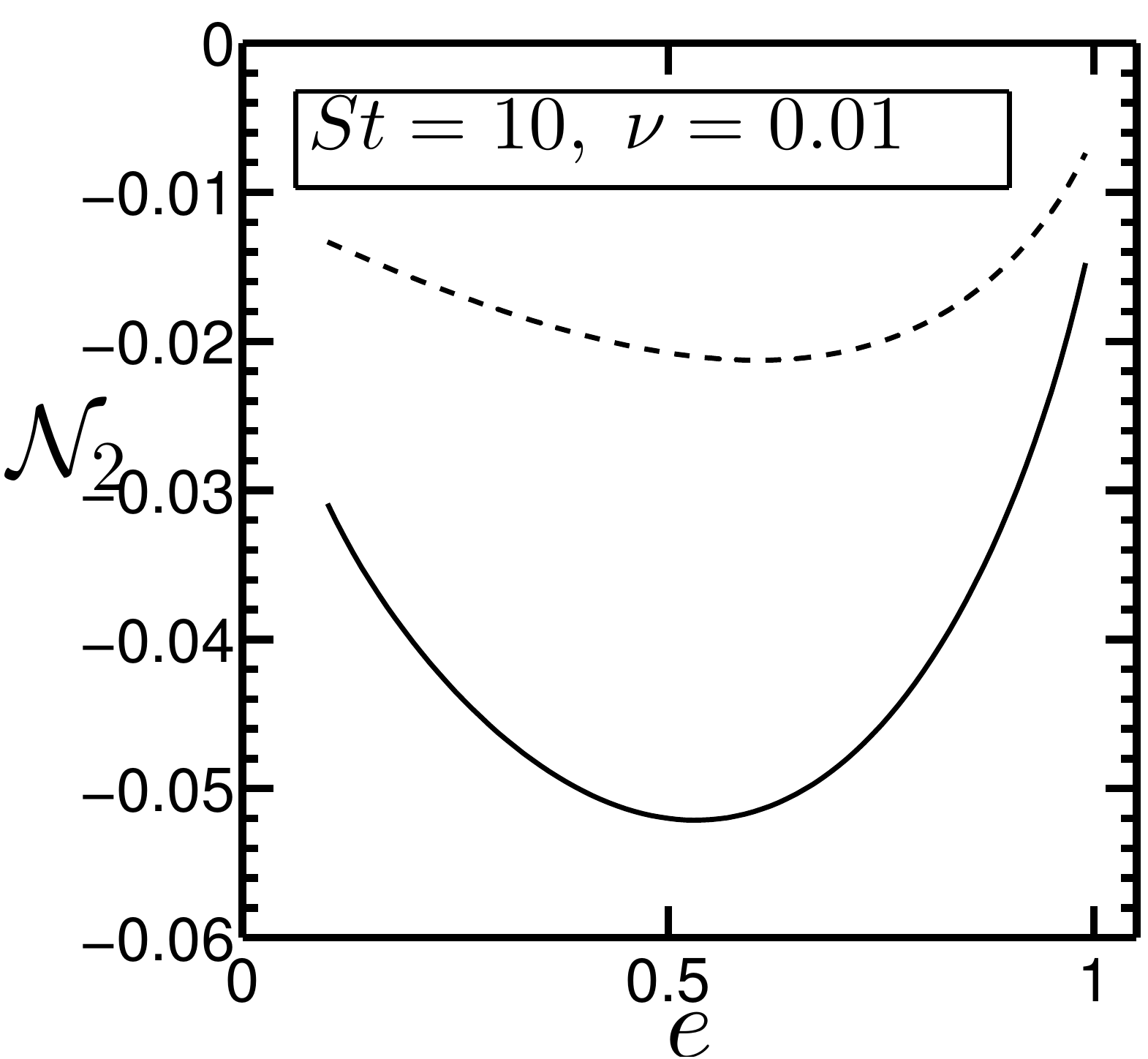}
\end{center}
 \caption{
 Comparisons of   (a) first  and (b) second normal-stress differences at $St=10$:
 (i) DSMC simulation (filled circles, \cite{Sangani1996}), (ii) present  theory (solid lines),
(iii) the standard Grad's moment expansion [dashed lines,  see Appendix E].
 }
\label{fig:fig-N12-com1}
\end{figure}

  It may be noted that in GME the quadratic nonlinear-terms (proportional to $P_{\alpha\beta}^2$)  need to be
   taken into account while evaluating the source term  $\aleph_{\alpha\beta}$ (\ref{eqn:source_second_moment_usf}) in order to obtain  `non-zero' 
   second normal-stress difference as  suggested by \cite{HH1982} for a Boltzmann (dilute) gas.
   A brief account of the related analysis for a gas-solid suspension of inelastic particles is provided in Appendix E  --
  the resulting expressions for ${\mathcal N}_1$ and ${\mathcal N}_2$  reduce   to (\ref{eqn:N12-SME}) 
   for elastically-colliding particles. On the other hand, the analysis of \cite{Sangani1996} did not include such
   nonlinear Grad-terms, resulting in ${\mathcal N}_2=0$;  the recent work of \cite{CRG2015}  also confirmed that
   the nonlinear Grad-terms are necessary for   ${\mathcal N}_2\neq 0$.
 It has been verified that the quadratic non-linear terms do not noticeably affect the value of ${\mathcal N}_1$ as well as  the shear viscosity.

 The effect of inelasticity on ${\mathcal N}_1$  and ${\mathcal N}_2$  can be ascertained from  figures~\ref{fig:fig-N12-com1}(a)
 and ~\ref{fig:fig-N12-com1}(b), respectively,  for a suspension with small Stokes number ($St=10$);
 other parameters are as in figure~\ref{fig:fig-N12-com}.
 It is clear from panel $a$ that the present predictions of ${\mathcal N}_1$ (solid line)  agree well with simulation data
 for the whole range of $e$, but the GME-predictions (dashed and dot-dashed lines) are slightly lower at $e<0.5$. 
 On the other hand, the GME theory grossly under-predicts (by a factor of $3$) the value of ${\mathcal N}_2$ for dissipative particles,
 see  figure~\ref{fig:fig-N12-com1}(b).

\subsection{From sheared suspension to `dry' ($St\to\infty$) granular gas}

To further understand  the predictions of  normal stress differences (${\mathcal N}_1$ and ${\mathcal N}_2$) from two theories (GME and AME) 
for dissipative particles ($e<1$), we  focus on the uniform shear flow of 
a dilute  granular gas ($St\to\infty$) --  the  molecular-dynamics (MD) simulations of inelastic hard-spheres with Lees-Edward boundary conditions
have  been carried out for a range of restitution coefficients $e\in(1, 0.3)$ at a particle volume fraction of $\nu=0.01$; a
relatively small system with $N=1000$ particles was simulated-- other simulation details can be found in \citep{AL2005,GA2008}. 
From these simulations, it is easy to extract data on two anisotropy parameters, namely, (i) the shear-plane temperature anisotropy $\eta$ [see (\ref{eqn:eta0})]
and (ii) the excess temperature $T^{ex}_z/T=\lambda^2$ [see (\ref{eqn:Texcess1})], which are marked by filled-circles in 
figures~\ref{fig:fig-Texcess-com}(a) and \ref{fig:fig-Texcess-com}(b), respectively.
In each panel, the theoretical predictions of \cite{SA2016} are shown by solid lines.
Overall, there is excellent agreement between AME theory and MD simulation.

\begin{figure}
\begin{center}
(a)
\includegraphics[scale=0.4]{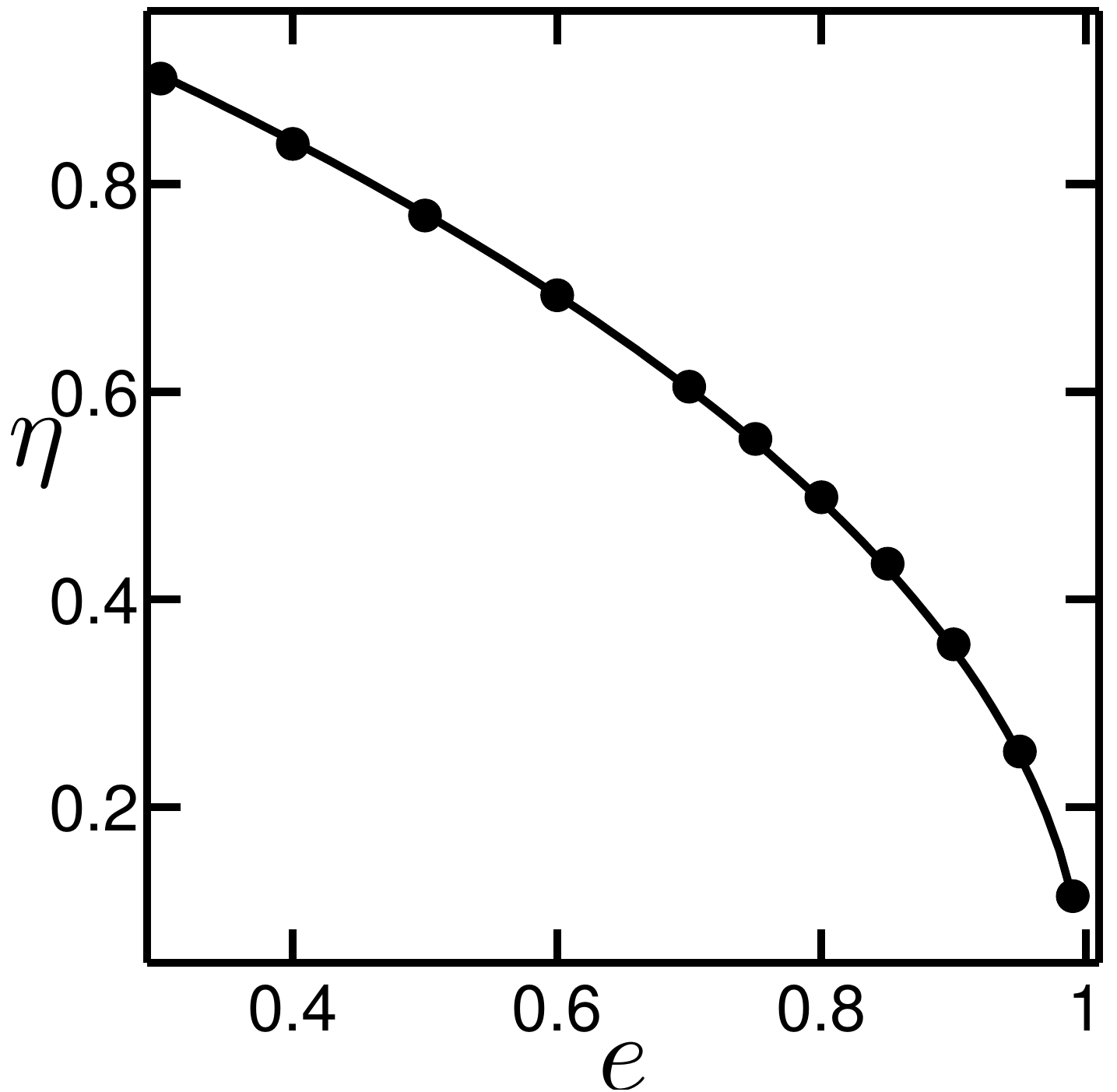}
(b)
\includegraphics[scale=0.4]{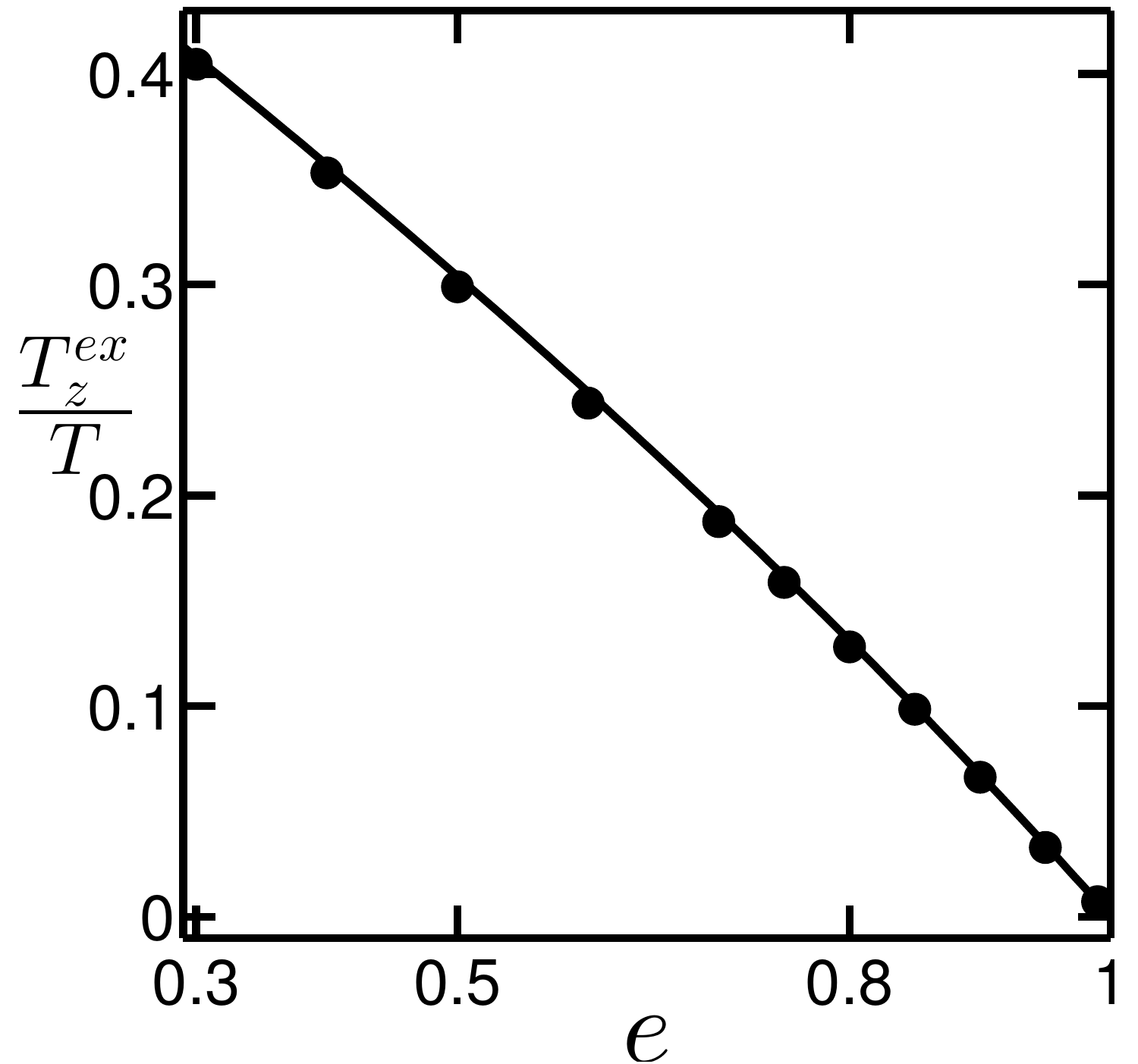}
\end{center}
 \caption{
 Comparisons of   (a) shear-plane temperature-anisotropy $\eta$ and (b)  the excess temperature $T^{ex}_z/T\equiv 2\lambda^2$
  in uniform shear flow of a granular gas ($St=\infty$):
  MD simulation (symbols) and   theory [solid line,~\cite{SA2016}].
  The particle volume fraction is $\nu=0.01$ and the number of particles is $N=1000$ in simulations.
 }
\label{fig:fig-Texcess-com}
\end{figure}

Figures~\ref{fig:fig-N12-com2}(a) and ~\ref{fig:fig-N12-com2}(b) compare the  MD simulation data (symbols) for ${\mathcal N}_1$ and ${\mathcal N}_2$,
respectively, with theory;
the AME predictions, denoted by solid lines, are calculated  from (\ref{eqn:N1}) and (\ref{eqn:N2}) by setting $St\to \infty$~\citep{SA2016},
 and the corresponding GME-predictions (Appendix E) are denoted by  dashed lines.
In addition, the dot-dashed line in each panel represents the super-Burnett-order  solution  of \cite{SG1998}, obtained from  the Chapman-Enskog expansion
of inelastic Boltzmann equation.
It is clear that both GME and AME theories predict almost the same value of ${\mathcal N}_1$ for a range of restitution coefficient $e\in (0.3, 1)$,
but the GME-prediction for ${\mathcal N}_2$ is consistently lower than that of AME and can be off by a factor of $3$ at $e=0.3$.
On the other hand, the AME-predictions for both ${\mathcal N}_1$ and ${\mathcal N}_2$ are comparable to those of Chapman-Enskog solution for $e\geq 0.8$,
but the latter becomes increasingly inaccurate  for $e<0.8$.
Therefore, the quantitative predictions of the AME  for two normal stress differences
are better than those of GME and Chapman-Enskog solution -- this overall conclusion holds for both gas-solid and dry granular suspensions of inelastic particles.

\begin{figure}
\begin{center}
(a)
\includegraphics[scale=0.38]{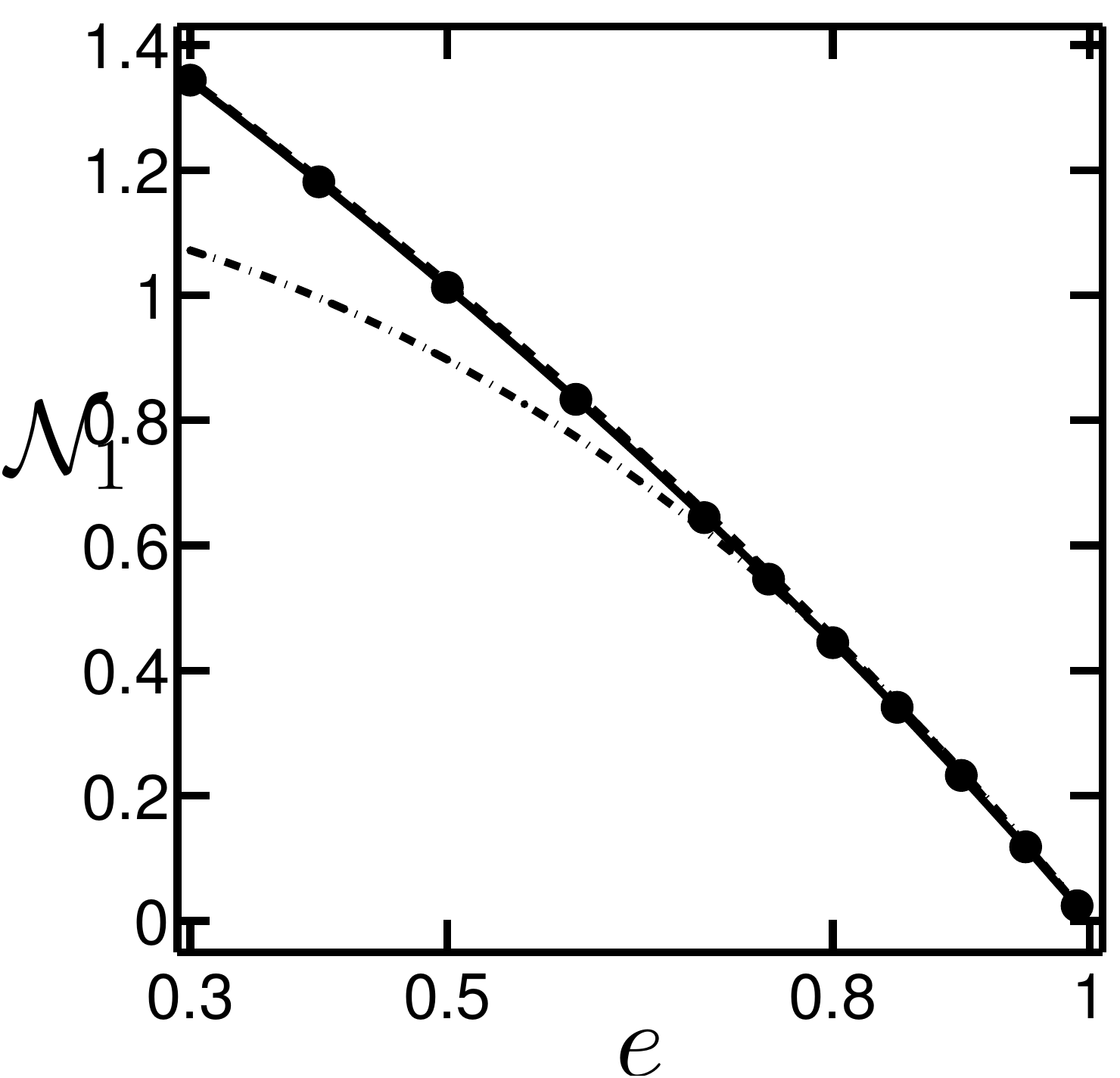}
 (b)
 \includegraphics[scale=0.38]{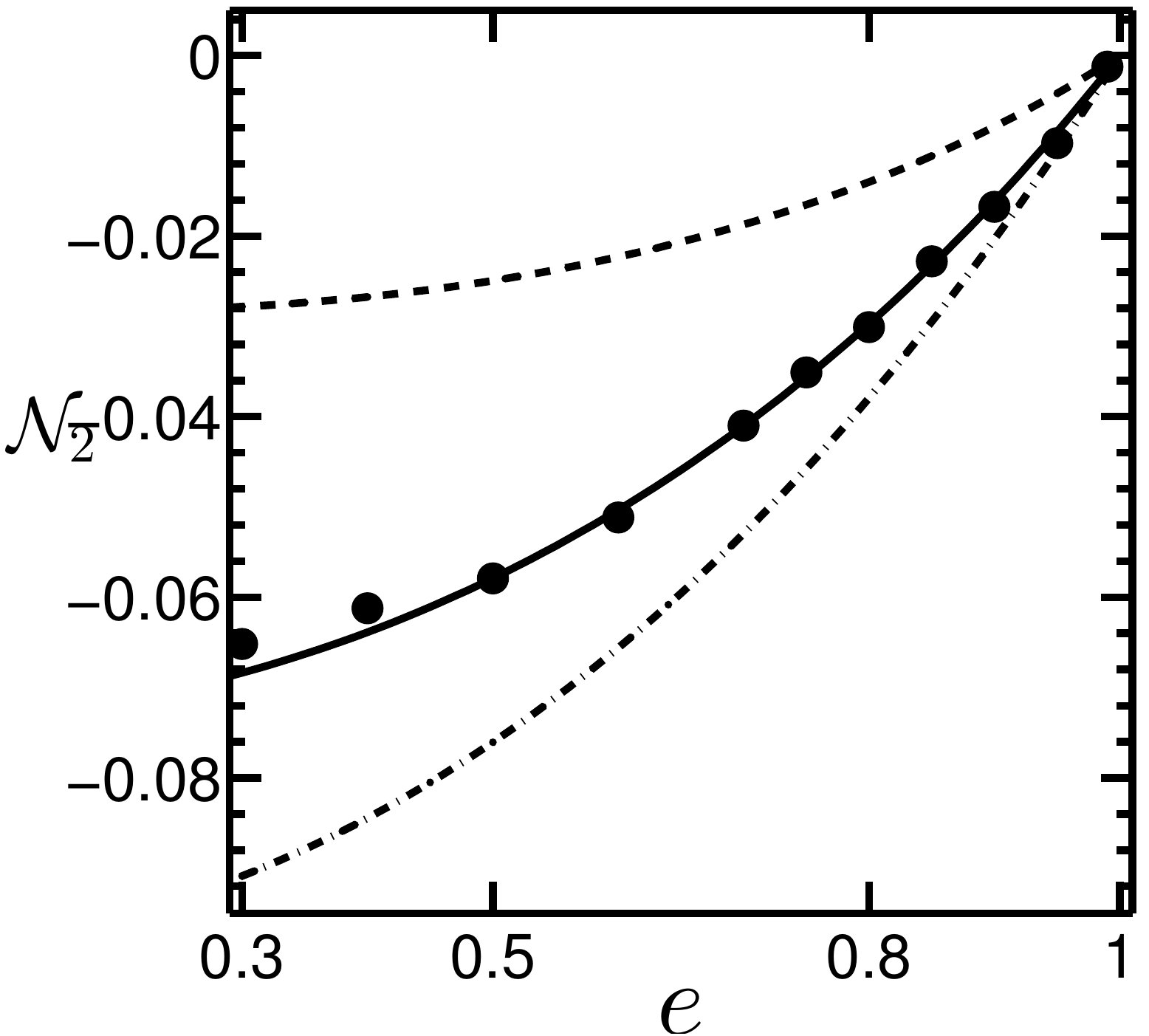}
\end{center}
 \caption{
 Comparisons of   (a) ${\mathcal N}_1$ and (b) ${\mathcal N}_2$ in uniform shear flow of a granular gas ($St=\infty$):
 (i) MD simulation (symbols), (ii) present  theory [solid lines,~\cite{SA2016}] and (iii) the standard Grad's moment theory (dashed line).
The dot-dash line in each panel represent the super-Burnett-order Chapman-Enskog solution   of \cite{SG1998}, see Appendix F.
Parameter values as in figure~\ref{fig:fig-Texcess-com}.
 }
\label{fig:fig-N12-com2}
\end{figure}

\section{Summary and Conclusion}

 The rheology of a  dilute gas-solid suspension, consisting of  inelastic spheres  suspended in a Newtonian fluid,
 undergoing simple shear flow is analysed, with  the effect of the gas-phase  being modelled via a Stokesian drag force.
 The pertinent  inelastic Boltzmann equation is solved using an anisotropic Gaussian as the single particle distribution function which is known to be appropriate for a sheared system.
 The resulting hydrodynamic model for the particle-phase consists of a 10-moment system $(\rho, {\boldsymbol u}, {\mathsfb M})$
 of density ($\rho$),  hydrodynamic velocity ($\bf u$) and the second-moment  (${\mathsfb M}=\langle {\boldsymbol C}{\boldsymbol C}\rangle$) of fluctuation/peculiar velocity.
 One focus of the present work has been  to analyse the anisotropy of ${\mathsfb M}$ in the  simple shear flow of a dilute gas-solid suspension
 and subsequently tie and explain the rheological quantities in terms of them.

 The seond-moment tensor has been characterized by three parameters: (i) the non-coaxiality angle ($\phi$, the angle between the principal eigen-direction
 of ${\mathsfb M}$ and the shear tensor ${\mathsfb D}$), (ii) the shear-plane temperature-anisotropy ($\eta$, the difference between the principal eigenvalues
 of ${\mathsfb M}$ on the shear plane, $\eta \propto T_x-T_y$, where $T_i$ is the granular temperature along $i$-th direction)  
 and (iii) the excess temperature ($\lambda^2\propto T -T_z$) along the vorticity direction;
 the first two [$\phi$ and $\eta$] are dubbed `shear-plane' anisotropies and the last-one ($\lambda^2$) is dubbed vorticity-plane anisotropy.
 The closed-form expressions for  three anisotropy parameters ($\phi$, $\eta$, $\lambda^2$) and  the granular temperature ($T$) have been obtained
  as functions of the Stokes number ($St$), the mean density ($\nu$)   and the restitution coefficient ($e$)  by solving the  second-moment balance equation;
  these are used to obtain analytical expressions for the particle-phase viscosity and two normal-stress differences.
Scaling relations have been obtained in the limits of small and large $St$ as well as small inelasticity $(1-e)$.

 Static multiple states of high and low temperatures are found when the Stokes number is small enough,
 thereby recovering the original ``ignited'' ($I$) and ``quenched'' ($Q$) states of \cite{TK1995} --
 the role of inelasticity on these states has been   examined.
  The high-temperature ignited state, in which the randomness of the particle motion is  high giving rise to a large value of granular temperature ($T$),
 exists above some minimum Stokes number $(St_{c_1})$ whose value increases with increasing  $e$.
  In contrast,  the low-temperature quenched state, in which most of the particles in the system follow the local fluid velocity,
 appears below a critical value of Stokes number $(St_{c_2})$  which is a decreasing function of both $e$ and  $\nu$.
 Both  these Stokes numbers ($St_{c_1}$ and $St_{c_2}$) have been determined analytically as  functions of $\nu$ and $e$, and
  the regions of co-existence of two states (quenched and ignited) along with the transition regimes have been 
  identified in a three-dimensional ($St, e, \nu$) phase diagram.

 The effect of inelasticity is found to reduce the particle-phase viscosity on both ignited and quenched states,
 with shear-thickening behaviour (increasing viscosity with increasing shear rate) being found in both states.
At any $e$, the shear-viscosity  undergoes a discontinuous jump with increasing $St$  at  ``$Q \rightarrow I$''  transition,
 which  can be interpreted as ``discontinuous shear thickening'' (DST).
 The two normal stress differences also undergo similar first-order jump-transitions: (i) ${\mathcal N}_1$ from large to small positive values 
 and (ii) ${\mathcal N}_2$ from positive to negative values.
 The sign-change of ${\mathcal N}_2$ (figure 10) has been identified with the system being making a ``$Q \leftrightarrow I$''  transition.
 The origin of this sign-change has been tied to  a competition between (i) the excess temperature ($T_z^{ex} \propto 3\lambda^2$)
 and (ii) the shear-plane anisotropies  ($\eta\sin{2\phi}\equiv {\mathcal N}_1/2$) of the second-moment tensor:
 while  the former dominates over the latter in the quenched state, the latter dominates  in the ignited state,
resulting in the sign-change of ${\mathcal N}_2$ at some finite value of $St$.
For both granular and gas-solid suspensions, the excess temperature along the vorticity direction is responsible
 for the origin of ${\mathcal N}_2\neq 0$, while  the temperature anisotropy $\eta$ and the non-coaxiality angle $\phi$ are responsible for ${\mathcal N}_1\neq 0$.
 

The comparative analyses in figures~\ref{fig:fig2}, \ref{fig:fig-N12-com}, \ref{fig:fig-N12-com1}, \ref{fig:fig-Texcess-com} and \ref{fig:fig-N12-com2} can be summarized as follows:
 the moment expansion about an anisotropic-Maxwellian (AME) yields  accurate transport coefficients (shear viscosity and normal stress differences)
 for dissipative particles ($e<1$) in both small and large Stokes number limits,
representative of gas-solid and dry granular suspensions, respectively.
The standard Grad's moment-expansion (GME) significantly under-predicts the value of the second normal stress difference ${\mathcal N}_2$,
although it is comparable with AME with respect to ${\mathcal N}_1$ up-to a  restitution coefficient of $e=0.5$.
On the other hand, the latter theory (GME) also over-predicts the shear viscosity ($\mu\propto \sqrt{T}$, viz.~figure~\ref{fig:fig2})
of small-$St$ suspensions even for moderately dissipative ($e=0.8$) particles;
the mismatch between GME and simulation increases with decreasing $e$.
Based on the present work we conclude  that the  superior predictive ability of the AME theory for hydrodynamics and 
rheology of `dry' ($St\to\infty$) sheared granular gases~\citep{SA2014,SA2016}
carries over to small-$St$ gas-solid suspensions of highly inelastic particles.

It would  be interesting to check the applicability of this theory to dense gas-solid suspensions of inelastic particles (with frictional interactions)
which can be taken up in future. 
The present work can  also be extended to include a `non-linear' drag law (dependence on particle Reynolds number) by modifying (\ref{eqn:stokesDrag1}) via 
well-known empirical correlations.
Lastly, the anisotropies ($\phi, \eta, \lambda^2$) of the second-moment tensor should be 
measured from simulations of finite-$St$ suspensions so that one-to-one comparisons with theory can be made in this regard.

\appendix

\section{Analysis in the ignited state for elastic hard-spheres}

For a gas-solid suspension of elastic hard-spheres ($e=1$), the collisional source of second-moment in the ignited state is given by
\begin{eqnarray}
    \aleph_{\alpha\beta} &=& \frac{-24\rho_p\nu^2}{\sigma\pi^\frac{3}{2}}\int(k_\alpha j_\beta+k_\beta j_\alpha)({\textbf{\emph k}}\cdot{{\mathsfb M}}\cdot{\textbf{\emph j}})({\textbf{\emph k}}\cdot{{\mathsfb M}}\cdot{\textbf{\emph k}})^\frac{1}{2} d{\boldsymbol k}
\nonumber\\
 & =&  -\frac{32\rho_p\nu^2 T^{3/2}}{35\sigma\sqrt{\pi}} \times \nonumber\\
&& \left[
 \begin{array}{ccc}
  \eta^2+21\lambda^2+21\eta\sin2\phi & -21\eta\cos2\phi & 0\\
  -21\eta\cos2\phi & \eta^2+21\lambda^2-21\eta\sin2\phi & 0\\
   0 & 0 & -2(\eta^2+21\lambda^2)
 \end{array}
\right] ,
\label{eqn:source_ignited_explicit-e1}
\end{eqnarray}
which is a function of $\nu$,$T$, $\eta$, $\phi$ and $\lambda^2$.

Four independent equations of second-moment balance,
\begin{equation}
    P_{\delta\beta}u_{\alpha,\delta}+P_{\delta\alpha}u_{\beta,\delta} +\frac{2\gamma}{St} P_{\alpha\beta} =\aleph_{\alpha\beta},
 \label{eqn:balance_second_moment=e1}
\end{equation}
can be rearranged  to yield a quartic-order equation,
\begin{equation}
 \omega^2\Big[12096 St^2\omega^2+\left(10260 St-420 St^3\right)\omega+3225-175 St^2\Big]=0,
 \label{eqn:cubicT}
\end{equation}
where $\omega$ is the rescaled temperature
\begin{equation}
    \omega = \frac{\nu}{\sqrt{\pi}} \frac{\sqrt{T}}{(\dot\gamma \sigma/2)}.
\end{equation}
In the following,  the temperature has been made dimensionless by dividing it by $(\dot\gamma \sigma/2)^2$. 
Three distinct solutions of (\ref{eqn:cubicT}) are
\begin{eqnarray}
   \label{dilute_T_3}
  \sqrt{T_{is}} &= & \frac{5\pi^\frac{1}{2}}{144}\frac{St}{\nu}\Omega(St),
\\
 \label{dilute_T_2}
 \sqrt{T_{us}} &=& \frac{5\pi^\frac{1}{2}}{144}\frac{St}{\nu}\Big[\frac{7 St^2-171-\sqrt{49 St^4-42 St^2-14103}}{14 St^2}\Big],\\
 \label{dilute_T_1}
 T_{qs} &=& 0,
 \end{eqnarray}
 with $T_{is}>T_{us}> T_{qs}$, where
 \begin{equation}
 \label{dilute_omega}
  \Omega(St)=\Big[\frac{7 St^2-171+\sqrt{49 St^4-42 St^2-14103}}{14 St^2}\Big] \equiv \frac{144}{5}\omega_{is} St^{-1}.
\end{equation}
In the above expressions, $T_{qs}$ corresponds to the quenched state temperature,
 $T_{us}$ corresponds to an unstable temperature and $T_{is}$ corresponds to the temperature in the ignited state.
 It is clear from (\ref{dilute_T_2}) that  a positive value for  $T_{us}$ requires the following condition on the Stokes number:
\begin{equation}
 \label{condition_st_c}
 7 St^2-171\geqslant 0, \qquad \Rightarrow \quad St\geqslant \sqrt{\frac{171}{7}} \approx 4.9425 \equiv St_{c_1}.
\end{equation}
Therefore, $St$ must be greater than or equal to $St_{c_1}$, and (\ref{condition_st_c}) provides a lower bound on $St$
for the existence of the ignited state in a dilute sheared gas-solid suspension.

The remaining   equations of (\ref{eqn:balance_second_moment=e1}) can be solved to yield
solutions for $\eta^2$ and $\lambda^2$ in the ignited state:
 \begin{equation}
 \left.
 \begin{array}{lcl}
 \eta^2 &=& \frac{(9 + \Omega St^2)}{4\left(1 + \frac{29}{84}\Omega St^2 + \frac{1}{36} \Omega^2 St^4\right)}\\
 \lambda^2 &=& \frac{(7 + \Omega St^2)}{14\left(1 + \frac{29}{84}\Omega St^2 + \frac{1}{36} \Omega^2 St^4\right)}
  \end{array}
 \right\} ;
 \end{equation}
  the solution for the non-coaxilality angle is
 \begin{equation}
   \sin(2\phi) = \frac{\eta}{1+\lambda^2}.
 \end{equation}
 Therefore, the normal stress differences in the ignited state are given by
\begin{eqnarray}
 \mathcal{N}_1& =& \frac{15}{5+24 St\omega} \equiv \frac{18}{6 + \Omega St^2}, \\
 \label{nsd2_elastic_dilute}
 -\mathcal{N}_2&=& \frac{270 St\omega(5+16 St\omega)}{(5+24 St\omega)(175+1740 St\omega+4032 St^2\omega^2)} .
\end{eqnarray}

In the ignited state, the expression for the  shear viscosity of the particle phase is
 \begin{equation}
 \label{viscosity_elastic_dilute}
\mu=-P_{xy}/\dot\gamma =\mu_N\Omega(St),
\end{equation}
where
\begin{equation}
     \mu_N =\frac{5\sqrt\pi}{96}\rho_p\sigma \sqrt{T}
\end{equation}
is the Newtonian viscosity of a dilute gas. 
Therefore, $\Omega(St)$ [(\ref{dilute_omega})]  is a measure of the deviation of particle-phase viscosity from the Newtonian viscosity of a dilute hard-sphere gas.

\section{Coefficients $a_i$}

Explicit expressions of the individual coefficients $a_i$ appearing in  (\ref{eqn:energy_balance}) are given by:
\begin{eqnarray}
 a_{10} &=& 86416243200 (3-e)^4 (1-e)^3 (1+e)^7\pi St^6\nu^7,\\
 a_9 &=& 28805414400 (3-e)^3 (1-e)^2 (1+e)^6 (19-13 e)\pi^{(3/2)} St^5\nu^6,\\
 a_8 &=& 28576800 (3-e)^2 (1-e) (1+e)^5\pi^2 St^4\nu^5 \Big(252 (197-278 e+93 e^2)\nonumber\\
         &&\qquad\qquad\qquad\qquad\qquad\qquad+5 (1747-1438 e+363 e^2) St^2\Big),\\
 a_7 &=& 3810240 (3-e) (1+e)^4 \sqrt{\pi} St^3\nu^4 \Big(2100 (1-e) (241-284 e+79 e^2)\pi^2\nonumber\\
         &&\qquad\qquad\qquad\qquad+25 (12607-19952 e+10099 e^2-1746 e^3)\pi^2 St^2\nonumber\\
         &&\qquad\qquad\qquad\qquad\qquad\qquad-3456 (3-e)^3 (1-e)^2 (1+e)^4 St^3\nu^3\Big),\\
 a_6 &=& 79380 (1+e)^3\pi St^2\nu^3 \Big(21000 (1-e) (871-854 e+199 e^2)\pi^2\nonumber\\
         &&\qquad\qquad\qquad\qquad+500 (56617-78677 e+35629 e^2-5361 e^3)\pi^2 St^2\nonumber\\
         &&\qquad\qquad\qquad\qquad-125 (1691+539 e-1223 e^2+337 e^3)\pi^2 St^4\nonumber\\
         &&\qquad\qquad\qquad\qquad-27648 (3-e)^3 (1-e) (1+e)^4 (29-23 e) St^3\nu^3\Big) ,\\
 a_5 &=& 18900 (1+e)^2\pi^{(3/2)} St\nu^2 \Big(441000 (1-e) (23-11 e)\pi^2\nonumber\\
         &&+10500 (3437-3093 e+688 e^2)\pi^2 St^2-875 (477+442 e-247 e^2)\pi^2 St^4\nonumber\\
         &&\qquad\qquad-580608 (3-e)^2 (1-e) (1+e)^4 (11-7 e) St^3\nu^3\nonumber\\
         &&\qquad\qquad-1152 (3-e)^2 (1+e)^4 (991-934 e+279 e^2) St^5\nu^3\Big) ,\\
 a_4 &=& 63 (1+e)\nu \Big(165375000 (1-e)\pi^4+656250 (2437-1069 e)\pi^4 St^2\nonumber\\
         &&-109375 (107+193 e)\pi^4 St^4-48384000 (3-e) (1-e) (1+e)^4 (37-19 e)\pi^2 St^3\nu^3\nonumber\\
         &&-288000 (3-e) (1+e)^4 (3917-3368 e+843 e^2)\pi^2 St^5\nu^3\nonumber\\
         &&-3024000 (3-e)^3 (1+e)^4\pi^3 St^6\nu^3+7962624 (3-e)^4 (1-e) (1+e)^8 St^6\nu^6\Big)  ,\\
 a_3 &=& 2520 \sqrt{\pi} St \Big(2296875\pi^4-504000 (1-e) (1+e)^4 (41-17 e)\pi^2 St\nu^3\nonumber\\
         &&-6000 (1+e)^4 (5617-4438 e+933 e^2)\pi^2 St^3\nu^3-189000 (3-e)^2 (1+e)^4\pi^3 St^4\nu^3\nonumber\\
         &&-1000 (1+e)^4 (1203-1002 e+247 e^2)\pi^2 St^5\nu^3\nonumber\\
         &&+663552 (3-e)^3 (1-e) (1+e)^8 St^4\nu^6\Big)   ,\\
 a_2 &=& -2400 (1+e)^3\pi St\nu^2 \Big(1323000 (1-e)\pi^2+15750 (383-151 e)\pi^2 St^2\nonumber\\
         &&+165375 (3-e)\pi^3 St^3+875 (789-305 e)\pi^2 St^4\nonumber\\
         &&-870912 (3-e)^2 (1-e) (1+e)^4 St^3\nu^3-1728 (3-e)^2 (1+e)^4 (47-39 e) St^5\nu^3\Big),\\
 a_1 &=& -2000 (1+e)^2\pi^{(3/2)} St^2\nu \Big(441000\pi^2+55125\pi^3 St+98000\pi^2 St^2\nonumber\\
         &&-580608 (3-e) (1-e) (1+e)^4 St\nu^3-3456 (3-e) (1+e)^4 (47-39 e) St^3\nu^3\Big),\\
 a_0 &=& 1440000 (1+e)^5\pi^2 St^2 (4+St^2)\nu^3 \Big(42 (1-e)+(13-9 e) St^2\Big).
 \end{eqnarray}

\section{Ordering analysis to determine three temperatures}
\label{sec:three_temperatures}

We will solve (\ref{eqn:energy_balance}) analytically in the asymptotic limit $\nu\ll1$, $St\gg1$, and $St^3\nu\ll1$~\citep{TK1995},
and  three feasible solutions have been found as described below.

\subsection{Temperature in the quenched state}

For $\xi\sim O(St^{3/2}\sqrt{\nu})$, the leading order term in (\ref{eqn:energy_balance}) is 
 $O(St^\frac{11}{2}\nu^\frac{3}{2})$ and consequently we have
\begin{equation}
 a_3\xi^3+a_1\xi=0,
\end{equation}
where
\begin{equation}
 a_3=5788125000 \upi^\frac{9}{2}St,\qquad a_1=-196000000 \upi^\frac{7}{2} (1+e)^2 St^4 \nu.
\end{equation}
The solution at this level of approximation is
\begin{equation}
 T_{qs}=\xi^2=\frac{32(1+e)^2}{945\upi}St^3\nu,
 \label{eqn:temperature_quenchedA}
\end{equation}
which corresponds to the temperature in the quenched state.
Note that the quenched temperature increases with increasing both $St$ and $\nu$.

\subsection{Unstable temperature}

When $\xi\sim O(St^3\nu)^{-1}$, the highest-order term in (\ref{eqn:energy_balance}) is $O(1/St^8\nu^3)$, and
on neglecting terms smaller than this, we have at leading order
\begin{equation}
 a_4\xi^4+a_3\xi^3=0,
\end{equation}
where
\begin{equation}
 a_4=-6890625 (1+e)(107+193e) \upi^4 St^4 \nu,\qquad a_3=5788125000 \upi^\frac{9}{2} St.
\end{equation}
Therefore, we have
\begin{equation}
   \sqrt{T_{us}}=\xi=\frac{840\sqrt{\upi}}{(1+e)(107+193e)}\left(\frac{1}{St^3\nu}\right),
 \label{eqn:temperature_unstableA}
\end{equation}
This is the temperature of an intermediate state which is unstable -- note that $T_{us}$ decreases with  increasing $St$ and $\nu$.

\subsection{Temperature in the ignited state}

In the asymptotic limit of $\xi\sim O(St/\nu)$, the leading order term of $a_i\xi^i\;i=0(1)11$ is $O(St^{12}/\nu^3)$ and consequently we have from (\ref{eqn:energy_balance})
\begin{equation}
 a_7\xi^7+a_6\xi^6=0,
\end{equation}
where
\begin{equation}
\left.
\begin{array}{lcl}
 a_7&=&95256000(3-e)(1+e)^4(12607-19952e+10099e^2-1746e^3)\upi^\frac{5}{2}St^5\nu^4,\\
 a_6&=&-9922500(1+e)^3(1691+539e-1223e^2+337e^3)\upi^3 St^6 \nu^3.
 \end{array}
 \right\}.
\end{equation}
Therefore, the temperature at this order of approximation is
\begin{equation}
 \sqrt{T_{is}}=\xi =\frac{5(1691+539e-1223e^2+337e^3)\sqrt{\upi}}{48(3-e)(1+e)(12607-19952e+10099e^2-1746e^3)}\left(\frac{St}{\nu}\right),
 \label{eqn:temperature_ignitedA}
\end{equation}
which corresponds to the temperature in the ignited state.
While $T_{is}$ increases with increasing $St$, it deceases with increasing the particle volume fraction $\nu$.

\section{Analytical determination of limit-points  $St_{c_1}$ and $St_{c_2}$}
\label{sec:critical_points}

At the critical/limit points, two solution branches of  (\ref{eqn:energy_balance}) corresponding to two different states 
[(i) quenched $(T_{qs})$ and unstable $(T_{us})$ states and (ii) unstable $(T_{us})$ and ignited $(T_{is})$ states]
meet and consequently we have  saddle-node bifurcations  from one stable state to another. Therefore, these limit points correspond 
to the double roots of (\ref{eqn:energy_balance}) at which  the following conditions must be satisfied:
\begin{equation}
 \mathcal{G}(\xi_c)=0
 \quad
 \mbox{and}
 \quad
 \mathcal{G}'(\xi_c)=0.
\end{equation}

\subsection{Determining $St_{c_1}$: discontinuous transition from ``ignited'' to ``quenched'' states}

The critical Stokes number, $St_{c_1}$, for the transition from the ignited to quenched states corresponds to the limit point at which
 the temperatures corresponding to the ignited ($T_{is}$) and unstable ($T_{us}$) branches overlap with each other.
Considering  $\xi\sim O(\nu St)^{-1} \gg 1$, and retaining the highest-order terms, (\ref{eqn:energy_balance}) reduces to
\begin{eqnarray}
 \mathcal{G} &\approx& a_{7}\xi^{7}+a_{6}\xi^{6}+a_{5}\xi^{5}+a_{4}\xi^{4}+a_{3}\xi^{3}=0 = a_{7}\xi^{4}+a_{6}\xi^{3}+a_{5}\xi^{2}+a_{4}\xi+a_{3},
             \label{eqn:4thdegree}\\
{\rm and}    && 4a_{7}\xi^{3}+3a_{6}\xi^{2}+2a_{5}\xi+a_{4}=0,         
\end{eqnarray}
where 
\begin{equation}
\left.
\begin{array}{rcl}
 a_7 &=& 95256000 (3-e) (1+e)^4 (12607-19952e+10099e^2-1746e^3)\pi^\frac{5}{2} St^5\nu^4,\\
 a_6 &=& 9922500 (1+e)^3\pi^3 St^4 \Big(4 (56617-78677e+35629e^2-5361e^3) \\
      &&  \qquad\qquad\qquad\qquad\qquad\qquad-(1691+539e-1223e^2+337e^3) St^2\Big)\nu^3,\\
 a_5 &=& 16537500 (1+e)^2\pi^\frac{7}{2} St^3 \Big(12(3437-3093e+688e^2)\\
      &&  \qquad\qquad\qquad\qquad\qquad\qquad-(477+442e-247e^2)St^2\Big)\nu^2,\\
 a_4 &=& 6890625 (1+e)\pi^4 St^2 (6(2437-1069e)-(107+193e)St^2)\nu,\\
 a_3 &=& 5788125000\pi^\frac{9}{2} St.
 \end{array}
\right\}
\end{equation}
Using the condition of equal roots of a fourth-degree polynomial (\ref{eqn:4thdegree}), we obtain an expression for the critical Stokes number
for the ``ignited-to-unstable'' transition:
\begin{equation}
 St_{c_1}\approx9.9-4.91e.
 \label{eqn:critical_stokes_1A}
\end{equation}
While decreasing the Stokes number along the ignited-state branch (see figure~\ref{fig:fig2}), the system jumps from the
ignited to the quenched state at $St<St_{c_1}$ for all $\nu < \nu^l_{us}$ (3.8). Therefore, (\ref{eqn:critical_stokes_1A}) represents the minimum/critical Stokes number
below which (\ref{eqn:energy_balance}) admits the unique ``quenched'' state solution.



\subsection{Determining $St_{c_2}$: discontinuous transition from ``quenched'' to ``ignited'' state}

The limit point corresponding to the overlap of the quenched and unstable branches of the system is denoted by the Stokes number $St_{c_2}$ at which 
 the  temperatures associated with the quenched ($T_{qs}$) and unstable ($T_{us}$) states coincide --
above this critical value of Stokes number the quenched state ceases to exist.
Mathematically, $St_{c_2}$ is the point of the double root $T_{is}=T_{us}$ of   (\ref{eqn:energy_balance}).
above which there exists only one feasible solution $T_{is}$ (corresponding to the ignited state) and the system jumps from the quenched state into the ignited state
At this order of approximation $\xi\sim O(1)$ and the highest order terms are of the orders of $\nu St^4$ and $St$. 
Therefore on neglecting the terms of $O(St^4\nu^2)$ and using the statement of $T_{is}=T_{us}$, we have from (\ref{eqn:energy_balance})
\begin{eqnarray}
 \mathcal{G}(\xi_c) &\approx& a_4\xi^4+a_3\xi^3+a_1\xi=0 = a_4\xi^3+a_3\xi^2+a_1,
         \label{eqn:1_St_c}
         \\
{\rm and}\quad
 \mathcal{G}'(\xi_c) &\approx& 3a_{4}\xi^2+2a_{3}\xi=0,
          \label{eqn:2_St_c}
\end{eqnarray}
where
\begin{equation}
\left.
\begin{array}{rcl}
 a_4 &=& -6890625 (1+e) (107+193e)\pi^4 St^4\nu,\\
 a_3 &=& 5788125000\pi^\frac{9}{2} St,\\
 a_1 &=& - 196000000 (1+e)^2\pi^\frac{7}{2} St^4\nu.
 \end{array}
\right\}
\end{equation}
It  follows from (\ref{eqn:2_St_c}) that
\begin{equation}
 \xi_c=\frac{-2a_3}{3a_4}=\frac{560\sqrt{\pi}}{(1+e)(170+193e)St^3\nu}.
 \label{eqn:3_St_c}
\end{equation}
On substituting (\ref{eqn:3_St_c}) into (\ref{eqn:1_St_c}) we obtain the {\it critical-surface}
\begin{equation}
 St_{c_2}^3\nu_c=\Bigg(\frac{3087000 \pi^2}{(1+e)^4 (107+193e)^2}\Bigg)^\frac{1}{3},
 \label{eqn:final_St_cA}
\end{equation}
above which only the ignited state exists.

\section{Grad's moment expansion (GME) for inelastic gas-solid suspension}

The standard Grad's moment expansion (GME) in terms of a truncated Hermite series around the Maxwellian~\citep{Grad1949}
has been employed by many researchers  \citep{HH1982,TK1995,CRG2015} to analyse the Boltzmann equation for  a ``sheared'' hard-sphere gas 
as well as  gas-solid suspensions.
\begin{itemize}
\item
\cite{HH1982}$\Rightarrow$ $e=1,\;St=\infty$ \quad (Dilute gas of elastic hard-spheres)
\item
\cite{TK1995}$\Rightarrow$ $e=1,\;St$ finite  \quad (Suspension of elastic hard-spheres)
\item
\cite{CRG2015}$\Rightarrow$ $e\neq1,\;St$ finite \quad (Suspension of inelastic hard-spheres)
\end{itemize}

For the case of a dilute gas-solid suspension of ``inelastic'' hard-spheres,
the  collisional production term of the second moment has been evaluated as:
\begin{eqnarray}
 \aleph_{\alpha\beta}&=&-\frac{8\rho_p\nu^2(1-e^2)T^\frac{3}{2}}{\sqrt{\pi}\sigma}\delta_{\alpha\beta}-\frac{24\nu(1+e)(3-e)T^\frac{1}{2}}{5\sqrt{\pi}\sigma}P_{\langle\alpha\beta\rangle}\nonumber\\
 &&+\frac{(1+e)}{35\sqrt{\pi}\sigma \rho_p T^\frac{1}{2}}\underline{\Big\{(5+3e)P_{\langle kl\rangle}P_{\langle kl\rangle}\delta_{\alpha\beta}+12(e-3)P_{\langle \alpha l\rangle}P_{\langle l \beta\rangle}\Big\}},
 \label{eqn:nonlinGrad1}
\end{eqnarray}
 where the underlined terms represent the quadratic nonlinearity in the pressure deviator $P_{\langle \alpha\beta\rangle}=P_{\alpha\beta} - p\delta_{\alpha\beta}$,
with $p=P_{\alpha\alpha}/3$; $\rho_p=m/(\pi\sigma^3/6)$ is the intrinsic/material density of particles, $\nu$ is the particle volume fraction and $e$ is the restitution coefficient.
 In fact,  the second normal-stress difference is zero (${\mathcal N}_2=0$) in the absence of the underlined non-linear terms in (\ref{eqn:nonlinGrad1}),
 see the proof at the end of this appendix.

Defining the non-dimensional quantities as
\begin{eqnarray}
 P^*=\frac{P}{\rho_p\nu (\dot\gamma \sigma/2)^2}, \qquad T^*=\frac{T}{(\dot\gamma \sigma/2)^2},
 \qquad \aleph^* = \frac{\aleph}{\rho_p\nu {\dot\gamma}^3 (\sigma/2)^2}, 
\end{eqnarray}
and  on omitting the $*$ signs, for convenience, 
the dimensionless second-moment balance for steady homogeneous shear flow,
\begin{eqnarray}
  P_{\delta\beta}u_{\alpha,\delta}+ P_{\delta\alpha}u_{\beta,\delta}+\frac{2}{St}P_{\alpha\beta}=\aleph_{\alpha\beta},
\end{eqnarray}
can be written in component form as follows:
\begin{eqnarray}
 &&(1+e)(5+3e)\Big(P_{\langle xx\rangle}^{2}+P_{\langle yy\rangle}^{2}
 +P_{\langle zz\rangle}^{2}+2P_{ xy}^{2}\Big)
 -12(1+e)(3-e)\Big(P_{\langle xx\rangle}^{2}+P_{ xy}^{2}\Big) \nonumber\\
 && \qquad \qquad -280(1-e^2)T^2-168(1+e)(3-e)T P_{\langle xx\rangle}
 -\frac{140\sqrt{\pi}\sqrt{T}P_{ xy}}{\nu} \nonumber\\
 &&\hspace{7cm}-\frac{140\sqrt{\pi}\sqrt{T}}{St\nu}(T+P_{\langle xx\rangle})=0,\\
 &&(1+e)(5+3e)\Big(P_{\langle xx\rangle}^{2}+P_{\langle yy\rangle}^{2}
 +P_{\langle zz\rangle}^{2}+2P_{ xy}^{2}\Big)
 -12(1+e)(3-e)\Big(P_{\langle yy\rangle}^{2}+P_{ xy}^{2}\Big) \nonumber\\
 && \qquad -280(1-e^2)T^{2}-168(1+e)(3-e)T P_{\langle yy\rangle}-\frac{140\sqrt{\pi}\sqrt{T}}{St\nu}(T+P_{\langle yy\rangle})
  =0,\\
 &&(1+e)(5+3e)\Big(P_{\langle xx\rangle}^2+P_{\langle yy\rangle}^2
 +P_{\langle zz\rangle}^2+2P_{ xy}^2\Big)  -12(1+e)(3-e)P_{\langle zz\rangle}^2\nonumber\\
 &&\qquad -280(1-e^2)T^2-168(1+e)(3-e)T P_{\langle zz\rangle}-\frac{140\sqrt{\pi}\sqrt{T}}{St\nu}(T+P_{\langle zz\rangle})
  =0,\\
 &&12(1+e)(3-e)P_{ xy} P_{\langle zz\rangle}
 -168(1+e)(3-e) T P_{ xy}
 -\frac{70\sqrt{\pi}\sqrt{T}}{\nu}(T+P_{\langle yy\rangle})\nonumber\\
 &&\hspace{8cm}-\frac{140\sqrt{\pi}T^\frac{1}{2}}{St\nu}P_{ xy} =0,
\end{eqnarray} 
along with  constraint $\widehat{P}_{\alpha\alpha}=0$.
These equations have been solved numerically for specified values of $e$, $St$ and $\nu$ to yield $T$, $P_{\langle \alpha\alpha\rangle}$ and $P_{xy}$;
two normal stress differences
 ${\mathcal N}_1$ and  ${\mathcal N}_2$ can be expressed in terms of $P_{\langle \alpha\alpha\rangle}$.  These are dubbed ``GME'' solutions
and their comparisons with the present theory (\S4) based on anisotropic-Maxwellian expansion (AME) are shown in figures~
\ref{fig:fig-N12-com}, \ref{fig:fig-N12-com1} and \ref{fig:fig-N12-com2}, as discussed in  \S5.1 and \S5.2.

\begin{theorem}
The source term is uniquely decomposed as  $\aleph_{\alpha\beta} = \left(\frac{1}{3}\aleph_{\gamma\gamma}\right) \delta_{\alpha\beta} + \aleph_{\langle\alpha\beta\rangle}$.
If $\aleph_{\langle\alpha\beta\rangle} =  B P_{\langle\alpha\beta\rangle}$,
then ${\mathcal N}_2=0$.
\end{theorem}

\begin{proof}
For the case of homogeneous shear $u_x=\dot{\gamma}y$, $u_y=0$, $u_z=0$; the balance of second moment for a granular gas is
\begin{equation}
 {\rm P}_{\delta\beta}{\rm u}_{\alpha,\delta}+{\rm P}_{\delta\alpha}
                               {\rm u}_{\beta,\delta} = \aleph_{\alpha\beta}.
\end{equation}
Now, upon substituting $\alpha=2$, $\beta=2$ and $\alpha=3$, $\beta=3$ we have
\begin{equation}
  \aleph_{22} =  0= \aleph_{33}.
\end{equation}
From ${P}_{ij}=p\delta_{ij}+{P}_{\langle ij\rangle}$, we can write
\begin{equation}
  {\mathcal{N}}_2 = \left({P}_{\langle 22\rangle} - {P}_{\langle 33\rangle}\right) =B^{-1}\left({\aleph}_{\langle 22\rangle} - {\aleph}_{\langle 33\rangle}\right) = 0.
\label{eqn:N2=0}
\end{equation}
\end{proof}

Of course, (\ref{eqn:N2=0}) is in  contradiction with (i) the nonlinear expression  (\ref{eqn:nonlinGrad1})  obtained from the standard Grad-moment expansion 
as well as with (ii)  our choice of anisotropic Maxwellian distribution function, both  yielding ${\mathcal N}_2\neq 0$.


\end{document}